\documentclass[a4paper,onecolumn,accepted=2026-07-12]{quantumarticle}
\pdfoutput=1

\usepackage{graphicx} % Required for inserting images

\usepackage{pdfpages}

\usepackage{marginfix}
\usepackage{multicol}
\usepackage[textheight=8.8in,textwidth=6in,top=1in,headheight=12pt,headsep=25pt,footskip=30pt]{geometry}
\usepackage{newpxtext}
\usepackage[utf8]{inputenc} % allow utf-8 input
\usepackage[T1]{fontenc}    % use 8-bit T1 fonts
\usepackage{hyperref}       % hyperlinks
\usepackage{url}            % simple URL typesetting
\usepackage{booktabs}       % professional-quality tables
\usepackage{amsfonts}       % blackboard math symbols
\usepackage{nicefrac}       % compact symbols for 1/2, etc.
\usepackage{microtype}      % microtypography
\usepackage{xcolor}         % colors
\usepackage{stmaryrd}       % funky symbols
\usepackage{multirow}       % Merged cells in tables
\usepackage{float}
\usepackage{subfig}
\usepackage[font=small]{caption}
\usepackage{mathtools}
\usepackage{amssymb}
\usepackage{amsthm}
\usepackage{thmtools}
\usepackage{wrapfig}
\usepackage{braket}
\usepackage{csquotes}
\usepackage{xspace}

\usepackage{tikz-cd}
\usepackage{tikzfig}
\usepackage{pgfplots}
\usepackage[shortlabels]{enumitem}
\usepackage{algpseudocode}
\usepackage{algorithm}
\usepackage{eqparbox}

\usepackage[style=apa, maxbibnames=99, sortcites]{biblatex}
\DeclareFieldFormat{doi}{\mkbibacro{DOI}\addcolon\space\href{https://doi.org/\detokenize{#1}}{\nolinkurl{https://doi.org/#1}}}
\DeclareFieldFormat{url}{\url{#1}}
\AtBeginBibliography{\setlength{\emergencystretch}{3em}\sloppy}
\usepackage{todonotes}

\tikzstyle{Z dot}=[inner sep=0mm, minimum size=2mm, shape=circle, draw=black, fill=zxGreen, tikzit fill={rgb,255: red,216; green,248; blue,216}, outer sep=-0.5mm]
\tikzstyle{X dot}=[shape=circle, draw=black, fill=zxRed, tikzit fill={rgb,255: red,221; green,165; blue,165}, inner sep=0 mm, minimum size=2 mm, outer sep=-0.5mm]
\tikzstyle{Z phase dot}=[draw=black, fill=zxGreen, shape=rectangle, minimum size=4.5mm, rounded corners=1.8mm, inner sep=0.5mm, outer sep=-0.5mm, scale=0.8, tikzit shape=circle, tikzit fill={rgb,255: red,216; green,248; blue,216}, font={\footnotesize\boldmath}]
\tikzstyle{X phase dot}=[Z phase dot, draw=black, fill=zxRed, tikzit fill={rgb,255: red,221; green,165; blue,165}]
\tikzstyle{H box}=[fill=zxHad, draw=black, shape=rectangle, inner sep=0.6mm, minimum height=1.5mm, minimum width=1.5mm, tikzit fill=yellow, font={\footnotesize\boldmath}]
\tikzstyle{box}=[draw=black, shape=rectangle, fill=white, minimum size=1em, inner sep=0.2em, scale=0.85, font={\scriptsize}, outer sep=-0.5mm]
\tikzstyle{black dot}=[fill=black, draw=black, shape=circle, inner sep=1pt]
\tikzstyle{sLabel}=[font={\scriptsize}, tikzit draw=black, auto]
\tikzstyle{fireZ}=[inner sep=0mm, minimum size=2mm, shape=circle, draw=zxFireColour, fill=zxGreen, tikzit fill={rgb,255: red,216; green,248; blue,216}, outer sep=-0.5mm, line width=0.3mm]
\tikzstyle{fireX}=[inner sep=0mm, minimum size=2mm, shape=circle, draw=zxFireColour, fill=zxRed, tikzit fill={rgb,255: red,221; green,165; blue,165}, outer sep=-0.5mm, line width=0.3mm]
\tikzstyle{scalar}=[rounded rectangle, rounded rectangle arc length=120, fill=gray, inner sep=2pt, font={\tiny\boldmath}, label distance=1mm, fill opacity=.25, text opacity=1]

\tikzstyle{dotted}=[-, style=dashed, draw={rgb,255: red,128; green,128; blue,128}]
\tikzstyle{hadamard}=[-, style=dashed, draw=blue]
\tikzstyle{X Web}=[-, preaction={line width=1mm, draw=zxDarkRed}, tikzit draw=red]
\tikzstyle{Z Web}=[-, preaction={line width=1.7mm, draw=zxDarkGreen}, tikzit draw=green]
\tikzstyle{XZ Web}=[-, preaction={line width=1.7mm, draw=zxDarkGreen}, preaction={line width=1mm, draw=zxDarkRed}, tikzit draw=blue]
\tikzstyle{braceedge}=[-, decorate, decoration={brace, amplitude=2mm, raise=-1mm}]
\tikzstyle{fault-free}=[-, draw={rgb,255: red,177; green,98; blue,255}, line width=1pt]
\tikzstyle{new edge style 1}=[->, draw={rgb,255: red,177; green,98; blue,255}, line width=1pt]
\tikzstyle{arrow}=[->]
\tikzstyle{logical}=[-, draw=blue]
\tikzstyle{thick}=[-, line width=0.3mm]

\input{floquet.tikzdefs}
\theoremstyle{definition}
\newtheorem{theorem}{Theorem}[section]
\newtheorem*{theorem*}{Theorem}
\newtheorem{corollary}[theorem]{Corollary}
 
\newtheorem{proposition}[theorem]{Proposition}

\newtheorem{definition}[theorem]{Definition}
 
\newtheorem{remark}[theorem]{Remark}
\newtheorem{example}[theorem]{Example}

\newtheorem{example*}[theorem]{Example*}
\newtheorem{examples*}[theorem]{Examples*}

\def\bR{\begin{color}{red}}  
\def\bB{\begin{color}{blue}}
\def\bG{\begin{color}{green}}
\def\bP{\begin{color}{purple}}
\def\bC{\begin{color}{cyan}}
\def\beginChange{\begin{color}{blue}}
\def\endChange{\end{color}}
\def\e{\end{color}}

\newcommand{\interp}[1]{\left\llbracket#1\right\rrbracket}
\newcommand{\code}[1]{\left\llbracket#1\right\rrbracket}

\definecolor{cbpink}{RGB}{214,130,211}
\definecolor{cbyellow}{RGB}{241,131,108}
\definecolor{cbblue}{RGB}{110,148,189}
\definecolor{cbred}{RGB}{162,4,162}
\definecolor{cborange}{RGB}{251,145,10}
\definecolor{cbgreen}{RGB}{5,162,162}
\definecolor{cbcyan}{RGB}{204,234,207}

\newcommand{\autorf}[1]{%
  \begingroup%
  \renewcommand\equationautorefname{Eq.}
  \renewcommand\theoremautorefname{Thm.}
  \renewcommand\propositionautorefname{Prop.}
  \autoref{#1}%
  \endgroup%
}

\newcommand{\TextFusion}{\textsc{Fusion}\xspace}
\newcommand{\TextColour}{\textsc{Colour}\xspace}

\newcommand{\TextZElim}{\textsc{Z-Elim}\xspace}
\newcommand{\TextXElim}{\textsc{X-Elim}\xspace}
\newcommand{\TextPiCommute}{\textsc{$\pi$-Commute}\xspace}

\title{Floquetifying stabiliser codes with distance-preserving rewrites}
\author{Benjamin Rodatz, Boldizsár Poór, Aleks Kissinger}
\affiliation{University of Oxford, Oxford, UK}
\begin{document}
\allowdisplaybreaks
%\maketitle
%% \input{sections/extended-abstract.tex}
%% \clearpage
%\newcommand{\maincontenttitle}{}
%\setcounter{page}{1}
\maketitle
\begin{abstract}
    Stabiliser codes with large weight measurements can be challenging to implement fault-tolerantly.
    To overcome this, we propose a Floquetification procedure which, given a stabiliser code, synthesises a novel Floquet code that only uses single- and two-qubit operations.
    Moreover, this procedure preserves the distance and number of logicals of the original code.
    The new Floquet code requires additional physical qubits.
    The overhead is linear in the weight of the largest measurement of the original code.

    Our method is based on the ZX calculus, a graphical language for representing and rewriting quantum circuits.
    However, a problem arises with the use of ZX in the context of rewriting error-correcting codes: ZX rewrites generally do not preserve code distance.
    Tackling this issue, we define the notion of \emph{distance-preserving rewrite} that enables the transformation of error-correcting codes without changing their distance.

    These distance-preserving rewrites are used to decompose arbitrary weight stabiliser measurements into quantum circuits with single- and two-qubit operations.
    As we only use distance-preserving rewrites, we are guaranteed that a single error in the resulting circuit creates at most a single error on the data qubits.
    These decompositions enable us to generalise the Floquetification procedure of~\textcite{townsend-teagueFloquetifyingColourCode2023} to arbitrary stabiliser codes, provably preserving the distance and number of logicals of the original code.
\end{abstract}
%\listoftodos

\section{Introduction}
One major challenge of quantum computing is dealing with noise~\parencite{gottesmanStabilizerCodes1997}.
The goal of quantum error correction is to encode data with sufficient redundancy to detect and correct errors.
The most common and well-studied class of error-correcting codes are stabiliser codes~\parencite{gottesmanStabilizerCodes1997}.
These include CSS codes~\parencite{calderbankGoodQuantumErrorcorrecting1996, steaneMultipleparticleInterferenceQuantum1997}, surface codes~\parencite{bravyiQuantumCodesLattice1998}, and colour codes~\parencite{bombinTopologicalQuantumDistillation2006}.
Stabiliser codes are defined by a set of commuting Pauli measurements.
They encode logical qubits into a larger physical space, where the code space is the subspace stabilised by the Pauli measurements.
This defines a static stabilising group through which codes are usually studied.

Subsystem codes are a generalisation of stabiliser codes that allow for anti-commuting measurements~\parencite{kribsUnifiedGeneralized2005}.
This results in measurements that can update the group of stabilisers, leading to different instantaneous stabiliser groups at each time step.
However, subsystem codes require that all stabiliser measurements commute with the logicals, so as to ensure that we do not lose any logical information. 

Generalising this, Floquet codes~\parencite{hastingsDynamicallyGenerated2021,vuillotPlanarFloquet2021} are a recently discovered class of error-correcting codes that allow non-commuting measurements that, when treated as subsystem codes, project onto the logicals.
However, when arranged in a specific order, the measurements of a Floquet code preserve all logical information and instead update the subspace in which the logicals are encoded.
Floquet codes can outperform comparable stabiliser codes in certain architectures~\parencite{hilaireEnhancedFaulttolerance2024} and may implement some logical Clifford gates fault-tolerantly at no extra cost~\parencite{aasenAdiabaticPaths2022}.
As Floquet codes require a specific ordering of the measurements, these codes are often studied via the quantum circuit defined by their measurement schedule.

In this paper, we leverage this circuit-centric view of Floquet codes to construct new error-correcting codes from existing ones.
In particular, given an arbitrary stabiliser code, we derive a new Floquet code that uses operations of weight at most two.
This is done while preserving the distance and number of encoded logical qubits.
The main advantage of this procedure is that, as the resulting Floquet codes have measurements of weight at most two, they are substantially easier to implement. 
To do this, we take the quantum circuits obtained from infinitely repeating the measurement schedules of the original stabiliser codes and manipulate them until they represent measurement circuits of Floquet codes with measurements of weight at most two. 
Leveraging the connection between the code distance and the circuit distance of the measurement circuit, we guarantee that the resulting Floquet code has the same distance as the original stabiliser code by only doing circuit manipulations that preserve the circuit distance. 

To perform the circuit manipulations, we utilise the ZX calculus~\parencite{coeckeInteractingQuantumObservables2008}, a graphical language for representing and rewriting quantum circuits.
The ZX calculus comes equipped with a set of rewrite rules that can be used to prove the equivalence of circuits.
The ZX calculus is increasingly being used to study error correction, with previous work primarily focusing on encoders~\parencite{kissingerPhasefreeZXDiagrams2022,kissingerScalableSpiderNests2024, huangGraphicalCSSCode2023, duncanGraphtheoreticSimplification2020} and lattice surgery~\parencite{debeaudrapZXCalculus2020}.
Following a more recent trend in ZX for quantum error correction~\parencite{townsend-teagueFloquetifyingColourCode2023,bombinUnifyingFlavorsFault2024, mcewenRelaxingHardwareRequirements2023,gidneyPairMeasurement2023}, our work will focus on the measurement circuit of error correction codes.

Our Floquetification procedure first expresses the measurement circuit of the original code as a ZX diagram.
It then uses ZX rewrites to synthesise a new circuit in terms of single- and two-qubit operations.
Finally, it interprets this circuit as a novel Floquet code.
An issue with the use of ZX calculus to manipulate error-correcting codes is that ZX rewrites generally do not preserve code distance.
To tackle this problem, we introduce the notion of \emph{distance-preserving rewrites}.
Such rewrites enable the transformation of error-correcting codes without changing their distance.
This gives us a highly flexible perspective on error-correcting codes, which will likely offer insights beyond the scope of this paper.

To guarantee that the Floquetification procedure preserves the distance of the original code, we restrict ourselves to distance-preserving rewrites.
Using these, we propose an algorithm that transforms any $\code{n, k, d}$ stabiliser code with Pauli measurements of weight at most $w$ into a new $\code{n + \lceil \frac{w}{2} \rceil + \ell, k, d}$ Floquet code with only single- and two-qubit operations where $\ell \leq \log_{2}{w}$.
Our process generalises the Floquetification of~\textcite{townsend-teagueFloquetifyingColourCode2023} to arbitrary stabiliser codes, provably preserving the number of logical qubits and code distance while having a smaller overhead of physical qubits\footnote{\textcite{townsend-teagueFloquetifyingColourCode2023} synthesise a $\code{12, 2, 2}$ Floquet code from the $\code{4, 2, 2}$ stabiliser code while our procedure creates a $\code{6, 2, 2}$ Floquet code. A presentation of the Floquetification of the $\code{4, 2, 2}$ proposed by \textcite{townsend-teagueFloquetifyingColourCode2023} in terms of the language of this paper is given in \autoref{appendix:floquetification-teague}.}.

In contrast to previous works, we start with a local perspective, first focusing on Floquetifying individual Pauli measurements.
A weight-$w$ measurement is expressed as an equivalent circuit with only single- and two-qubit operations using at most $\lceil \frac{w}{2} \rceil + \log_2{w}$ ancillary qubits.
Restricting to distance-preserving rewrites guarantees that a single, undetectable error in the circuit results in at most one error on the data qubits.

Our Floquetification procedure provides a more dynamic perspective on stabiliser codes.
This enables compiling stabiliser codes with specific hardware restrictions in mind.
Simultaneously, distance-preserving rewrites lay the foundation for diagrammatic reasoning about quantum circuits in a noisy setting.

\section{Preliminaries}
\label{sec:prelims}

\subsection{Clifford circuits as quantum error-correcting codes}
Following~\parencite{delfosseSpacetimeCodes2023,townsend-teagueFloquetifyingColourCode2023,bombinUnifyingFlavorsFault2024, mcewenRelaxingHardwareRequirements2023}, we take a circuit-centric approach to quantum error correction.
Similar to~\textcite{townsend-teagueFloquetifyingColourCode2023, blackwellCodeDistance2025}, we define:

\begin{definition}[Floquet code]
 A \textit{Floquet code} on $n$ qubits consists of a measurement schedule $\mathcal M = [M_1, M_2, \dots, M_m]$ of Pauli operators. 
 We call the circuit $C_{\mathcal M}$ obtained by measuring the Pauli operators of the measurement schedule in an infinite loop the \textit{measurement circuit}.
\end{definition}

\noindent The measurement circuit can be seen as infinitely performing the measurement schedule on some unknown input state. 
As introduced by~\textcite{hastingsDynamicallyGenerated2021}, we consider the instantaneous stabiliser group $S_t$ for each step $t$ of the measurement schedule. 
We define $S_t$ as the group of known Pauli stabilisers of the circuit after the first $t$ operations.
As we assume the input state to be unknown, the instantaneous stabiliser group only depends on the measurement circuit. 
Any Pauli measurement of the measurement circuit either preserves the size of the ISG or doubles it~\parencite{hastingsDynamicallyGenerated2021}.
We call the circuit \emph{established} at timestep $T$, when for all $t \geq T$, $|S_t| = |S_T|$.
This is also known as the \emph{steady stage} of the code \parencite{blackwellCodeDistance2025}.
Let $m$ be the size of a minimal generating set for $S_T$. 
Then, at this point, the code can encode $k = n - m$ logical qubits~\parencite{townsend-teagueFloquetifyingColourCode2023}\footnote{We remark that the infinite repetition of the measurement schedule is a common mathematical trick when analysing Floquet codes to avoid boundary conditions (e.g. \parencite{townsend-teagueFloquetifyingColourCode2023,blackwellCodeDistance2025}). Boundary conditions in the beginning are addressed via the idea of establishment. To prevent having to do something similar in the end, we simply never stop.}.

This definition of Floquet codes generalises stabiliser codes and subsystem codes, each of which can be viewed as Floquet codes by taking a generating set of their stabilisers as the measurement schedule of a Floquet code. 
However, the Floquet codes obtained from stabiliser codes and subsystem codes have special properties~\parencite{townsend-teagueFloquetifyingColourCode2023}: 
\begin{definition}[Stabiliser codes]
  Let $\mathcal M = [M_1, M_2, \dots, M_m]$ be a Floquet code and $\mathcal G = M_1, M_2, \dots, M_k$ be the union of all Pauli operations in the measurement schedule. 
  We say $\mathcal M$ is a stabiliser code if $\mathcal G$ is Abelian.
\end{definition}

For subsystem codes, we observe that usually, all measurements are taken together to form the gauge group. 
The logical operators are those Paulis that commute with all the gauges. 
We define:
\begin{definition}[Associated subsystem code of an Floquet code]
  Let $\mathcal M = [M_1, M_2, \dots, M_m]$ be a Floquet code.
  The associated subsystem code of $\mathcal M$ is defined by the gauge group $\mathcal G_{\mathcal M} = \langle i \rangle M_1, M_2, \dots, M_k$. 
  Let $\mathcal S$ be a stabiliser group such that $\langle i \rangle S = Z(\mathcal G_{\mathcal M})$, i.e.\@ the centre of $\mathcal G_{\mathcal M}$.
  Then we say $\mathcal G_{\mathcal M}$ has $k$ logicals when $\mathcal N(\mathcal S) / \mathcal G_{\mathcal M}$ is isomorphic to $\langle X_1, Z_1, \dots X_k, Z_k \rangle$.
\end{definition}

Floquet codes were originally defined as codes that, when read as subsystem codes, encode less information than they actually do. 
We turn this definition around and define subsystem codes as Floquet codes that encode the same amount of information, no matter whether they are read as subsystem codes or Floquet codes:
\begin{definition}[Subsystem code]
  Let $\mathcal M = [M_1, M_2, \dots, M_m]$ be a Floquet code that encodes $k$ logical qubits. 
  We say that $\mathcal M$ is a subsystem code when the associated subsystem code $\mathcal G_{\mathcal M}$ of $\mathcal M$ has $k$ logicals. 
\end{definition}

\noindent \textcite{townsend-teagueFloquetifyingColourCode2023} shows that the measurement circuits of stabiliser codes and subsystem codes satisfy the definitions above.
Floquet codes that do not satisfy these properties cannot be understood as stabiliser codes (as they have anticommuting measurements) or subsystem codes (as they would appear to encode fewer logicals than they actually do).

Finally, we define: 
\begin{definition}[Proper Floquet code]
  We say $\mathcal M = [M_1, M_2, \dots, M_m]$ is a proper Floquet code when it is not a stabiliser code or a subsystem code.
\end{definition}

As such, Floquetification can be trivialised by viewing any stabiliser code as a Floquet code. 
However, in this work, our procedure produces Floquet codes whose measurement weight is at most two, leading to codes that are substantially simpler to implement. 
Furthermore, in~\autoref{sec:proper-floquet}, we show that the codes resulting from our procedure are proper Floquet codes.

An important property of any code is its distance, which quantifies how good a code is at detecting and correcting errors. 
For Floquet codes, we define the distance as a property of their measurement circuit. 
Thus, we first have to define what we mean by errors on the measurement circuit.
We consider our measurement circuits to be implemented fault-tolerantly, meaning that a single error inside a measurement circuit can create at most one measurement flip and a single data error:

\begin{definition}[Errors]
  \label{def:errors}
 Let $\mathcal M$ define a Floquet code. 
 Then the set of atomic errors $\mathcal E$ on $C_{\mathcal M}$ of $\mathcal M$ consists of single-qubit Pauli flips between measurements and measurement-flips which can change the outcome of the measurement and/or create a single-qubit flip on one of the outputs of the measurement. 
 The errors $E \subseteq \mathcal E$ on the measurement circuit of $\mathcal M$ are subsets of atomic errors. 
 The weight of an error $E$ is the number of atomic errors that it consists of, i.e.\@ $weight(E) = |E|$.
 Applying an error to a measurement circuit gives a new circuit $C_{\mathcal M}^E$. 
\end{definition}

Next, we observe that there are certain errors which are trivial in that they do not have any consequence on the circuit.
For example, in one time step of the measurement circuit, an error that is in the ISG of the circuit at that time step corresponds to a stabiliser and is therefore trivial. 
In other work, these errors are called \emph{benign} \parencite{blackwellCodeDistance2025} or \emph{inconsequential} \parencite{fuenteXYZRuby2024,ruschCompletenessFault2025a}, we call them \textit{trivial errors}:
\begin{definition}[Trivial error]
 Let $\mathcal M$ be a Floquet code and $E \subseteq \mathcal E$ be an error on the measurement circuit $C_{\mathcal M}$. 
 Then we say $E$ is trivial if $C_{\mathcal M} = C_{\mathcal M}^E$, i.e.\@ when the error has no impact on the underlying linear map.
\end{definition}
We remark that \textcite{blackwellCodeDistance2025} define the trivial errors in terms of local gauges of the tensor network representing our quantum circuit.
However, \textcite{ruschCompletenessFault2025a} show that this is exactly equivalent to the set of all errors that do not change the underlying linear map represented by the circuit.

Next, there is another special class of errors, namely the detectable ones. 
Following \textcite{townsend-teagueFloquetifyingColourCode2023}, without loss of generality, we can treat the measurement circuit not as a sequence of measurements of the Pauli operators but rather as projections onto the $+1$-eigenspace of the operators, representing one trivial syndrome outcome. 
For more details, see the Appendix of \parencite{townsend-teagueFloquetifyingColourCode2023}.
In the noise-free setting, there is always a non-zero chance of observing this particular measurement outcome. 
However, some errors make it impossible to ever get this expected measurement outcome. 
In particular, we can observe that if the sequence of projections onto the trivial syndrome outcome is $0$, then the probability of ever measuring this outcome is zero. 
In other words, we never obtain the trivial syndrome outcome and therefore detect the error.
\begin{definition}
 Let $\mathcal M$ be a Floquet code and $E \subseteq \mathcal E$ be an error on the measurement circuit $C_{\mathcal M}$. 
 We say $E$ is detectable when $C_{\mathcal M}^E = 0$. 
\end{definition}

Finally, using all of these concepts, we can define code distance. 
For this, we observe that the error correcting capabilities of a code only take effect after establishment, i.e.\@ once the ISG group does not grow any more.
The distance of a code can only be meaningfully defined after its establishment \parencite{blackwellCodeDistance2025}.
We have: 
\begin{definition}[Distance of a Floquet code]
 Let $\mathcal M$ be a Floquet code. 
 Then the distance of $\mathcal M$ is the minimum weight of all non-trivial, undetectable errors that happen after the code is established.
\end{definition}

This definition of distance is a natural generalisation of the distance of stabiliser codes:

\begin{proposition}
  \label{prop:stabiliser-code-distance}
 Let $S_1, \dots, S_m$ be a generating set for the stabiliser group of an $\interp{n, k, d}$ stabiliser code $S$.
 Let $\mathcal{M}_S = [S_1, \dots, S_m]$ denote the Floquet code obtained from $S$ by infinitely repeating its stabiliser generators.
 Then the distance of the Floquet code $\mathcal{M}_S$ is $d$.
\end{proposition}
\begin{proof}
 We first show that $dist(\mathcal{M}_S) \leq dist(S)$.
 Let $E$ be a minimum-weight non-trivial, undetectable error on $S$.
 Placing $E$ between any two measurements of $\mathcal{M}_S$ after establishment produces an error on $C_{\mathcal{M}_S}$ of the same weight as $E$.
 We know $E$ commutes with all stabilisers of $S$. 
 Therefore, $E$ does not flip any of the measurements and is thus undetectable in $\mathcal{M}_S$.
 Furthermore, as $E$ is non-trivial in $S$, it must also be non-trivial in $C_{\mathcal{M}_S}$.
 Hence, $E$ is a non-trivial, undetectable error on $\mathcal{M}_S$, and thus $dist(\mathcal{M}_S) \leq |E| = dist(S)$.

 To show that $dist(S) \leq dist(\mathcal{M}_S)$, let $E$ be a minimum-weight non-trivial, undetectable error on $C_{\mathcal{M}_S}$.
 In general, $E$ may contain both qubit-flip and measurement-flip components distributed over multiple time steps in the spacetime of $\mathcal M_S$.
 We will show that $E$ is equivalent to an error $E'$ of at most the weight of $E$ that consists only of qubit flips occurring in a single time step.
 
 For this, we will first push all qubit flips into a single time step. 
 This will, potentially, introduce additional measurement flips in other time steps.
 Then we will argue that, as $E$ was assumed to be undetectable, these measurement flips must cancel each other out, such that we are left only with qubit flips in a single time step. 
 
  Let $t_0$ be the first time step on which $E$ acts non-trivially.
  We push every qubit-flip component of $E$ backwards to time $t_0$.
  Whenever a qubit flip is moved past an anticommuting measurement, a compensating measurement flip is introduced so that the resulting error remains equivalent.

  This procedure does not increase the number of qubit flips.
  Indeed, qubit flips may combine or cancel, but no additional qubit flips are created.
  Hence we obtain an equivalent error $E'$ consisting of at most $|E|$ qubit flips only at time $t_0$, together with possibly some measurement flips at later times.
  Since $E$ is undetectable, so is $E'$.
  All measurements after the code has been established are predetermined, so any non-trivial measurement flip would be detectable.
  Therefore, all measurement flips in $E'$ must cancel, and $E'$ reduces to an error consisting only of qubit flips in $t_0$.

  Thus $E'$ is a non-trivial undetectable error of weight at most $|E|$, supported on a single time step.
  Interpreting $E'$ as an error on the static code $S$, we obtain a non-trivial undetectable error on $S$ of weight at most $|E|$.
  Therefore $dist(S) \leq dist(\mathcal{M}_S)$.

  Since the reverse inequality holds as well, we conclude that $dist(S) = dist(\mathcal{M}_S)$.
\end{proof}

\subsection{ZX calculus}
\label{subsec:zx-calculus}

This section introduces the ZX calculus~\parencite{coeckeInteractingQuantumObservables2008}, its generators, and its set of axioms.
Our presentation only aims to deliver a minimal understanding of the framework necessary to read this paper.
For a more thorough discussion on the topic, we refer to~\parencite{vandeweteringZXcalculusWorkingQuantum2020, KissingerWetering2024Book}.

The ZX calculus is a graphical language used to represent and reason about quantum computation.
For the purposes of this paper, we restrict to the Clifford fragment of the ZX calculus, which is complete for stabiliser quantum mechanics~\parencite{backensZXcalculusComplete2014}.
Its elementary building blocks are the green \emph{Z-spider} and the red \emph{X-spider} (therefore ZX)\footnote{If you are reading this in black and white or have limited colour vision, green = lightly shaded, red = darkly shaded.}:
\begin{align*}
 \textit{Z spider:} \qquad
 \input{ZX/generators/z-spider.tikz}
 \quad &:= \quad
 \ket{0}^{\otimes n}\! \bra{0}^{\otimes m} + e^{i k\frac{\pi}{2}} \ket{1}^{\otimes n}\! \bra{1}^{\otimes m}\\[8pt]
 \textit{X spider:} \qquad
 \input{ZX/generators/x-spider.tikz}
 \quad &:= \quad
 \ket{+}^{\otimes n}\! \bra{+}^{\otimes m} + e^{i k\frac{\pi}{2}} \ket{-}^{\otimes n}\! \bra{-}^{\otimes m}
\end{align*}
A spider has an arbitrary number of legs, corresponding to qubit ports, and a phase parameter. 
Restricting ourselves to the Clifford fragment of the ZX calculus, phases are restricted to $\alpha = k \frac{\pi}{2}$.
Since the $e^{i k \frac{\pi}{2}}$ function in the interpretation is $2 \pi$ periodic, the parameter of spiders is taken modulo $2 \pi$.
Finally, as a shorthand, we define the Hadamard box:
\[
  \input{ZX/generators/h-def-scalar-accurate.tikz} \ \ =\ \ \ket{0}\bra{+} + \ket{-}\bra{1}
\]
%\[
%\tikzfig{ZX/generators/h-def} \ \ :=\ \ \ket{0}\bra{+} + \ket{-}\bra{1}
%\]
Here, the free floating number in the grey spider represents the global scalar applied to the linear map. 
When the global scalar is one, we omit it.
Throughout most of this paper, we will work up to global scalars, except for one proof where they are necessary. 

Using these building blocks, we can simply represent common quantum maps as well as any unitary; some examples are shown in \autoref{fig:circuit-to-zx}.
\begin{figure}[H]
\begin{minipage}{.33\textwidth}
    \begin{align*}
 \ket{0}
 \quad 
      &\longmapsto \quad
 \input{ZX/elements/x-state-zero.tikz} \\[8pt]
      % \ket{1} 
      % \quad &\longmapsto \quad
      % \tikzfig{ZX/elements/x-state-pi}\\
 \ket{+} 
 \quad &\longmapsto \quad
 \input{ZX/elements/z-state-zero.tikz} \\[8pt]
      % \ket{-}
      % \quad &\longmapsto \quad
      % \tikzfig{ZX/elements/z-state-pi}
 \input{QuantumCircuit/not.tikz} 
 \quad &\longmapsto \quad
 \input{ZX/elements/not.tikz} \\[8pt]
 \input{QuantumCircuit/y.tikz} 
 \quad &\longmapsto \quad
 \input{ZX/elements/y.tikz} \\[8pt]
 \input{QuantumCircuit/s.tikz} 
 \quad &\longmapsto \quad
 \input{ZX/elements/s.tikz} \\[8pt]
      % \tikzfig{QuantumCircuit/z-rotate-alpha}
      % \quad &\longmapsto \quad
      % \tikzfig{ZX/elements/z-rotate-alpha}
    \end{align*}
  \end{minipage}
  \begin{minipage}{.33\textwidth}
    \begin{align*}
 \input{QuantumCircuit/cnot.tikz}
 \quad \longmapsto& \quad
 \input{ZX/elements/cnot.tikz} \\[8pt]
 \input{QuantumCircuit/zz-parity-check.tikz} 
 \quad \longmapsto& \quad
 \input{ZX/elements/ParityCheck.tikz} \\
 \overset{(k=0)}{=}& \quad
 \input{ZX/elements/ParityCheck-2.tikz} \\
    \end{align*}
  \end{minipage}
  \begin{minipage}{.33\textwidth}
    \begin{align*}
 \input{QuantumCircuit/zzzz-parity-check.tikz}
 \quad &\longmapsto \quad
 \input{ZX/elements/ParityCheck-4.tikz} \\
 \input{QuantumCircuit/zxyz-parity-check.tikz}
 \quad &\longmapsto \quad
 \input{ZX/elements/ParityCheck-ZXYZ.tikz}
    \end{align*} \\
  \end{minipage}
  \caption{Translation between quantum circuits and ZX diagrams. The red $k\pi$ spider indicates the outcome of the Pauli measurements.}
  \label{fig:circuit-to-zx}
\end{figure}

\begin{definition}[ZX diagram]
 We define a ZX diagram $D$ as a graph $G = (V, E)$ with phases $\alpha \colon V \to \{0, \frac{\pi}{2}, \pi, \frac{3\pi}{2}\}$ and vertex types $\tau \colon V \to \{Z, X, In\text{(put), Out\text((put))}\}$. 
 We call the collection of input and output nodes \emph{boundary nodes}, and we restrict boundary nodes to have a single leg.
 We call edges that are connected to an input or an output \textit{input edges} and \textit{output edges}; together, we call them \textit{boundary edges}.
 All other edges are called \textit{internal edges}.
 For a ZX diagram $D$, we denote the linear map it represents by $\interp{D}$.
 This function $\interp{\cdot}$ that maps diagrams to linear maps is called the \emph{interpretation} function.
\end{definition}

The ZX calculus comes equipped with a set of graphical rewrite rules.
We can apply these rewrite rules to equate diagrams that interpret to the same linear map.
In fact, the axioms of the ZX calculus are sufficient to derive any equality between qubit maps~\parencite{ngUniversalCompletion2017,jeandelDiagrammaticReasoning2018}.

\begin{remark}
  \label{rem:scalar-zx}
 There are various axiomatizations of the ZX calculus.
 The version we present is for the stabiliser fragment.
 It has been shown that any two ZX diagrams that represent the same linear map can be rewritten into one another using only these rules~\parencite{backensZXcalculusComplete2014, backensSimplifiedStabilizer2017}.
 Versions of the ZX calculus can be found in~\textcite{jeandelDiagrammaticReasoning2018} for the Clifford+T fragment, and in~\textcite{vilmartMinimalAxiomatisation2019} for the full language.
\end{remark}

%\[
% \tikzfig{ZX/axioms}
%\]
\[
  \begin{tikzpicture}
	\begin{pgfonlayer}{nodelayer}
		\node [style=Z phase dot] (0) at (-9.75, 2.75) {$k_1 \frac{\pi}{2}$};
		\node [style=none] (1) at (-11.25, 3.5) {};
		\node [style=none] (2) at (-11.25, 2) {};
		\node [style=none] (3) at (-7.75, 3.5) {};
		\node [style=none] (4) at (-7.75, 2) {};
		\node [style=none] (5) at (-7.5, 3) {$\vdots$};
		\node [style=none] (6) at (-11, 3) {$\vdots$};
		\node [style=Z phase dot] (7) at (-8.75, -0.25) {$k_2 \frac{\pi}{2}$};
		\node [style=none] (8) at (-10.75, 0.5) {};
		\node [style=none] (9) at (-10.75, -1) {};
		\node [style=none] (10) at (-7.25, 0.5) {};
		\node [style=none] (11) at (-7.25, -1) {};
		\node [style=none] (12) at (-7.5, 0) {$\vdots$};
		\node [style=none] (13) at (-11, 0) {$\vdots$};
		\node [style=Z phase dot] (14) at (-2.75, 1.25) {$(k_1 + k_2) \frac{\pi}{2}$};
		\node [style=none] (15) at (-1.25, 3.5) {};
		\node [style=none] (16) at (-1.25, 2) {};
		\node [style=none] (17) at (-1.25, 0.5) {};
		\node [style=none] (18) at (-1.25, -1) {};
		\node [style=none] (19) at (-4.25, 3.5) {};
		\node [style=none] (20) at (-4.25, 2) {};
		\node [style=none] (21) at (-4.25, 0.5) {};
		\node [style=none] (22) at (-4.25, -1) {};
		\node [style=none] (23) at (-1.5, 3) {$\vdots$};
		\node [style=none] (24) at (-1.5, 0) {$\vdots$};
		\node [style=none] (25) at (-4, 3) {$\vdots$};
		\node [style=none] (26) at (-4, 0) {$\vdots$};
		\node [style=none] (27) at (-11.25, 0.5) {};
		\node [style=none] (28) at (-11.25, -1) {};
		\node [style=none] (29) at (-7.25, 3.5) {};
		\node [style=none] (30) at (-7.25, 2) {};
		\node [style=none, label={above:\Fusion}] (31) at (-5.75, 1.25) {$=$};
		\node [style=Z phase dot] (32) at (-8.75, -3.5) {$k \frac{\pi}{2}$};
		\node [style=none] (33) at (-7.25, -2.75) {};
		\node [style=none] (34) at (-7.25, -4.25) {};
		\node [style=none] (35) at (-10.25, -2.75) {};
		\node [style=none] (36) at (-10.25, -4.25) {};
		\node [style=H box] (37) at (-7.75, -2.75) {};
		\node [style=H box] (38) at (-7.75, -4.25) {};
		\node [style=H box] (39) at (-9.75, -2.75) {};
		\node [style=H box] (40) at (-9.75, -4.25) {};
		\node [style=none, label={above:\Colour}] (41) at (-5.75, -3.5) {$=$};
		\node [style=none] (42) at (-7.5, -3.25) {$\vdots$};
		\node [style=X phase dot] (43) at (-2.75, -3.5) {$k \frac{\pi}{2}$};
		\node [style=none] (44) at (-1.25, -2.75) {};
		\node [style=none] (45) at (-1.25, -4.25) {};
		\node [style=none] (46) at (-4.25, -2.75) {};
		\node [style=none] (47) at (-4.25, -4.25) {};
		\node [style=none] (48) at (-1.75, -2.75) {};
		\node [style=none] (49) at (-1.75, -4.25) {};
		\node [style=none] (50) at (-3.75, -2.75) {};
		\node [style=none] (51) at (-3.75, -4.25) {};
		\node [style=none] (52) at (-4, -3.25) {$\vdots$};
		\node [style=none] (53) at (-1.5, -3.25) {$\vdots$};
		\node [style=none] (54) at (-10.25, -2.75) {};
		\node [style=none] (55) at (-10, -3.25) {$\vdots$};
		\node [style=X phase dot] (56) at (3, 0.5) {$\pi$};
		\node [style=Z phase dot] (57) at (4.25, 0.5) {$k \frac{\pi}{2}$};
		\node [style=none, label={above:\PiCommute}] (58) at (7.25, 0.5) {$=$};
		\node [style=none] (59) at (5.25, 1.25) {};
		\node [style=none] (60) at (5.25, -0.25) {};
		\node [style=none] (61) at (2.25, 0.5) {};
		\node [style=Z phase dot] (62) at (10.5, 0.5) {$\minu k \frac{\pi}{2}$};
		\node [style=X phase dot] (63) at (11.75, 1.25) {$\pi$};
		\node [style=X phase dot] (64) at (11.75, -0.25) {$\pi$};
		\node [style=none] (65) at (8.5, 0.5) {};
		\node [style=none] (66) at (12.5, 1.25) {};
		\node [style=none] (67) at (12.5, -0.25) {};
		\node [style=X dot] (68) at (3.75, -6.25) {};
		\node [style=Z dot] (69) at (4.75, -6.25) {};
		\node [style=none, label={above:\Copy}] (70) at (7.25, -6.25) {$=$};
		\node [style=X dot] (71) at (8.75, -5.5) {};
		\node [style=X dot] (72) at (8.75, -7) {};
		\node [style=none] (73) at (9.75, -5.5) {};
		\node [style=none] (74) at (9.75, -7) {};
		\node [style=none] (75) at (5.5, -5.5) {};
		\node [style=none] (76) at (5.5, -7) {};
		\node [style=none] (77) at (2.25, 3.25) {};
		\node [style=none] (78) at (4.25, 3.25) {};
		\node [style=none] (79) at (6.25, 3.25) {};
		\node [style=none] (80) at (8.25, 3.25) {};
		\node [style=none] (81) at (10.25, 3.25) {};
		\node [style=none] (82) at (12.25, 3.25) {};
		\node [style=X dot] (83) at (11.25, 3.25) {};
		\node [style=Z dot] (84) at (3.25, 3.25) {};
		\node [style=none, label={above:\ZElim}] (85) at (5.25, 3.25) {$=$};
		\node [style=none, label={above:\XElim}] (86) at (9.25, 3.25) {$=$};
		\node [style=none, label={above:\Bigebra}] (87) at (7.25, -3.25) {$=$};
		\node [style=none] (88) at (9, -2.625) {};
		\node [style=none] (89) at (9, -3.875) {};
		\node [style=X dot] (90) at (9.75, -3.25) {};
		\node [style=Z dot] (91) at (10.75, -3.25) {};
		\node [style=none] (92) at (11.5, -2.625) {};
		\node [style=none] (93) at (11.5, -3.875) {};
		\node [style=none] (94) at (3, -2.625) {};
		\node [style=none] (95) at (3, -3.875) {};
		\node [style=Z dot] (96) at (3.625, -2.625) {};
		\node [style=Z dot] (97) at (3.625, -3.875) {};
		\node [style=X dot] (98) at (5.125, -2.625) {};
		\node [style=X dot] (99) at (5.125, -3.875) {};
		\node [style=none] (100) at (5.75, -2.625) {};
		\node [style=none] (101) at (5.75, -3.875) {};
		\node [style=none] (102) at (8.75, -2.625) {};
		\node [style=none] (103) at (8.75, -3.875) {};
		\node [style=none] (104) at (11.75, -2.625) {};
		\node [style=none] (105) at (11.75, -3.875) {};
		\node [style=none] (106) at (5.75, -0.25) {};
		\node [style=none] (107) at (5.75, 1.25) {};
		\node [style=none] (108) at (5.75, -5.5) {};
		\node [style=none] (109) at (5.75, -7) {};
		\node [style=scalar] (110) at (-4.25, -6.25) {$0$};
		\node [style=none, label={above:\Zero}] (111) at (-5.75, -6.25) {=};
		\node [style=X phase dot] (112) at (-7.25, -6.25) {$\pi$};
		\node [style=scalar] (119) at (9.75, 1.5) {$e^{i k \frac{\pi}{2}}$};
		\node [style=scalar] (120) at (2.5, -3.25) {$\sqrt{2}$};
		\node [style=scalar] (121) at (4, -5.5) {$\sqrt{2}$};
	\end{pgfonlayer}
	\begin{pgfonlayer}{edgelayer}
		\draw [in=180, out=30] (0) to (3.center);
		\draw [in=180, out=-30] (0) to (4.center);
		\draw [in=360, out=120] (0) to (1.center);
		\draw [in=0, out=-120] (0) to (2.center);
		\draw [in=180, out=60] (7) to (10.center);
		\draw [in=180, out=-60] (7) to (11.center);
		\draw [in=360, out=150] (7) to (8.center);
		\draw [in=0, out=-150] (7) to (9.center);
		\draw [in=180, out=30] (14) to (16.center);
		\draw [in=180, out=60] (14) to (15.center);
		\draw [in=180, out=-30] (14) to (17.center);
		\draw [in=180, out=-60] (14) to (18.center);
		\draw [in=0, out=120] (14) to (19.center);
		\draw [in=0, out=150] (14) to (20.center);
		\draw [in=0, out=-150] (14) to (21.center);
		\draw [in=0, out=-120] (14) to (22.center);
		\draw (27.center) to (8.center);
		\draw (28.center) to (9.center);
		\draw (4.center) to (30.center);
		\draw (3.center) to (29.center);
		\draw [in=135, out=-45] (0) to (7);
		\draw (39) to (35.center);
		\draw (40) to (36.center);
		\draw (37) to (33.center);
		\draw (38) to (34.center);
		\draw [in=180, out=-75] (32) to (38);
		\draw [in=0, out=-105] (32) to (40);
		\draw [in=360, out=105] (32) to (39);
		\draw [in=180, out=75] (32) to (37);
		\draw (50.center) to (46.center);
		\draw (51.center) to (47.center);
		\draw (48.center) to (44.center);
		\draw (49.center) to (45.center);
		\draw [in=180, out=-75] (43) to (49.center);
		\draw [in=0, out=-105] (43) to (51.center);
		\draw [in=360, out=105] (43) to (50.center);
		\draw [in=180, out=75] (43) to (48.center);
		\draw (57) to (56);
		\draw [in=180, out=60] (57) to (59.center);
		\draw [in=180, out=-60] (57) to (60.center);
		\draw (61.center) to (56);
		\draw [in=180, out=60] (62) to (63);
		\draw [in=180, out=-60] (62) to (64);
		\draw (64) to (67.center);
		\draw (66.center) to (63);
		\draw (65.center) to (62);
		\draw (69) to (68);
		\draw (72) to (74.center);
		\draw (71) to (73.center);
		\draw [in=180, out=60] (69) to (75.center);
		\draw [in=180, out=-60] (69) to (76.center);
		\draw (77.center) to (84);
		\draw (84) to (78.center);
		\draw (79.center) to (80.center);
		\draw (81.center) to (83);
		\draw (83) to (82.center);
		\draw (90) to (91);
		\draw [in=0, out=-105, looseness=0.75] (90) to (89.center);
		\draw [in=0, out=105, looseness=0.75] (90) to (88.center);
		\draw [in=180, out=75, looseness=0.75] (91) to (92.center);
		\draw [in=-180, out=-75, looseness=0.75] (91) to (93.center);
		\draw (94.center) to (96);
		\draw (95.center) to (97);
		\draw (98) to (100.center);
		\draw (99) to (101.center);
		\draw (96) to (99);
		\draw (97) to (98);
		\draw [bend left=15] (96) to (98);
		\draw [bend right=15] (97) to (99);
		\draw (89.center) to (103.center);
		\draw (102.center) to (88.center);
		\draw (104.center) to (92.center);
		\draw (105.center) to (93.center);
		\draw (106.center) to (60.center);
		\draw (107.center) to (59.center);
		\draw (109.center) to (76.center);
		\draw (108.center) to (75.center);
	\end{pgfonlayer}
\end{tikzpicture}

\]
Note that we use explicit scalars, a construction more rigorously characterised in previous works~\parencite{poorUniqueNormalForm2022,poorQupitStabiliser2023}.
For the remainder of this paper, we only show equalities up to some non-zero scalar, unless the scalars play an important role in some derivation.

We can use these rules to show relationships between ZX diagrams.
For example, we can use \TextPiCommute to show that $\pi$-phased spiders commute past each other: 
\[\input{pi-commute-proof.tikz}\]
More generally, if $D_1 = D_2$ for some diagrams $D_1$ and $D_2$, then $D_1$ can be replaced with $D_2$ in all diagrams $D$ in which it occurs:
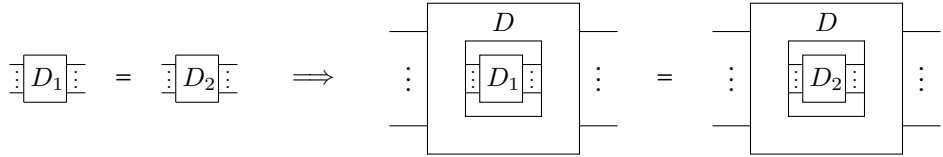
\[\begin{tikzpicture}
	\begin{pgfonlayer}{nodelayer}
		\node [style=none] (92) at (0, 0) {=};
		\node [style=none] (143) at (-14.25, 0) {=};
		\node [style=none] (144) at (-9.25, 0) {$\Longrightarrow$};
		\node [style=none] (145) at (-6.25, 2) {};
		\node [style=none] (146) at (-4.875, 0.625) {};
		\node [style=none] (147) at (-3.75, 0.625) {};
		\node [style=none] (148) at (-4.875, -0.625) {};
		\node [style=none] (149) at (-3.75, -0.625) {};
		\node [style=none] (150) at (-4.875, 0.375) {};
		\node [style=none] (151) at (-5.25, 0.375) {};
		\node [style=none] (152) at (-5.25, -0.375) {};
		\node [style=none] (153) at (-4.875, -0.375) {};
		\node [style=none] (154) at (-3.25, 0.375) {};
		\node [style=none] (155) at (-3.75, 0.375) {};
		\node [style=none] (156) at (-3.75, -0.375) {};
		\node [style=none] (157) at (-3.25, -0.375) {};
		\node [style=none] (158) at (-4.25, 0) {$D_1$};
		\node [style=none] (159) at (-5.25, 1) {};
		\node [style=none] (160) at (-3.25, 1) {};
		\node [style=none] (161) at (-5.25, -1) {};
		\node [style=none] (162) at (-3.25, -1) {};
		\node [style=none] (163) at (-2.25, -2) {};
		\node [style=none] (164) at (-2.25, 2) {};
		\node [style=none] (165) at (-6.25, -2) {};
		\node [style=none] (166) at (-6.25, 1.25) {};
		\node [style=none] (167) at (-6.25, -1.25) {};
		\node [style=none] (168) at (-7.25, 1.25) {};
		\node [style=none] (169) at (-7.25, -1.25) {};
		\node [style=none] (170) at (-6.75, 0.2) {\vdots};
		\node [style=none] (171) at (-1.25, 1.25) {};
		\node [style=none] (172) at (-1.25, -1.25) {};
		\node [style=none] (173) at (-2.25, 1.25) {};
		\node [style=none] (174) at (-2.25, -1.25) {};
		\node [style=none] (175) at (-1.75, 0.2) {\vdots};
		\node [style=none] (176) at (-4.25, 1.5) {$D$};
		\node [style=none, scale=.7] (177) at (-5.05, 0.125) {$\vdots$};
		\node [style=none, scale=.7] (178) at (-3.5, 0.125) {$\vdots$};
		\node [style=none] (179) at (2.25, 2) {};
		\node [style=none] (180) at (3.625, 0.625) {};
		\node [style=none] (181) at (4.75, 0.625) {};
		\node [style=none] (182) at (3.625, -0.625) {};
		\node [style=none] (183) at (4.75, -0.625) {};
		\node [style=none] (184) at (3.625, 0.375) {};
		\node [style=none] (185) at (3.25, 0.375) {};
		\node [style=none] (186) at (3.25, -0.375) {};
		\node [style=none] (187) at (3.625, -0.375) {};
		\node [style=none] (188) at (5.25, 0.375) {};
		\node [style=none] (189) at (4.75, 0.375) {};
		\node [style=none] (190) at (4.75, -0.375) {};
		\node [style=none] (191) at (5.25, -0.375) {};
		\node [style=none] (192) at (4.25, 0) {$D_2$};
		\node [style=none] (193) at (3.25, 1) {};
		\node [style=none] (194) at (5.25, 1) {};
		\node [style=none] (195) at (3.25, -1) {};
		\node [style=none] (196) at (5.25, -1) {};
		\node [style=none] (197) at (6.25, -2) {};
		\node [style=none] (198) at (6.25, 2) {};
		\node [style=none] (199) at (2.25, -2) {};
		\node [style=none] (200) at (2.25, 1.25) {};
		\node [style=none] (201) at (2.25, -1.25) {};
		\node [style=none] (202) at (1.25, 1.25) {};
		\node [style=none] (203) at (1.25, -1.25) {};
		\node [style=none] (204) at (1.75, 0.2) {\vdots};
		\node [style=none] (205) at (7.25, 1.25) {};
		\node [style=none] (206) at (7.25, -1.25) {};
		\node [style=none] (207) at (6.25, 1.25) {};
		\node [style=none] (208) at (6.25, -1.25) {};
		\node [style=none] (209) at (6.75, 0.2) {\vdots};
		\node [style=none] (210) at (4.25, 1.5) {$D$};
		\node [style=none, scale=.7] (211) at (3.45, 0.125) {$\vdots$};
		\node [style=none, scale=.7] (212) at (5, 0.125) {$\vdots$};
		\node [style=none] (213) at (-16.875, 0.625) {};
		\node [style=none] (214) at (-15.75, 0.625) {};
		\node [style=none] (215) at (-16.875, -0.625) {};
		\node [style=none] (216) at (-15.75, -0.625) {};
		\node [style=none] (217) at (-16.875, 0.375) {};
		\node [style=none] (218) at (-17.25, 0.375) {};
		\node [style=none] (219) at (-17.25, -0.375) {};
		\node [style=none] (220) at (-16.875, -0.375) {};
		\node [style=none] (221) at (-15.25, 0.375) {};
		\node [style=none] (222) at (-15.75, 0.375) {};
		\node [style=none] (223) at (-15.75, -0.375) {};
		\node [style=none] (224) at (-15.25, -0.375) {};
		\node [style=none] (225) at (-16.25, 0) {$D_1$};
		\node [style=none, scale=.7] (227) at (-17.05, 0.125) {$\vdots$};
		\node [style=none, scale=.7] (228) at (-15.5, 0.125) {$\vdots$};
		\node [style=none] (229) at (-12.875, 0.625) {};
		\node [style=none] (230) at (-11.75, 0.625) {};
		\node [style=none] (231) at (-12.875, -0.625) {};
		\node [style=none] (232) at (-11.75, -0.625) {};
		\node [style=none] (233) at (-12.875, 0.375) {};
		\node [style=none] (234) at (-13.25, 0.375) {};
		\node [style=none] (235) at (-13.25, -0.375) {};
		\node [style=none] (236) at (-12.875, -0.375) {};
		\node [style=none] (237) at (-11.25, 0.375) {};
		\node [style=none] (238) at (-11.75, 0.375) {};
		\node [style=none] (239) at (-11.75, -0.375) {};
		\node [style=none] (240) at (-11.25, -0.375) {};
		\node [style=none] (241) at (-12.25, 0) {$D_2$};
		\node [style=none, scale=.7] (243) at (-13.05, 0.125) {$\vdots$};
		\node [style=none, scale=.7] (244) at (-11.5, 0.125) {$\vdots$};
	\end{pgfonlayer}
	\begin{pgfonlayer}{edgelayer}
		\draw (146.center)
			 to (147.center)
			 to (149.center)
			 to (148.center)
			 to cycle;
		\draw (151.center) to (150.center);
		\draw (152.center) to (153.center);
		\draw (155.center) to (154.center);
		\draw (156.center) to (157.center);
		\draw (164.center)
			 to (145.center)
			 to (165.center)
			 to (163.center)
			 to cycle;
		\draw (160.center)
			 to (162.center)
			 to (161.center)
			 to (159.center)
			 to cycle;
		\draw (168.center) to (166.center);
		\draw (169.center) to (167.center);
		\draw (173.center) to (171.center);
		\draw (174.center) to (172.center);
		\draw (180.center)
			 to (181.center)
			 to (183.center)
			 to (182.center)
			 to cycle;
		\draw (185.center) to (184.center);
		\draw (186.center) to (187.center);
		\draw (189.center) to (188.center);
		\draw (190.center) to (191.center);
		\draw (198.center)
			 to (179.center)
			 to (199.center)
			 to (197.center)
			 to cycle;
		\draw (194.center)
			 to (196.center)
			 to (195.center)
			 to (193.center)
			 to cycle;
		\draw (202.center) to (200.center);
		\draw (203.center) to (201.center);
		\draw (207.center) to (205.center);
		\draw (208.center) to (206.center);
		\draw (213.center) to (214.center);
		\draw (214.center) to (216.center);
		\draw (216.center) to (215.center);
		\draw (215.center) to (213.center);
		\draw (218.center) to (217.center);
		\draw (219.center) to (220.center);
		\draw (222.center) to (221.center);
		\draw (223.center) to (224.center);
		\draw (229.center) to (230.center);
		\draw (230.center) to (232.center);
		\draw (232.center) to (231.center);
		\draw (231.center) to (229.center);
		\draw (234.center) to (233.center);
		\draw (235.center) to (236.center);
		\draw (238.center) to (237.center);
		\draw (239.center) to (240.center);
	\end{pgfonlayer}
\end{tikzpicture}
\]

One last important \enquote{meta-rule} of the ZX-calculus is called \emph{Only Connectivity Matters \hypertarget{eq:OCM}{(OCM)}}.
This rule allows us to treat ZX diagrams based solely on their connectivity and disregard relative positions, enabling us to treat ZX diagrams as open graphs.
In practice, this allows us to rearrange ZX diagrams without changing their interpretation, as long as we preserve the connectivity.

\subsection{ZX Errors and ZX Distance}
\label{sec:zx-errors}
Earlier, we defined errors on measurement circuits as Paulis that originate from qubit flips and measurement flips. 
We can similarly define a notion of errors on ZX diagrams. 
As we will later use the ZX diagrams to reason about the measurement circuits of Floquet codes, Pauli errors on these circuits correspond to adding $\pi$-phased spiders to edges of the diagram.
We have:
\begin{definition}[Edge flip]
 Let $D = (G, \alpha, \tau)$ be a ZX diagram with $G=(V,E)$, an edge flip in $D$ is a tuple $(e, t)$ where
    \begin{itemize}
        \item $e \in E$ is an edge in $D$ and
        \item $t \in \{X, Y, Z\}$ is the error type.
    \end{itemize}
 We write $D^{\{(e, t)\}}$ to denote the diagram with edge $e$ flipped with type $t$.
\end{definition}

\noindent For example, we can apply an edge flip ((6,7), X) to the following ZX diagram (annotated with vertex IDs):
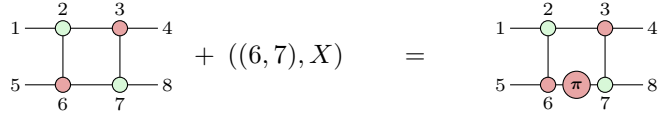
\[
 \input{prerequisites/error-example-base.tikz}^{\{(6,7), X\}}
 \qquad = \qquad
 \input{prerequisites/edge-flip-example.tikz}
\]

But then, similar to how we defined the set of errors as all possible combinations of atomic errors, we now define:
\begin{definition}[Error]
 Given a ZX diagram $D$, an error ${E}$ in $D$ is a set of edge flips ${E} = \{(e_1, t_1), \dots, (e_n, t_n)\}$, each in $D$.
 The weight of the error is the number of edge flips in ${E}$, i.e.\@ $wt({E}) = |E|$.
 We write $D^{E}$ to denote the diagram with the error.
 We say that two errors ${E}_1$ and ${E}_2$ in $D$ are \textit{equivalent} when $\interp{D^{E}_1} = \interp{D^{E}_2}$.
 And an error ${E}$ in $D$ is \textit{trivial} if $\interp{D^{E}} = \interp{D}$.
 An error is \emph{detectable} if $\interp{D^E} = 0$.
\end{definition}

We defined the distance of a Floquet code to be the minimum weight of all non-trivial, undetectable errors. 
Similarly, for ZX diagrams, we define:
\begin{definition}[ZX distance]
  \label{def:zx-dist}
 Given a ZX diagram $D$, the ZX distance of $D$ is the minimum weight of all non-trivial, undetectable errors.
\end{definition}
This notion of distance for ZX diagrams is independent of how the ZX diagrams were created. 
Throughout this work, we will consider rewrites that preserve this ZX distance. 

However, in the context of Floquet codes, there is a second notion of distance we are interested in.
Naively, a ZX diagram corresponding to the measurement circuit of a Floquet code will have ZX distance 1 --- any edge flip on an input edge will be non-trivial and undetectable. 
Therefore, to reason about measurement circuits of Floquet codes, we define: 
\begin{definition}[ZX distance after establishement]
  \label{def:zx-dist-establishement}
  Let $\mathcal M$ be a Floquet code with measurement circuit $C_{\mathcal M}$. 
  Let $D_{\mathcal M}$ be the ZX diagram obtained by translating $C_{\mathcal M}$ according to \autoref{fig:circuit-to-zx}.
  The ZX distance of $D_{\mathcal{M}}$ after establishment is defined as the minimum weight of all non-trivial, undetectable errors that only act on edges in $D_{\mathcal{M}}$ that correspond to measurements after $\mathcal M$ established. 
\end{definition}
The fact that we treat the establishment phase as error-free formally corresponds to assuming that we have performed a fault-tolerant state preparation before feeding it to the Floquet code.

We observe: 
\begin{proposition}
  \label{prop:distance-correspondence}
  Let $\mathcal M$ be a Floquet code and $D_{\mathcal M}$ be the ZX diagram corresponding to $\mathcal M$'s measurement circuit. 
  Then, the ZX distance of $D_{\mathcal M}$ after establishment is the same as the distance of $\mathcal M$.
\end{proposition}
\begin{proof}
  We show that for any error on $C_{\mathcal M}$, there exists an error of the same weight on $D_{\mathcal M}$ that has the same effect, and vice versa. 
 On $C_{\mathcal M}$, the possible atomic errors are qubit flips or measurement flips, where the latter may simultaneously induce a qubit flip. 
 Qubit flips correspond directly to edge flips on the corresponding qubit edges. 
 Measurement flips can be modelled as edge flips on one of the edges adjacent to the corresponding parameterised measurement spider. 
 A measurement flip together with a $Z$-type qubit flip on, say, the last qubit of the measurement, corresponds to a $Y$-type edge flip on the corresponding measurement edge:
  \[
 \input{fault-equivalence-zx-circuit.tikz}
  \]
 Similarly, a measurement flip combined with an $X$- or $Y$-type qubit flip corresponds to a qubit flip occurring before the measurement. 
 Thus, all atomic errors in $C_{\mathcal M}$ have corresponding edge flips in $D_{\mathcal M}$.

 Conversely, all edge flips correspond to atomic errors on $C_{\mathcal M}$.
 Edge flips on edges that correspond to qubits map directly to the respective qubit flips.
 An $X$-type edge flip on a measurement edge corresponds to a measurement flip.
 A $Z$-type edge flip can be pushed onto an adjacent qubit edge and thus corresponds to a $Z$-type qubit flip at the corresponding location.
 A $Y$-type edge flip on a measurement edge corresponds to a measurement flip together with a $Z$ qubit flip on that edge.
  Thus, we have shown that all errors on $C_{\mathcal M}$ correspond to edge flips, and vice versa.
 But then, any undetectable, non-trivial errors on $C_{\mathcal M}$ have corresponding undetectable, non-trivial errors on $D_{\mathcal M}$ of the same weight, and vice versa.
 Therefore, the minimum weight of all undetectable, non-trivial errors on each must be the same, and thus they have the same distance. 
\end{proof}
\subsection{Pauli webs}
\label{sec:pauli-webs}
Pauli webs~\parencite{bombinUnifyingFlavorsFault2024} are a decoration for the Clifford fragment of the ZX calculus.
They can be used to diagrammatically track relevant properties of ZX diagrams representing measurement circuits of Floquet codes, including detecting regions, stabilisers and logical operators.

A Pauli web highlights the edges of a ZX diagram in green and/or red (Z and/or X) according to a set of simple rules.
\begin{definition}[Pauli web]
 Given a ZX diagram $D$, a Pauli web $\{(e_1, t_1), \dots, (e_n, t_n)\}$ is a set of pairs containing an edge $e_i \in D$ and the corresponding type $t_i \in \{Z, X\}$, such that:
\begin{itemize}
    \item a spider with $k \pi$ phase can have:
    \begin{itemize}
      \item an even number of legs highlighted in its own colour and all or none of its legs highlighted in the opposite colour
    \end{itemize}
    \item a spider with $\pm \frac{\pi}{2}$ phase can have \textbf{either}:
    \begin{itemize}
      \item an even number of legs highlighted in its own colour and no legs highlighted in the opposite colour, \textbf{or}
      \item an odd number of legs highlighted in its own colour and all legs highlighted in the opposite colour
    \end{itemize}
\end{itemize}
\end{definition}
\noindent Examples of Pauli webs on diagrams with one spider include:
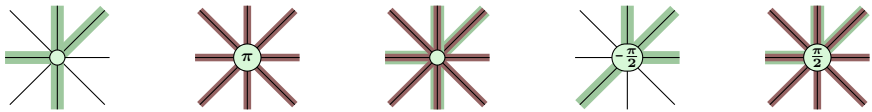
\[
 \input{prerequisites/pauli-web-examples.tikz}
\]

Locally, restricted to the boundary edges of a single spider, the conditions of a Pauli web guarantee that the Pauli web corresponds to a stabiliser of that spider. 
This means that Pauli webs can be understood as tracking how spiders can `fire' $\pi$ phases without changing the semantics of a diagram~\parencite{borghansZXcalculusQuantumStabilizer2019}.
That is, for any spider covered by a Pauli web, we can introduce $\pi$ spiders matching the colour of each highlighted edge without changing the underlying linear map.
For example, we have:
\[\input{prerequisites/firing-examples.tikz}\]
We define the term \enquote{firing a spider $s$ according to a Pauli web $w$} to mean the action of introducing $\pi$ spiders of type $t$ on all edges $((s, s'), t) \in w$ that involve $s$.
For clarity, we draw the spiders that are being fired according to the Pauli web with an orange boundary.

Following~\textcite{townsend-teagueFloquetifyingColourCode2023}, we differentiate detecting regions, stabilising Pauli webs and logical Pauli webs, depending on the types of the covered edges:

\begin{definition}[Detecting region]
 Given a ZX diagram $D$, a detecting region is a Pauli web that highlights no boundary edges. 
\end{definition}

\noindent Detecting regions allow us to reason about which errors are detectable, i.e.\@ which errors make the diagram go to 0.
More precisely, if a detecting region highlights an edge in green, edge flips of type $X$ and $Y$ on that edge are detected.
Similarly, if the edge is highlighted in red, edge flips of type $Y$ and $Z$ are detected, and if it is highlighted in red and green, errors of type $X$ and $Z$ are detected. 
That is, edge flips that are different from the type of the highlighting are detected.
\begin{definition}[Overlap between an error and a detecting region]
 Let $E = \{(e_1, t_1), \dots, (e_n, t_n)\}$ be an error and $w = \{(e'_1, t'_1), \dots, (e'_m, t'_m)\}$ be a detecting region on some diagram $D$.
 The overlap of $E$ and $w$ is defined as
  $\{(e, t) | (e, t) \in E \land (e, t') \in w \text{ such that } t \not= t'\}$.
\end{definition}

\begin{theorem}
  \label{thm:detecting-region}
 Let $E$ be an error and $w$ be a detecting region on some diagram $D$ such that $\interp{D} \neq 0$.
 If $E$ has an odd overlap with $w$, then $\interp{D^E} = 0$.
 We say that $E$ is detectable by $w$.
\end{theorem}
\noindent The proof is similar to the one given by~\textcite[Lemma 5]{fuenteXYZRuby2024}, here adapted for Pauli webs on ZX diagrams.
\begin{proof}
  Note that this proof relies on global scalars. 
  Therefore, in contrast to the rest of the paper, for this proof, we will use the scalar-accurate version of the ZX calculus.\\ 
 Let $D$ be a ZX diagram with a detecting region $w$ and an error $E$.
 First, we fire all the spiders in $D$ according to $w$, i.e.\@ without the error.
 We observe that this might change the global phase of $D$ by some factor $c$.
 Every edge in $D$ now either has no $\pi$ spiders (if it is not in $w$) or two $\pi$ spiders of the type determined by $w$.
 We can remove these spiders without affecting the global phase using \TextFusion, \TextXElim and \TextZElim, bringing us back to a diagram of the original shape.
 This means that $\interp{D} = \interp{c * D}$, implying that $c = 1$ in the error-free case, as we assumed $\interp{D} \neq 0$.

 Now, we repeat the procedure on $D^E$, i.e.\@ the diagram that includes the errors.
 Once again, we fire the spiders of $D$ in $D^E$ according to $w$. 
 So we do not fire the spiders that were introduced by $E$.
 As we have just shown, this does not affect the global phase.
 However, for an edge that has an error on it, $-1$ scalars are introduced when we commute the fired spiders through the errors:
  \[
 \input{prerequisites/detecting-region-proof.tikz}
  \]
 Once the fired spiders are commuted past the errors, they cancel each other out to give back the original diagram (up to the global phase).
 If we have an odd overlap, we commuted an odd number of fired spiders past anti-commuting errors, and therefore we have $\interp{D^E} = \interp{-(D^E)}$.
 This implies that $\interp{D^E} = 0$, concluding the proof.
\end{proof}

As an example, we can show that two consecutive weight-two Pauli Z measurements form a detecting region.
That is, in the error-free case, the two measurements should give the same outcome; however, if an error occurs, the following happens:
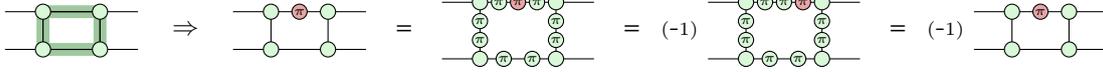
\[
 \input{prerequisites/detecting-region-example.tikz}
\]
This means that the probability of the above diagram is zero, proving that we cannot get the same measurement outcome when a single error happens between them.

We can also use Pauli webs to track stabilisers and logicals.
\begin{definition}[(Co-) stabilising Pauli web]
 Given a ZX diagram $D$, a  Pauli web that highlights no input edges is a stabilising Pauli web, and one that highlights no output edges is a co-stabilising Pauli web. 
\end{definition}

\noindent For example, for a stabilising Pauli web, we have:
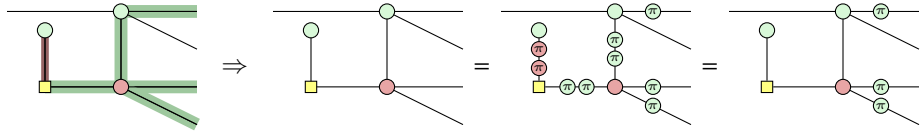
\[\input{prerequisites/stabiliser-example.tikz}\]
In the first step, we fire all spiders in the diagram according to the web.
All internal edges in the diagram have either no (if they are not highlighted) or two spiders with a $\pi$ phase introduced on them.
Therefore, in the second step, the spiders on the internal edges cancel out.
Only the introduced spiders on the boundary edges remain.
As such, we have shown that this circuit stabilises the corresponding Pauli $Z_1Z_3Z_4$.

Given a ZX diagram for a quantum circuit, there is an exact correspondence between stabilising Pauli webs and the stabilisers of the circuit as tracked using group theory. For more details, see \autoref{prop:webBijectiveStabs}.

\noindent Similarly, we can use Pauli webs to track how logical operators propagate through a circuit:
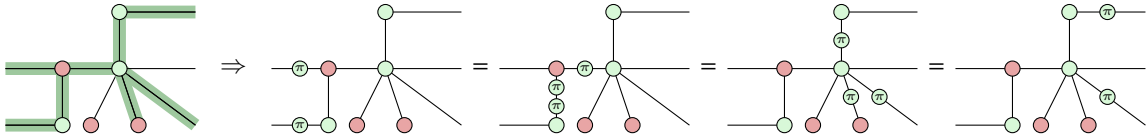
\[\input{prerequisites/logical-example.tikz}\]
where we push the $\pi$ spiders through the diagram by firing the spiders according to the web.

\noindent To identify Pauli webs corresponding to non-trivial logical operators, we can define:
\begin{definition}[Logical Pauli webs]
 Given a ZX diagram $D$, a (non-trivial) logical Pauli web is a Pauli web that is not a combination of stabilising and co-stabilising Pauli webs.
\end{definition}
This definition guarantees that logical Pauli webs highlight input and output edges, as usually required of logical Pauli webs~\parencite{townsend-teagueFloquetifyingColourCode2023}.

Once again, the relationship between logical Pauli webs and logicals of the circuit, i.e.\@ $N(S_t) / S_t$ for some instantaneous stabilising group $S_t$, can be formalised; see \autoref{prop:logicals}.
As we mostly care about ZX diagrams up to global phase, except for when proving the properties of Pauli webs, from now on, we will omit global phases.
\section{Distance-preserving rewrites}
\label{sec:distance-preservation}
\subsection{Definition}
The ZX calculus has a complete set of rewrite rules which preserve the semantics of the diagrams.
However, these rules can change other attributes, including the ZX distance.
As the distance of a ZX diagram corresponds to the number of errors the diagram can detect, reducing the distance is undesirable in the context of quantum error correction.
Therefore, we propose a new notion we call \textit{distance-preserving rewrites}.
These are a restriction of the allowed rewrites to the ones that are guaranteed to preserve the distance of the ZX diagram.
We present a family of distance-preserving rewrites, which are sufficient to rewrite any Pauli measurement in terms of spiders of degree at most three.

Before providing a formal definition of distance-preservation, let us consider the rewrite~\eqref{naive-rewrite} that is not distance-preserving: the naive implementation of a four-qubit Pauli-Z measurement in terms of one auxiliary qubit and four CNOTs.
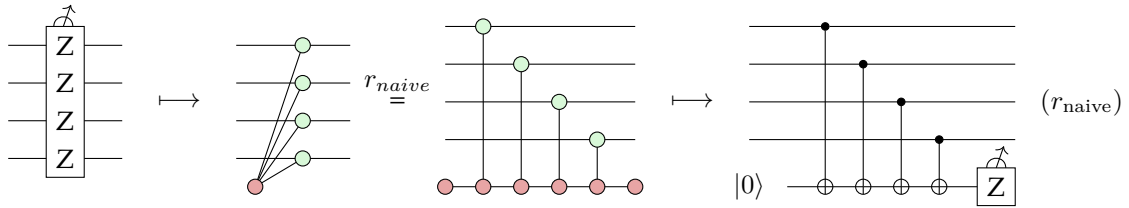
\begin{gather}
\tag{$r_{\text{naive}}$}\label{naive-rewrite}\refstepcounter{equation}
\input{figures/rewrites/naive-rewrite.tikz}
\end{gather}
\noindent This is a sound rewrite, consisting of a set of spider unfusions. Therefore, we are guaranteed that the two circuits implement the same linear map.
However, it decreases the distance.
A weight-one error in the new diagram propagates through the circuit to become a weight-two error in the original circuit:
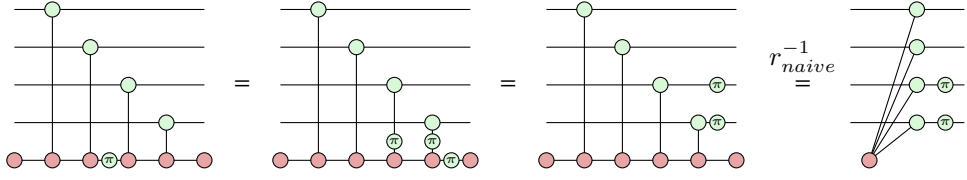
\[\input{figures/rewrites/naive-rewrite-error.tikz}\]
\noindent Thus, if we have a distance-two code, for example, the $\code{4, 2, 2}$ code, yet use this implementation of the measurement, there exist single errors that are not detectable — we have decreased the error-detecting capability of our code.

To prevent this, we restrict the set of allowed ZX rewrites:

\begin{definition}[Distance preserving]
Let $D_1, D_2$ be diagrams that are semantically equivalent, i.e.\@ $\interp{D_1} = \interp{D_2}$.
Then we say the rewrite $r: D_1 \to D_2$ is \emph{distance-preserving} if for all diagrams $D$ we have:
\[\input{rewrites/dist-preservation.tikz}\]
\end{definition}

Distance preservation states that if we replace $D_1$ with $D_2$, no matter in what context, the distance of the overall diagram is preserved.
This means that observing that two diagrams, $D_1$ and $D_2$, can distance-preservingly be rewritten into one another is a non-trivial statement, even if they have distance 1.
This is due to the fact that we can consider a larger diagram $D$ in which $D_1$ occurs that has a higher distance $d > 1$.
Now, rewriting $D_1$ into $D_2$ in this context, we are guaranteed that the resulting diagram will still have distance $d$.
In turn, this means that, as distance preservation quantifies over all possible contexts, formulated like this, it is difficult to verify.
Therefore, we observe the following:

\begin{proposition}
\label{prop:dist-preservation}
A semantic preserving rewrite $r : D_1 \to D_2$ is distance-preserving if, for any error $E_1$ in $D_1$, either
\begin{itemize}
\item $E_1$ is detectable in $D_1$, or
\item there exists an error $E_2$ in $D_2$ such that $|E_2| \leq |E_1|$ and $\interp{D_1^{E_1}} = \interp{D_2^{E_2}}$.
\end{itemize}
Similarly, for any error $E_2$ in $D_2$ the same condition holds.
\end{proposition}
\begin{proof}
Let $D$ be some context in which $D_1$ occurs.
Let $E_1$ be a non-trivial, non-detectable error in $D$ with $D_1$.
 Then we consider the error $E'_1$, which is the restriction of $E_1$ to internal edges of $D_1$.
 By assumption, there must exist some error $E'_2$ in $D_2$ with $|E'_2| \leq |E'_1|$ such that $\interp{D_1^{E'_1}} = \interp{D_2^{E'_2}}$. 
 But then we can construct a non-trivial, non-detectable error $E_2$ in $D$ with $D_2$ by replacing $E'_1$ with $E'_2$. We have $|E_2| \leq |E_1|$. As this holds for all non-detectable errors in $D$ with $D_1$, we must have that the distance of $D$ with $D_1$ is at least the distance of $D$ with $D_2$.
The same argument holds in the other direction; therefore, for all $D$ the distance of $D$ with $D_1$ must be the same as the distance of $D$ with $D_2$.
\end{proof}

\begin{remark}
In follow-up work \autocite{rodatzFaultTolerance2025}, we focus on this sufficient condition for distance preservation and call it \textit{fault-equivalent rewrites}.
In \autoref{appendix:fe-equals-dist-pres}, we show that the two definitions are, in fact, equivalent.
However, they have very different intuitions: in this paper, we leverage the distance-preserving properties of fault equivalence to manipulate quantum error correction protocols.
In \autocite{rodatzFaultTolerance2025}, we use fault equivalence to reason about fault-tolerant circuits, observing that fault-tolerant gadgets, such as state preparation or syndrome extraction, are correct if and only if they are fault equivalent to some idealised specification.
\end{remark}

To further simplify checking of distance preservation, we propose another sufficient condition for distance preservation that is even easier to check.
Instead of checking that all undetectable errors have an equivalent error on the other diagram, we only have to check that all undetectable errors can be pushed to the boundary of the diagram without increasing in size.
This will immediately give us a corresponding error in the other diagram:
\begin{proposition}
\label{prop:pushing-out}
A semantic preserving rewrite $r : D_1 \to D_2$ is distance-preserving if, for any error $E_1$ in $D_1$, either
\begin{itemize}
\item $E_1$ is detectable in $D_1$, or
        \item there exists an error $E_1'$ in $D_1$ that only acts on the boundary edges of $D_1$ such that $|E_1'| \leq |E_1|$ and $\interp{D_1^{E_1}} = \interp{D_1^{E_1'}}$.
\end{itemize}
Similarly, for any error $E_2$ in $D_2$ the same condition holds.
\end{proposition}
\begin{proof}
Let $r: D_1 \to D_2$ satisfy the condition above that any undetectable error in $D_1$ can be pushed to the boundary without increasing in size.
 Then, for any error $E_1$ in $D_1$, we can find an equivalent error $E_1'$ of at most the same weight on the boundary. 
 But this gives us an error $E_2$ on $D_2$ by simply taking an error that acts on exactly the same boundary edges as $E_1'$ on $D_1$. 
 This error satisfies $|E_2| = |E_1'| \leq |E_1|$ and $\interp{D_2^{E_2}} = \interp{D_1^{E_1'}} = \interp{D_1^{E_1}}$.
Analogously, we can find equivalent errors for all errors $E_2$ in $D_2$.
But then, by \autoref{prop:dist-preservation}, $r$ is distance-preserving.
\end{proof}

\begin{corollary}
\label{corr:internal-edges}
Let $r: D_1 \to D_2$ be a semantics-preserving rewrite, where neither $D_1$ nor $D_2$ have internal edges.
Then $r$ is distance-preserving.
\end{corollary}
\begin{proof}
As neither $D_1$ nor $D_2$ have internal edges, all errors can trivially be pushed to the boundary.
Therefore, by \autoref{prop:pushing-out}, they must be distance-preserving.
\end{proof}
\subsection{Basic distance-preserving rewrites}
To explore whether a rewrite $r: D_1 \to D_2$ is distance-preserving, we consider all possible errors in $D_2$ and check whether they are either detectable or have an equivalent error of at most the same weight in $D_1$ and vice versa.
By \autoref{prop:pushing-out}, we can simplify this process by showing that all errors can be pushed to the boundary without increasing in size, in which case they naturally have a correspondence on the other diagram:

\begin{theorem}
\label{thm:elim-rewrite}
The following rewrite is distance-preserving:
\begin{gather}
\tag{$r_{\text{elim}}$}\label{elim-rewrite}\refstepcounter{equation}
\input{figures/rewrites/elim-rewrite.tikz}
\end{gather}
\end{theorem}
\begin{proof}
Neither diagram has internal edges; therefore, by \autoref{corr:internal-edges}, the rewrite is distance-preserving.
\end{proof}

\noindent To see \autoref{prop:pushing-out} in action, consider the following error on $D_2$:
\[\input{figures/rewrites/elim-rewrite-proof.tikz}\]
Here, we have a single X flip on the right boundary leg of the spider.
We can undo the rewrite $r_{elim}: D_1 \to D_2$ to obtain an equivalent error on $D_1$, fulfilling our distance-preservation condition.
However, in order to undo $r$, no errors can be inside $D_2$.
Therefore, we first have to push them out.

\begin{theorem}
\label{thm:fuse-rewrite}
The following rewrite of any green spider with more than one edge is distance-preserving:
\begin{gather}
\tag{$r_{\text{fuse}}$}\label{fuse-rewrite}\refstepcounter{equation}
\input{figures/rewrites/fuse-rewrite.tikz}
\end{gather}
\end{theorem}
\begin{proof}
The diagram on the RHS only has one internal edge.
Therefore, we have to consider three possible errors on this edge: an X flip, a Y flip and a Z flip.
We have:
  \[\input{figures/rewrites/fuse-rewrite-proof.tikz}\]
The X error is trivial as it can be absorbed by the newly introduced spider using \TextPiCommute.
The Y error consists of an X error and a Z error.
Due to \TextPiCommute and since we ignore global phases, it does not matter in which order we add them to the edge.
The X flip can be absorbed like the X error, and the Z flip can be pushed out without increasing in weight using \TextFusion.
The last error can similarly be pushed out without increasing in weight.
As above, after having pushed out the errors, we can undo the rewrite to obtain an equivalent error on the original diagram with at most the same weight.
Therefore, in this direction, the rewrite satisfies the sufficient condition of \autoref{prop:pushing-out}.

For the inverse, we observe that the LHS has no internal edges; therefore, the inverse rewrite also satisfies \autoref{prop:pushing-out}, and so the rewrite is distance-preserving.
\end{proof}

Next, we introduce a rewrite to express four-legged spiders in terms of three-legged spiders.

\begin{theorem}
\label{thm:five-legged-spider}
The following rewrite for the five-legged spider is distance-preserving:
\begin{gather}
\tag{$r_{5}$}\label{five-legged}\refstepcounter{equation}
\input{figures/rewrites/five-legged.tikz}
\end{gather}
\end{theorem}
\begin{proof}
To prove that~\eqref{five-legged} is distance-preserving, we first consider errors with only X flips and errors with only Z flips before considering more general errors.

 For X flips, we observe that the diagram has one green detecting region that can detect an odd number of flips of type X (i.e.\@ the opposite colour of the detecting region).
  \[\input{figures/rewrites/five-legged-detecting-region.tikz}\]
\noindent Therefore, we only have to consider errors with an even number of type X flips.
Up to symmetry, these are:
  \[\input{figures/rewrites/five-legged-X-flips.tikz}\]
That is, errors with an even number of flips of type X can be pushed to the boundary without increasing their weight.

For Z flips, we observe that the diagram only consists of Z spiders.
Therefore, the flips can be pushed around freely, cancelling each other out.
Eventually, we end with at most one Z flip, which we can push to any boundary edge.

For errors of weight $w$ that consist of X, Y, and Z flips, we can decompose the Y flips into X flips and Z flips.
As the number of X, Y, and Z flips initially totals $w$, after decomposition, we have at most $w$ X flips.
We have shown that these can be pushed out without increasing in weight.
The remaining Z flips can first be combined to leave at most one flip, which can be pushed to an outer edge that already has an X flip, not increasing the weight of the errors on the outer edges.
Therefore, the error on the outer edges is at most weight $w$.
This way, any non-detectable combination of flips can be pushed out without increasing their weight.
As the LHS has no internal edges, the rewrite preserves distance.
\end{proof}

\begin{theorem}
\label{thm:four-legged-spider}
The following rewrite for the four-legged spider is distance-preserving:
\begin{gather}
\tag{$r_{4}$}\label{four-legged}\refstepcounter{equation}
\input{figures/rewrites/four-legged.tikz}
\end{gather}
\end{theorem}
\begin{proof}
The proof that~\eqref{four-legged} is distance-preserving follows from the following derivation that uses solely distance-preserving rewrites:
  \[
\input{rewrites/four-legged-proof.tikz}
  \]
\end{proof}

\subsection{A family of distance-preserving rewrites}
Next, we utilise the implementation of four-legged spiders to create distance-preserving rewrites of arbitrary even-legged spiders.

While the intuition behind the proofs remains unchanged, we present a more general approach that is based on linear algebra over $\mathbb F_2$.
This is similar to~\textcite{derksDesigningFaulttolerant2024}, which employs this approach to study fault-tolerant circuits at various levels of abstraction.

\begin{definition}[Edge-flip-vector for ZX diagrams]
 Given a ZX diagram $D$ with internal edges $e_1, \dots, e_n$, an error vector $\vec{v} \in (\mathbb{F}_2)^{2n}$ for an error $E$ on $D$ is defined as:
    \[
\vec{v}[i] =
\begin{dcases}
 1 & \text{if } (i \leq n \land (e_i, X) \in E) \lor (i > n \land (e_{i - n}, Z) \in E) \\
0 & \text{otherwise}
\end{dcases}
    \]
\end{definition}
\noindent The first $n$ entries of an error vector correspond to $X$ flips on the internal edges and the second $n$ entries correspond to $Z$ flips.
To capture whether an error is detectable, we furthermore define:

\begin{definition}[Detector error matrix]
Given a ZX diagram $D$ with $n$ internal edges and $m$ independent detecting regions $D_1, \dots, D_m$, the detector error matrix $P$ is an $m{\times}2n$-matrix over $\mathbb F_2$, where each row corresponds to a detection region.
We have:
  \[
P[i][j] =
\begin{dcases}
 1 & \text{if edge flip } e_i \text{ is detectable by detector } D_j \\
0 & \text{otherwise}
\end{dcases}
  \]
\end{definition}

\noindent By \autoref{thm:detecting-region}, we know edge flips are detectable if they have odd overlap with at least one of the detecting regions of the diagram, i.e.\@ if $P \vec{v}$ has at least one non-zero entry.
As such, to explore all undetectable edge flips of a distance-preserving rewrite, we have to consider the null space of $P$.
In particular, we give a basis for the null space and show that any linear combination of basis states can be pushed out without increasing their weight.
\begin{restatable}{theorem}{inductivePlusDistRewrite}
\label{thm:inductive-plus-dist-rewrite}
The following rewrites for $(2n \plus 1)$-legged spiders is distance-preserving:
\begin{gather}
\tag{$r_{2n^+}$}\label{n-plus-legged}\refstepcounter{equation}
\begin{tikzpicture}
	\begin{pgfonlayer}{nodelayer}
		\node [style=Z dot] (0) at (3.75, 2.5) {};
		\node [style=Z dot] (1) at (3.75, 1.5) {};
		\node [style=Z dot] (2) at (3.75, -1.5) {};
		\node [style=Z dot] (3) at (3.75, -2.5) {};
		\node [style=Z dot] (4) at (1.75, 1.75) {};
		\node [style=Z dot] (5) at (1.75, -1.75) {};
		\node [style=none] (6) at (4.5, 2.75) {};
		\node [style=none] (7) at (4.5, 2.25) {};
		\node [style=none] (8) at (4.5, 1.75) {};
		\node [style=none] (9) at (4.5, 1.25) {};
		\node [style=none] (10) at (4.5, -1.25) {};
		\node [style=none] (11) at (4.5, -1.75) {};
		\node [style=none] (12) at (4.5, -2.25) {};
		\node [style=none] (13) at (4.5, -2.75) {};
		\node [style=none] (14) at (3.75, 0.625) {};
		\node [style=none] (15) at (3.75, -0.375) {};
		\node [style=none] (16) at (3.75, 0.375) {};
		\node [style=none] (17) at (3.75, -0.625) {};
		\node [style=none] (18) at (4.25, 0.2) {$\vdots$};
		\node [style=Z dot] (19) at (-3.75, 0) {};
		\node [style=none] (20) at (-2.75, 1) {};
		\node [style=none] (21) at (-2.75, 0.5) {};
		\node [style=none] (22) at (-2.75, -1) {};
		\node [style=none] (23) at (-2.75, 0.025) {$\vdots$};
		\node [style=none] (24) at (-2.25, 1) {};
		\node [style=none] (25) at (-2.25, -1) {};
		\node [style=none] (26) at (-1.5, 0) {$2n$};
		\node [style=none] (27) at (0, 0) {=};
		\node [style=none] (28) at (0, 0.75) {$r_{2n^+}$};
		\node [style=none] (29) at (-3.75, -1.5) {};
		\node [style=none] (30) at (1.75, -3) {};
	\end{pgfonlayer}
	\begin{pgfonlayer}{edgelayer}
		\draw (4) to (0);
		\draw (0) to (5);
		\draw (1) to (4);
		\draw (1) to (5);
		\draw (2) to (4);
		\draw (2) to (5);
		\draw (3) to (4);
		\draw (5) to (3);
		\draw (13.center) to (3);
		\draw (12.center) to (3);
		\draw (11.center) to (2);
		\draw (10.center) to (2);
		\draw (9.center) to (1);
		\draw (8.center) to (1);
		\draw (7.center) to (0);
		\draw (6.center) to (0);
		\draw (4) to (14.center);
		\draw (4) to (15.center);
		\draw (5) to (17.center);
		\draw (5) to (16.center);
		\draw (19) to (20.center);
		\draw (21.center) to (19);
		\draw (19) to (22.center);
		\draw [style=braceedge] (24.center) to (25.center);
		\draw (29.center) to (19);
		\draw (30.center) to (5);
	\end{pgfonlayer}
\end{tikzpicture}

\end{gather}
where we construct a $(2n \plus 1)$-legged spider out of an $n$-legged spider, an $(n + 1)$-legged spider, and $n$ four-legged spiders.
\end{restatable}
\begin{proof}
See \autoref{appendix:distance-preserving-rewrites}.
\end{proof}

\begin{restatable}{theorem}{inductiveDistRewrite}
\label{thm:inductive-dist-rewrite}
The following rewrite for a $2n$-legged spider is distance-preserving:
\begin{gather}
\tag{$r_{2n}$}\label{n-legged}\refstepcounter{equation}
\begin{tikzpicture}
	\begin{pgfonlayer}{nodelayer}
		\node [style=Z dot] (0) at (2, 2.5) {};
		\node [style=Z dot] (1) at (2, 1.5) {};
		\node [style=Z dot] (2) at (2, -1.5) {};
		\node [style=Z dot] (3) at (2, -2.5) {};
		\node [style=Z dot] (4) at (0, 1.75) {};
		\node [style=Z dot] (5) at (0, -1.75) {};
		\node [style=none] (10) at (2.75, 2.75) {};
		\node [style=none] (11) at (2.75, 2.25) {};
		\node [style=none] (12) at (2.75, 1.75) {};
		\node [style=none] (13) at (2.75, 1.25) {};
		\node [style=none] (14) at (2.75, -1.25) {};
		\node [style=none] (15) at (2.75, -1.75) {};
		\node [style=none] (16) at (2.75, -2.25) {};
		\node [style=none] (17) at (2.75, -2.75) {};
		\node [style=none] (26) at (2, 0.625) {};
		\node [style=none] (27) at (2, -0.375) {};
		\node [style=none] (30) at (2, 0.375) {};
		\node [style=none] (31) at (2, -0.625) {};
		\node [style=none] (34) at (2.5, 0.2) {$\vdots$};
		\node [style=Z dot] (35) at (-5.5, 0) {};
		\node [style=none] (36) at (-4.5, 1) {};
		\node [style=none] (37) at (-4.5, 0.5) {};
		\node [style=none] (38) at (-4.5, -1) {};
		\node [style=none] (39) at (-4.5, 0.025) {$\vdots$};
		\node [style=none] (55) at (-4, 1) {};
		\node [style=none] (56) at (-4, -1) {};
		\node [style=none] (60) at (-3.25, 0) {$2n$};
		\node [style=none] (61) at (-1.75, 0) {=};
		\node [style=none] (62) at (-1.75, 0.75) {$r_{2n}$};
	\end{pgfonlayer}
	\begin{pgfonlayer}{edgelayer}
		\draw (4) to (0);
		\draw (0) to (5);
		\draw (1) to (4);
		\draw (1) to (5);
		\draw (2) to (4);
		\draw (2) to (5);
		\draw (3) to (4);
		\draw (5) to (3);
		\draw (17.center) to (3);
		\draw (16.center) to (3);
		\draw (15.center) to (2);
		\draw (14.center) to (2);
		\draw (13.center) to (1);
		\draw (12.center) to (1);
		\draw (11.center) to (0);
		\draw (10.center) to (0);
		\draw (4) to (26.center);
		\draw (4) to (27.center);
		\draw (5) to (31.center);
		\draw (5) to (30.center);
		\draw (35) to (36.center);
		\draw (37.center) to (35);
		\draw (35) to (38.center);
		\draw [style=braceedge] (55.center) to (56.center);
	\end{pgfonlayer}
\end{tikzpicture}

\end{gather}
where we construct a $2n$-legged spider out of two $n$-legged spiders and $n$ four-legged spiders.
\end{restatable}
\begin{proof}
Similar to the proof of \autoref{thm:four-legged-spider}, we unfuse a Z spider with~\eqref{fuse-rewrite}, apply~\eqref{n-plus-legged}, and fuse the spider back with~\eqref{fuse-rewrite}.
Since each of these rewrites is distance-preserving, so is the rewrite in question.
\end{proof}
\section{Measurement-circuit flow}
\label{sec:impl-interp}

To complete the Floquetification procedure, we formalise and extend the world lines drawn by~\textcite{townsend-teagueFloquetifyingColourCode2023}, defining a property for ZX diagrams we call \textit{measurement-circuit flow}.
In this section, we show that we can extract a quantum circuit from any diagram that has this property using only distance-preserving rewrites.

First, we define the family of quantum circuits that are of interest to us:
\begin{definition}[Measurement Circuit with local Cliffords form]
  \label{def:mclc-form}
 A circuit with state preparations, measurements, local Clifford gates, and Pauli measurements of weight one and two is generated from the following set of gates:
  
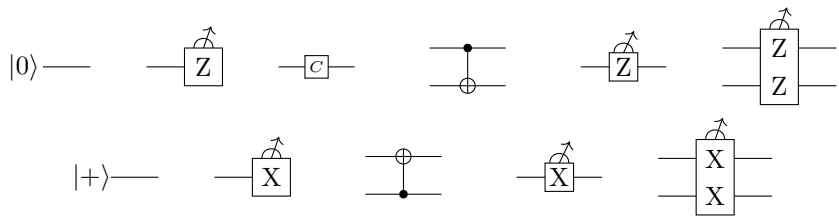
\begin{gather*}
 \begin{tikzpicture}
	\begin{pgfonlayer}{nodelayer}
		\node [style=none] (0) at (1.25, 0) {};
		\node [style=none] (1) at (0, 0) {};
		\node [style=none] (2) at (-0.5, 0) {$\ket{0}$};
	\end{pgfonlayer}
	\begin{pgfonlayer}{edgelayer}
		\draw (1.center) to (0.center);
	\end{pgfonlayer}
\end{tikzpicture}

 \qquad
 \input{QuantumCircuit/measurez.tikz}
 \qquad
 \begin{tikzpicture}
	\begin{pgfonlayer}{nodelayer}
		\node [style=none] (0) at (-1, 0) {};
		\node [style=none] (1) at (1, 0) {};
		\node [style=box] (2) at (0, 0) {$C$};
	\end{pgfonlayer}
	\begin{pgfonlayer}{edgelayer}
		\draw (0.center) to (2);
		\draw (2) to (1.center);
	\end{pgfonlayer}
\end{tikzpicture}

 \qquad
 \input{QuantumCircuit/z-parity-check.tikz}
 \qquad
 \input{QuantumCircuit/zz-parity-check.tikz}\\
 \begin{tikzpicture}
	\begin{pgfonlayer}{nodelayer}
		\node [style=none] (0) at (1.25, 0) {};
		\node [style=none] (1) at (0, 0) {};
		\node [style=none] (2) at (-0.5, 0) {$\ket{+}$};
	\end{pgfonlayer}
	\begin{pgfonlayer}{edgelayer}
		\draw (1.center) to (0.center);
	\end{pgfonlayer}
\end{tikzpicture}

 \qquad
 \input{QuantumCircuit/measurex.tikz}
 \qquad
 \input{QuantumCircuit/x-parity-check.tikz}
 \qquad
 \input{QuantumCircuit/xx-parity-check.tikz}
  \end{gather*}
  where $C$ is an arbitrary single-qubit Clifford unitary.
 We say a quantum circuit is in \emph{Measurement Circuit with local Cliffords (MCLC) form} if it is generated from the gate set above.
\end{definition}

\noindent The first row of gates above respectively corresponds to the following ZX diagrams:
\begin{gather}
 \input{flow/flow-examples.tikz}
  \label{eq:flow-examples}\\
 \input{flow/flow-examples-2.tikz}
\end{gather}
where $C$ denotes an arbitrary single-qubit Clifford unitary and measurements are post-selected on the +1 outcome.
The second row corresponds to the same diagrams with the colours reversed.

\begin{definition}[Measurement Circuit with Cliffords]
  \label{def:mcc-form}
 If a circuit only includes gates of \autoref{def:mclc-form} as well as CNOT gates, we say it is in \emph{Measurement Circuit with Cliffords (MCC) form}.
\end{definition}

We now define a notion of flow for ZX diagrams that guarantees that for any ZX diagram, we can extract a quantum circuit in MCLC or MCC form using only distance-preserving rewrites.

Flow is usually defined as a pair of a partial order and a successor function that satisfies certain conditions.
However, this is not adequate for our purposes, as we are restricted to distance-preserving operations.
To address this issue, we modify an equivalent definition of causal flow that uses vertex-disjoint paths~\parencite{debeaudrapFindingFlows2008}.
But before that, let us define:

\begin{definition}[Directed path]
 Let $G = (V, E)$ be a graph.
 A \emph{directed path} $\mathcal{P}$ is a sequence of directed edges $(e_1, \dots, e_{n-1})$ which joins a sequence of vertices $(v_1, \dots, v_n)$ such that all vertices are distinct and $(v_i, v_{i + 1}) \in E$.
 For vertices $v, w \in V$ or an edge $e \in E$, we write $v \in \mathcal{P}$ and $e \in \mathcal{P}$ to denote that the vertex or the edge is covered by the path $\mathcal{P}$, and $v \to w \in \mathcal{P}$ if the directed edge is in $\mathcal{P}$.
\end{definition}
\begin{definition}[Partial order]
 Let $\leq$ be a binary relation on a set $P$.
 Then $\leq$ is a \emph{partial order} if it satisfies:
  \begin{enumerate}
    \item Reflexivity: $a \leq a$ for all $a \in P$,
    \item Transitivity: if $a \leq b$ and $b \leq c$ then $a \leq c$ for all $a, b, c \in P$, and
    \item Antisymmetry: if $a \leq b$ and $b \leq a$ then $a = b$ for all $a, b \in P$.
  \end{enumerate}
\end{definition}
\begin{definition}
 We define the set of neighbours of a node $u$ as $N_G(u) \coloneqq \{w \in V \ |\ (w,u) \in E\}$.
\end{definition}
\noindent We can now define the following:
\begin{definition}[Measurement-circuit flow]
 Let $D$ be a ZX diagram with spiders $V$ and edges $E$.
 A \emph{measurement-circuit flow} or \emph{MC flow} on $D$ is a partial order $\leq$ on $V$ and a set of directed paths $\{\mathcal{P}_{k}\}_{k = 1}^{q}$ such that
  \begin{description}
    \item[(O1)] the directed paths respect the partial order $\leq$ on the vertices, i.e.\@ if $x \to y \in \mathcal{P}_i$, then $x \leq y$,
    \item [(O2)] for each $x \to y \in \mathcal{P}_i$ and $z \in N_G(y)$, $x \leq z$ or $\exists \mathcal{P}_{j} \in \{\mathcal{P}_{k}\}_{k = 1}^{q}$ such that $z \to y \in \mathcal{P}_{j}$,
    \item[(P1)] for every boundary and Hadamard node there is a path that fully covers it,
    \item[(P2)] each Z and X spider has at most one edge not covered by any path,
    \item[(P3)] no edge is covered by more than one path.
  \end{description}
\end{definition}

\noindent This definition of flow permits single-legged spiders to be part of the diagram as well as paths to overlap on vertices.
Therefore, unlike causal flow~\parencite{danosDeterminismOneway2006}, it can also represent measurements, preparations, and single-qubit Pauli measurements, and it admits a different circuit extraction procedure.

\begin{remark}
 While we use a partial order and a set of paths with certain properties to express flow, it is possible to define a similar notion using a successor function instead of the set of paths.
 In that case, the successor function would map vertices $V$ to the powerset of $V$, but this function would not be surjective, as some nodes are the single-legged part of a one-qubit Pauli measurement.
 It is worth remembering that the set of paths carries some extra information that tells us exactly where each qubit is.
 Therefore, these two definitions would be equivalent up to some permutation of spiders.
\end{remark}

\begin{definition}[Strict MC flow]
 We say an MC flow $(\leq, \{\mathcal{P}_{k}\}_{k = 1}^{q})$ on a ZX diagram is \emph{strict} if it also satisfies:
  \begin{description}
    \item[(PS)] any two connected spiders with at least 3 legs each either have the same colour, or there is a path $\mathcal{P}_j$ that covers the edge between them.
  \end{description}
\end{definition}

\begin{definition}[Well-covered MC flow]
 We say an MC flow $(\leq, \{\mathcal{P}_{k}\}_{k = 1}^{q})$ on a ZX diagram is \emph{well-covered}, if it also satisfies:
  \begin{description}
    \item[(PWC)] no path starts or finishes at a Z or X spider with degree more than 1
  \end{description}
\end{definition}

\begin{proposition}
  \label{prop:to-well-covered}
 Given a ZX-diagram with MC flow, we can obtain a new ZX-diagram with a well-covered MC flow using only distance-preserving rewrites.
\end{proposition}
\begin{proof}
 For all spiders $s$ with $\deg(s) > 1$ where a path finishes (or starts), we can perform one of the following rewrites (or its colour-flipped version) until $n = k$:
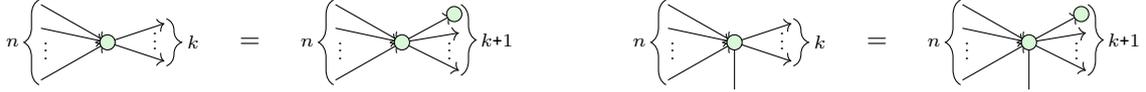
  \[
 \input{flow/well-covered-transformation.tikz}
  \]
 Both of these rewrites are distance-preserving by \autoref{thm:fuse-rewrite}, and applying either introduces an edge connected to a new single-legged phaseless spider.
 After the rewrite, we can add the new edge to a directed path that finishes (or starts) in $s$, therefore decreasing the number of places where (PWC) is violated.
 Once $n=k$ for all spiders in the diagram, we obtain a ZX diagram with no violations of (PWC), and thus, a ZX diagram with a well-covered MC flow.
\end{proof}

\subsection{Circuit extraction of arbitrary diagrams}
If a diagram with a maximum spider degree less than $4$ has a well-covered MC flow, its structure is close to a circuit in MCC form.
In fact, two simple distance-preserving rewrites enable us to extract such a circuit:
\begin{proposition}
  \label{prop:flow-correct-3}
 If a ZX diagram $D$ has a well-covered MC flow $(\leq, \{\mathcal{P}_{k}\}_{k = 1}^{q})$ and the maximum vertex degree in $D$ is not larger than $3$, then we can distance-preservingly extract a circuit in MCC form (\autoref{def:mcc-form}).
\end{proposition}
The proof is similar to the circuit extraction procedure for causal flow~\parencite[Theorem 5.8]{duncanGraphicalApproach2013}, but no unfusions are necessary because of the maximum degree condition and the different gate set.
\begin{proof}
 First, by (PWC) and the maximum vertex degree condition, every spider $s$ with $\deg(s) > 1$ has exactly one path going through it.
 In particular, a spider with degree 3 must have two covered edges and an uncovered one, while a degree 2 spider has both its legs covered; higher arity spiders are not allowed by the maximum degree condition.
 Because of this, the conditions on the partial order (O1) and (O2) are equivalent to the one given by causal flow.
 This implies that the paths $\{\mathcal{P}_{k}\}_{k = 1}^{q}$ define the qubits of our system and there are no causal loops in $D$, i.e.\@ it is circuit-like~\parencite[Definition 3.14]{duncanGraphicalApproach2013}.
 By (P2), spiders have only a single uncovered edge, so unfusion is not necessary when extracting a circuit.

 One-legged spiders that are covered by a path can be extracted as a measurement or preparation.
 Any spider with a degree of $2$ is extracted as a single-qubit gate.
 Let $s$ be a spider with degree $3$ that is connected to $q$ via its uncovered edge.
 In this context, $q$ cannot be a boundary or have degree $2$ by (P1) and (PWC); the degree of $q$ is $1$ or $3$. 
 Thus, $q$ and $s$ can only be one of the following (or their colour-flipped version):
  \begin{equation}
 \input{flow/flow-preserving-extraction-3.tikz}
      \label{eq:flow-preserving-extraction-3}
  \end{equation}
 The leftmost diagram can be rewritten with~\eqref{fuse-rewrite}, which is distance-preserving by \autoref{thm:elim-rewrite} and can be extracted as the identity.
 The remaining diagram already corresponds to a circuit in MCC form.
 This covers every spider that can appear in the diagram and concludes our proof.
\end{proof}

\begin{example}
 The diagram on the LHS has a well-covered MC flow, and it is built up from spiders with three legs; therefore, we can extract the circuit on the RHS:
  \[
 \input{flow/four-legged-spider.tikz}
  \]
\end{example}

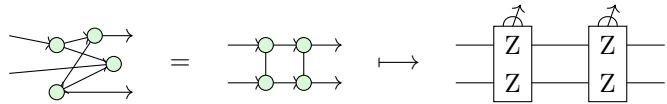
\begin{proposition}
  \label{prop:strict-flow-correct}
 If a ZX diagram $D$ has a strict and well-covered MC flow $(\leq, \{\mathcal{P}_{k}\}_{k = 1}^{q})$ and the maximum vertex degree in $D$ is not larger than $3$, then we can distance-preservingly extract a circuit in MCLC form (\autoref{def:mclc-form}).
\end{proposition}
\begin{proof}
 The same procedure applies as in \autoref{prop:flow-correct-3}; the only thing to show is that this procedure will not result in extracted CNOT gates if the MC flow is strict.
  \autoref{prop:flow-correct-3} only extracts two-qubit gates if it finds two connected 3-legged spiders ($s$ and $q$) that are connected via an uncovered edge; see \autoref{eq:flow-preserving-extraction-3}.
 By (PS), two connected degree-3 spiders either have the same colour or there is a path $\mathcal{P}_j$ that covers the edge between them.
 Since $s$ and $q$ are explicitly connected via their uncovered edge, they must have the same colour, and thus, are never extracted as a CNOT gate.
\end{proof}
\subsection{Circuit extraction with spider of high degree}
Now, we direct our attention to the distance-preserving extraction of circuits from larger diagrams with no bound on maximum spider degree:
\begin{proposition}
  \label{prop:flow-correct}
 If a ZX diagram $D$ has a well-covered MC flow $(\leq, \{\mathcal{P}_{k}\}_{k = 1}^{q})$, we can rewrite it to an equivalent ZX diagram $D'$ with a well-covered MC flow using only distance-preserving rewrites such that all spiders in $D'$ have degree at most three.
  Moreover, if the MC flow is strict, the resulting MC flow will also be strict.
\end{proposition}
\begin{proof}
 To prove the proposition, we iteratively decompose all spiders with a degree of more than three using distance-preserving rewrites.
 Additionally, we show that at every step the existence of a well-covered MC flow is preserved.

 For each spider $s$ in $D$ with $n \coloneqq \deg(s) > 3$, we can apply the following rewrites:
  \begin{equation}
 \input{flow/flow-preserving-rewrites.tikz}
    \label{eq:iterative-decomposition}
  \end{equation}
 The rewrites in the first two rows strictly reduce the maximum degree.
 The rewrites in the last row increase the maximum degree by two; however, they allow one rewrite of the second row to be applied after.
 Thus, the maximum degree is reduced to $\frac{n + 2}{2}$, which, for $n > 4$, is less than $n$.
 Thus, iteratively applying the appropriate rewrite terminates, and the final diagram is expressed as a composition of spiders of degree at most three.

 Each rewrite we apply in \autoref{eq:iterative-decomposition} is distance-preserving by \autoref{thm:four-legged-spider}, \autoref{thm:five-legged-spider}, \autoref{thm:inductive-dist-rewrite}, \autoref{thm:inductive-plus-dist-rewrite}, and \autoref{thm:fuse-rewrite}.
 The rewrites preserve the existence of a well-covered MC flow, as indicated by the paths shown in \autoref{eq:iterative-decomposition}.
  Moreover, the only new edges that are not part of any path are introduced between spiders of the same colour.
  Thus, if the original MC flow was strict, the final diagram will also have a strict MC flow, completing the proof.
\end{proof}

\begin{corollary}
 Any diagram with an MC flow can be extracted into a corresponding quantum circuit using only distance-preserving rewrites.
\end{corollary}

\subsubsection{Resource overhead}
\label{sec:classifying-overhead}
Given a ZX diagram with a well-covered MC flow and spiders of degree at most three, by \autoref{prop:flow-correct-3}, we can extract a circuit in MCC form.
This circuit uses at most as many qubits as there are directed paths in the MC flow, fewer if we have paths that have no temporal overlap and thus allow for qubit reuse.

However, given an arbitrary diagram with MC flow, creating a diagram with a well-covered MC flow and spiders of degree at most three can introduce additional paths and therefore qubits; see \autoref{prop:to-well-covered} and \autoref{prop:flow-correct}.
In this section, we focus on \autoref{prop:flow-correct}.
We characterise the number of gates and additional qubits a well-covered spider of high degree requires to be expressed as a distance-preserving diagram with spiders of degree at most three.

\begin{remark}
 The same resource overhead calculations apply to extracting circuits in MCLC form using a strict MC flow.
\end{remark}

\paragraph{Number of additional qubits}
To characterise the number of additional qubits required to implement a spider of high degree, we can look at the inductive decomposition given by \autoref{prop:flow-correct}.
For each inductive step, we can characterise how many additional paths it introduces:
\[
 f(n) =
    \begin{dcases}
 0 & \text{if } n = 4\\
 f\! \left(\frac{n}{2}\right) & \text{if } n \bmod 4 = 0\\
 1 + f\left(n + 2\right) & \text{if } n \bmod 4 = 2\\
    \end{dcases}
\]
We observe that for an even $n > 4$ this function always terminates at the base case $n = 4$.
In \autoref{prop:flow-correct}, we argued that the decomposition always finishes.
So, it remains to be shown that it always finishes by decomposing a spider of degree four.
We observe that the only other rewrite to reduce the degree is $r_n$, i.e.\@ the second case.
This rewrite halves $n$ assuming that $n \bmod 4 = 0$.
But then, to get to a case of $n \leq 4$, previously $n$ had to be $8$.
Thus, if the decomposition terminates, it terminates by decomposing a spider of degree $4$.

This function is bounded by $\log_2(n)$; see~\autoref{appendix:proof-f-bound}.
For $n = 2^k$, the number of additional qubits is $0$.

\paragraph{Number of gates}
Similarly, we can characterise the number of weight-two Paulis required to implement a high-degree spider. 
Once again, considering the inductive definition, we get: 
\[
 g(n) =
    \begin{dcases}
 2 & \text{if } n = 4\\
 n + 2g\!\left(\frac{n}{2}\right) & \text{if } n \bmod 4 = 0\\
 g\left(n + 2\right) & \text{if } n \bmod 4 = 2\\
    \end{dcases}
\]
This function is bounded by $2 n\log(n)$; see~\autoref{appendix:proof-g-bound}.

\section{Decomposing high-weight measurements}
\label{sec:rewriting-measurements}
Recent works have explored how to express high-weight measurements in terms of more easily implementable, low-weight measurements at the cost of using additional qubits~\parencite{gidneyPairMeasurement2023, grans-samuelssonFaulttolerantPairwise2024, fuenteDynamicalWeight2024, moflicConstantDepth2024}, often using the rewrites offered by the ZX calculus.
In this line, we will present a procedure to express arbitrary Pauli measurements as quantum circuits in low-weight Clifford and measurement form. 
However, by restricting ourselves to distance-preserving rewrites, we guarantee that single errors in the implementation do not propagate to become multiple data errors. 

Given an arbitrary Pauli measurement expressed as a ZX diagram, we will use distance-preserving rewrites to bring it into a form that has a strict and well-covered MC flow.
By \autoref{prop:flow-correct}, this means that the diagram can be rewritten into an equivalent diagram with a strict and well-covered MC flow and spiders of degree at most three using only distance-preserving rewrites.
But then, by \autoref{prop:flow-correct-3} we know that we can extract a corresponding quantum circuit of the desired shape. 

We have:
\begin{proposition}[Decomposing Pauli measurements]
    \label{prop:decomposing-measurements}
 Any Pauli measurement can be distance-preservingly rewritten into an equivalent diagram with a strict and well-covered MC flow.
\end{proposition}
\begin{proof}
 Any Pauli measurement can be written as a Pauli-Z measurement preceded and followed by local Cliffords~\parencite{KissingerWetering2024Book}.
 We will represent these as $lc$. 
    
 We differentiate whether the Pauli measurement has even or odd weight.
 If the Pauli has an even weight, we have:
    
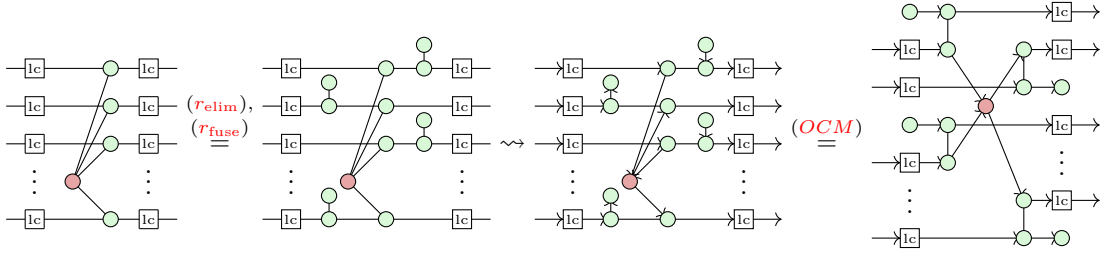
\[\begin{tikzpicture}
	\begin{pgfonlayer}{nodelayer}
		\node [style=none] (0) at (-2, 2) {};
		\node [style=none] (1) at (-2, 1) {};
		\node [style=none] (2) at (-2, 0) {};
		\node [style=none] (3) at (-2, -2) {};
		\node [style=none] (4) at (2.5, 2) {};
		\node [style=none] (5) at (2.5, 1) {};
		\node [style=none] (6) at (2.5, 0) {};
		\node [style=none] (7) at (2.5, -2) {};
		\node [style=Z dot] (8) at (0.75, 2) {};
		\node [style=Z dot] (9) at (0.75, 1) {};
		\node [style=Z dot] (10) at (0.75, 0) {};
		\node [style=Z dot] (11) at (0.75, -2) {};
		\node [style=X dot] (12) at (-0.25, -1) {};
		\node [style=none] (13) at (-1.5, 2.25) {};
		\node [style=none] (14) at (-1, 2.25) {};
		\node [style=none] (15) at (-1.5, 1.75) {};
		\node [style=none] (16) at (-1, 1.75) {};
		\node [style=none] (17) at (-1.25, 2) {\tiny lc};
		\node [style=none] (18) at (1.5, 2.25) {};
		\node [style=none] (19) at (2, 2.25) {};
		\node [style=none] (20) at (1.5, 1.75) {};
		\node [style=none] (21) at (2, 1.75) {};
		\node [style=none] (22) at (1.75, 2) {\tiny lc};
		\node [style=none] (23) at (-1.5, 1.25) {};
		\node [style=none] (24) at (-1, 1.25) {};
		\node [style=none] (25) at (-1.5, 0.75) {};
		\node [style=none] (26) at (-1, 0.75) {};
		\node [style=none] (27) at (-1.25, 1) {\tiny lc};
		\node [style=none] (28) at (1.5, 1.25) {};
		\node [style=none] (29) at (2, 1.25) {};
		\node [style=none] (30) at (1.5, 0.75) {};
		\node [style=none] (31) at (2, 0.75) {};
		\node [style=none] (32) at (1.75, 1) {\tiny lc};
		\node [style=none] (33) at (-1.5, 0.25) {};
		\node [style=none] (34) at (-1, 0.25) {};
		\node [style=none] (35) at (-1.5, -0.25) {};
		\node [style=none] (36) at (-1, -0.25) {};
		\node [style=none] (37) at (-1.25, 0) {\tiny lc};
		\node [style=none] (38) at (1.5, 0.25) {};
		\node [style=none] (39) at (2, 0.25) {};
		\node [style=none] (40) at (1.5, -0.25) {};
		\node [style=none] (41) at (2, -0.25) {};
		\node [style=none] (42) at (1.75, 0) {\tiny lc};
		\node [style=none] (43) at (-1.5, -1.75) {};
		\node [style=none] (44) at (-1, -1.75) {};
		\node [style=none] (45) at (-1.5, -2.25) {};
		\node [style=none] (46) at (-1, -2.25) {};
		\node [style=none] (47) at (-1.25, -2) {\tiny lc};
		\node [style=none] (48) at (1.5, -1.75) {};
		\node [style=none] (49) at (2, -1.75) {};
		\node [style=none] (50) at (1.5, -2.25) {};
		\node [style=none] (51) at (2, -2.25) {};
		\node [style=none] (52) at (1.75, -2) {\tiny lc};
		\node [style=none] (53) at (-1.5, 2) {};
		\node [style=none] (54) at (-1, 2) {};
		\node [style=none] (55) at (-1.5, 1) {};
		\node [style=none] (56) at (-1, 1) {};
		\node [style=none] (57) at (-1.5, 0) {};
		\node [style=none] (58) at (-1, 0) {};
		\node [style=none] (59) at (-1.5, -2) {};
		\node [style=none] (60) at (-1, -2) {};
		\node [style=none] (61) at (1.5, -2) {};
		\node [style=none] (62) at (2, -2) {};
		\node [style=none] (63) at (1.5, 0) {};
		\node [style=none] (64) at (2, 0) {};
		\node [style=none] (65) at (1.5, 1) {};
		\node [style=none] (66) at (2, 1) {};
		\node [style=none] (67) at (1.5, 2) {};
		\node [style=none] (68) at (2, 2) {};
		\node [style=none] (69) at (-1.25, -0.8) {$\vdots$};
		\node [style=none] (70) at (1.75, -0.8) {$\vdots$};
		\node [style=none] (71) at (0.75, 2.625) {};
	\end{pgfonlayer}
	\begin{pgfonlayer}{edgelayer}
		\draw (12) to (11);
		\draw (12) to (10);
		\draw (9) to (12);
		\draw (8) to (12);
		\draw (13.center)
			 to (14.center)
			 to (16.center)
			 to (15.center)
			 to cycle;
		\draw (18.center)
			 to (19.center)
			 to (21.center)
			 to (20.center)
			 to cycle;
		\draw (23.center)
			 to (24.center)
			 to (26.center)
			 to (25.center)
			 to cycle;
		\draw (28.center) to (29.center);
		\draw (29.center) to (31.center);
		\draw (31.center) to (30.center);
		\draw (30.center) to (28.center);
		\draw (33.center) to (34.center);
		\draw (34.center) to (36.center);
		\draw (36.center) to (35.center);
		\draw (35.center) to (33.center);
		\draw (38.center) to (39.center);
		\draw (39.center) to (41.center);
		\draw (41.center) to (40.center);
		\draw (40.center) to (38.center);
		\draw (43.center) to (44.center);
		\draw (44.center) to (46.center);
		\draw (46.center) to (45.center);
		\draw (45.center) to (43.center);
		\draw (48.center) to (49.center);
		\draw (49.center) to (51.center);
		\draw (51.center) to (50.center);
		\draw (50.center) to (48.center);
		\draw (0.center) to (53.center);
		\draw (54.center) to (8);
		\draw (8) to (67.center);
		\draw (68.center) to (4.center);
		\draw (1.center) to (55.center);
		\draw (56.center) to (9);
		\draw (9) to (65.center);
		\draw (66.center) to (5.center);
		\draw (2.center) to (57.center);
		\draw (58.center) to (10);
		\draw (10) to (63.center);
		\draw (64.center) to (6.center);
		\draw (3.center) to (59.center);
		\draw (60.center) to (11);
		\draw (11) to (61.center);
		\draw (62.center) to (7.center);
	\end{pgfonlayer}
\end{tikzpicture}
 \overset{\eqref{elim-rewrite},}{\overset{\eqref{fuse-rewrite}}{=}} \begin{tikzpicture}
	\begin{pgfonlayer}{nodelayer}
		\node [style=none] (0) at (-2, 2) {};
		\node [style=none] (1) at (-2, 1) {};
		\node [style=none] (2) at (-2, 0) {};
		\node [style=none] (3) at (-2, -2) {};
		\node [style=none] (4) at (4, 2) {};
		\node [style=none] (5) at (4, 1) {};
		\node [style=none] (6) at (4, 0) {};
		\node [style=none] (7) at (4, -2) {};
		\node [style=Z dot] (8) at (1.25, 2) {};
		\node [style=Z dot] (9) at (1.25, 1) {};
		\node [style=Z dot] (10) at (1.25, 0) {};
		\node [style=Z dot] (11) at (1.25, -2) {};
		\node [style=none] (12) at (-1.5, 2.25) {};
		\node [style=none] (13) at (-1, 2.25) {};
		\node [style=none] (14) at (-1.5, 1.75) {};
		\node [style=none] (15) at (-1, 1.75) {};
		\node [style=none] (16) at (-1.25, 2) {\tiny lc};
		\node [style=none] (17) at (3, 2.25) {};
		\node [style=none] (18) at (3.5, 2.25) {};
		\node [style=none] (19) at (3, 1.75) {};
		\node [style=none] (20) at (3.5, 1.75) {};
		\node [style=none] (21) at (3.25, 2) {\tiny lc};
		\node [style=none] (22) at (-1.5, 1.25) {};
		\node [style=none] (23) at (-1, 1.25) {};
		\node [style=none] (24) at (-1.5, 0.75) {};
		\node [style=none] (25) at (-1, 0.75) {};
		\node [style=none] (26) at (-1.25, 1) {\tiny lc};
		\node [style=none] (27) at (3, 1.25) {};
		\node [style=none] (28) at (3.5, 1.25) {};
		\node [style=none] (29) at (3, 0.75) {};
		\node [style=none] (30) at (3.5, 0.75) {};
		\node [style=none] (31) at (3.25, 1) {\tiny lc};
		\node [style=none] (32) at (-1.5, 0.25) {};
		\node [style=none] (33) at (-1, 0.25) {};
		\node [style=none] (34) at (-1.5, -0.25) {};
		\node [style=none] (35) at (-1, -0.25) {};
		\node [style=none] (36) at (-1.25, 0) {\tiny lc};
		\node [style=none] (37) at (3, 0.25) {};
		\node [style=none] (38) at (3.5, 0.25) {};
		\node [style=none] (39) at (3, -0.25) {};
		\node [style=none] (40) at (3.5, -0.25) {};
		\node [style=none] (41) at (3.25, 0) {\tiny lc};
		\node [style=none] (42) at (-1.5, -1.75) {};
		\node [style=none] (43) at (-1, -1.75) {};
		\node [style=none] (44) at (-1.5, -2.25) {};
		\node [style=none] (45) at (-1, -2.25) {};
		\node [style=none] (46) at (-1.25, -2) {\tiny lc};
		\node [style=none] (47) at (3, -1.75) {};
		\node [style=none] (48) at (3.5, -1.75) {};
		\node [style=none] (49) at (3, -2.25) {};
		\node [style=none] (50) at (3.5, -2.25) {};
		\node [style=none] (51) at (3.25, -2) {\tiny lc};
		\node [style=none] (52) at (-1.5, 2) {};
		\node [style=none] (53) at (-1, 2) {};
		\node [style=none] (54) at (-1.5, 1) {};
		\node [style=none] (55) at (-1, 1) {};
		\node [style=none] (56) at (-1.5, 0) {};
		\node [style=none] (57) at (-1, 0) {};
		\node [style=none] (58) at (-1.5, -2) {};
		\node [style=none] (59) at (-1, -2) {};
		\node [style=none] (60) at (3, -2) {};
		\node [style=none] (61) at (3.5, -2) {};
		\node [style=none] (62) at (3, 0) {};
		\node [style=none] (63) at (3.5, 0) {};
		\node [style=none] (64) at (3, 1) {};
		\node [style=none] (65) at (3.5, 1) {};
		\node [style=none] (66) at (3, 2) {};
		\node [style=none] (67) at (3.5, 2) {};
		\node [style=none] (68) at (-1.25, -0.8) {$\vdots$};
		\node [style=none] (69) at (3.25, -0.8) {$\vdots$};
		\node [style=X dot] (70) at (0.25, -1) {};
		\node [style=Z dot] (71) at (2.25, 2) {};
		\node [style=Z dot] (72) at (2.25, 2.625) {};
		\node [style=Z dot] (73) at (-0.25, 1) {};
		\node [style=Z dot] (74) at (-0.25, 1.625) {};
		\node [style=Z dot] (75) at (2.25, 0) {};
		\node [style=Z dot] (76) at (2.25, 0.625) {};
		\node [style=Z dot] (77) at (-0.25, -2) {};
		\node [style=Z dot] (78) at (-0.25, -1.375) {};
	\end{pgfonlayer}
	\begin{pgfonlayer}{edgelayer}
		\draw (12.center)
			 to (13.center)
			 to (15.center)
			 to (14.center)
			 to cycle;
		\draw (20.center)
			 to (19.center)
			 to (17.center)
			 to (18.center)
			 to cycle;
		\draw (22.center)
			 to (23.center)
			 to (25.center)
			 to (24.center)
			 to cycle;
		\draw (30.center)
			 to (29.center)
			 to (27.center)
			 to (28.center)
			 to cycle;
		\draw (35.center)
			 to (34.center)
			 to (32.center)
			 to (33.center)
			 to cycle;
		\draw (40.center)
			 to (39.center)
			 to (37.center)
			 to (38.center)
			 to cycle;
		\draw (42.center)
			 to (43.center)
			 to (45.center)
			 to (44.center)
			 to cycle;
		\draw (48.center)
			 to (50.center)
			 to (49.center)
			 to (47.center)
			 to cycle;
		\draw (0.center) to (52.center);
		\draw (53.center) to (8);
		\draw (67.center) to (4.center);
		\draw (1.center) to (54.center);
		\draw (9) to (64.center);
		\draw (65.center) to (5.center);
		\draw (2.center) to (56.center);
		\draw (57.center) to (10);
		\draw (63.center) to (6.center);
		\draw (3.center) to (58.center);
		\draw (11) to (60.center);
		\draw (61.center) to (7.center);
		\draw (8) to (70);
		\draw (70) to (9);
		\draw (10) to (70);
		\draw (11) to (70);
		\draw (8) to (71);
		\draw (71) to (66.center);
		\draw (72) to (71);
		\draw (55.center) to (73);
		\draw (73) to (9);
		\draw (74) to (73);
		\draw (10) to (75);
		\draw (75) to (62.center);
		\draw (76) to (75);
		\draw (59.center) to (77);
		\draw (77) to (11);
		\draw (78) to (77);
	\end{pgfonlayer}
\end{tikzpicture}
 \rightsquigarrow \begin{tikzpicture}
	\begin{pgfonlayer}{nodelayer}
		\node [style=none] (0) at (-2, 2) {};
		\node [style=none] (1) at (-2, 1) {};
		\node [style=none] (2) at (-2, 0) {};
		\node [style=none] (3) at (-2, -2) {};
		\node [style=none] (4) at (4.5, 2) {};
		\node [style=none] (5) at (4.5, 1) {};
		\node [style=none] (6) at (4.5, 0) {};
		\node [style=none] (7) at (4.5, -2) {};
		\node [style=none] (8) at (1.5, 2) {};
		\node [style=none] (9) at (1.5, 1) {};
		\node [style=none] (10) at (1.5, 0) {};
		\node [style=none] (11) at (1.5, -2) {};
		\node [style=none] (12) at (-1.25, 2.25) {};
		\node [style=none] (13) at (-0.75, 2.25) {};
		\node [style=none] (14) at (-1.25, 1.75) {};
		\node [style=none] (15) at (-0.75, 1.75) {};
		\node [style=none] (16) at (-1, 2) {\tiny lc};
		\node [style=none] (17) at (3.25, 2.25) {};
		\node [style=none] (18) at (3.75, 2.25) {};
		\node [style=none] (19) at (3.25, 1.75) {};
		\node [style=none] (20) at (3.75, 1.75) {};
		\node [style=none] (21) at (3.5, 2) {\tiny lc};
		\node [style=none] (22) at (-1.25, 1.25) {};
		\node [style=none] (23) at (-0.75, 1.25) {};
		\node [style=none] (24) at (-1.25, 0.75) {};
		\node [style=none] (25) at (-0.75, 0.75) {};
		\node [style=none] (26) at (-1, 1) {\tiny lc};
		\node [style=none] (27) at (3.25, 1.25) {};
		\node [style=none] (28) at (3.75, 1.25) {};
		\node [style=none] (29) at (3.25, 0.75) {};
		\node [style=none] (30) at (3.75, 0.75) {};
		\node [style=none] (31) at (3.5, 1) {\tiny lc};
		\node [style=none] (32) at (-1.25, 0.25) {};
		\node [style=none] (33) at (-0.75, 0.25) {};
		\node [style=none] (34) at (-1.25, -0.25) {};
		\node [style=none] (35) at (-0.75, -0.25) {};
		\node [style=none] (36) at (-1, 0) {\tiny lc};
		\node [style=none] (37) at (3.25, 0.25) {};
		\node [style=none] (38) at (3.75, 0.25) {};
		\node [style=none] (39) at (3.25, -0.25) {};
		\node [style=none] (40) at (3.75, -0.25) {};
		\node [style=none] (41) at (3.5, 0) {\tiny lc};
		\node [style=none] (42) at (-1.25, -1.75) {};
		\node [style=none] (43) at (-0.75, -1.75) {};
		\node [style=none] (44) at (-1.25, -2.25) {};
		\node [style=none] (45) at (-0.75, -2.25) {};
		\node [style=none] (46) at (-1, -2) {\tiny lc};
		\node [style=none] (47) at (3.25, -1.75) {};
		\node [style=none] (48) at (3.75, -1.75) {};
		\node [style=none] (49) at (3.25, -2.25) {};
		\node [style=none] (50) at (3.75, -2.25) {};
		\node [style=none] (51) at (3.5, -2) {\tiny lc};
		\node [style=none] (52) at (-1.25, 2) {};
		\node [style=none] (53) at (-0.75, 2) {};
		\node [style=none] (54) at (-1.25, 1) {};
		\node [style=none] (55) at (-0.75, 1) {};
		\node [style=none] (56) at (-1.25, 0) {};
		\node [style=none] (57) at (-0.75, 0) {};
		\node [style=none] (58) at (-1.25, -2) {};
		\node [style=none] (59) at (-0.75, -2) {};
		\node [style=none] (60) at (3.25, -2) {};
		\node [style=none] (61) at (3.75, -2) {};
		\node [style=none] (62) at (3.25, 0) {};
		\node [style=none] (63) at (3.75, 0) {};
		\node [style=none] (64) at (3.25, 1) {};
		\node [style=none] (65) at (3.75, 1) {};
		\node [style=none] (66) at (3.25, 2) {};
		\node [style=none] (67) at (3.75, 2) {};
		\node [style=none] (68) at (-1, -0.8) {$\vdots$};
		\node [style=none] (69) at (3.5, -0.8) {$\vdots$};
		\node [style=none] (70) at (0.5, -1) {};
		\node [style=none] (71) at (2.5, 2) {};
		\node [style=none] (72) at (2.5, 2.625) {};
		\node [style=none] (73) at (0, 1) {};
		\node [style=none] (74) at (0, 1.625) {};
		\node [style=none] (75) at (2.5, 0) {};
		\node [style=none] (76) at (2.5, 0.625) {};
		\node [style=none] (77) at (0, -2) {};
		\node [style=none] (78) at (0, -1.375) {};
		\node [style=Z dot] (79) at (0, 1) {};
		\node [style=Z dot] (80) at (0, 1.625) {};
		\node [style=Z dot] (81) at (1.5, 2) {};
		\node [style=Z dot] (82) at (1.5, 1) {};
		\node [style=Z dot] (83) at (1.5, 0) {};
		\node [style=Z dot] (84) at (2.5, 0) {};
		\node [style=Z dot] (85) at (2.5, 2) {};
		\node [style=Z dot] (86) at (2.5, 2.625) {};
		\node [style=Z dot] (87) at (2.5, 0.625) {};
		\node [style=Z dot] (88) at (0, -2) {};
		\node [style=Z dot] (89) at (0, -1.375) {};
		\node [style=Z dot] (90) at (1.5, -2) {};
		\node [style=X dot] (91) at (0.5, -1) {};
	\end{pgfonlayer}
	\begin{pgfonlayer}{edgelayer}
		\draw (15.center)
			 to (14.center)
			 to (12.center)
			 to (13.center)
			 to cycle;
		\draw (20.center)
			 to (19.center)
			 to (17.center)
			 to (18.center)
			 to cycle;
		\draw (22.center)
			 to (23.center)
			 to (25.center)
			 to (24.center)
			 to cycle;
		\draw (30.center)
			 to (29.center)
			 to (27.center)
			 to (28.center)
			 to cycle;
		\draw (35.center)
			 to (34.center)
			 to (32.center)
			 to (33.center)
			 to cycle;
		\draw (40.center)
			 to (39.center)
			 to (37.center)
			 to (38.center)
			 to cycle;
		\draw (42.center)
			 to (43.center)
			 to (45.center)
			 to (44.center)
			 to cycle;
		\draw (48.center)
			 to (50.center)
			 to (49.center)
			 to (47.center)
			 to cycle;
		\draw [style=arrow] (0.center) to (52.center);
		\draw [style=arrow] (53.center) to (81);
		\draw [style=arrow] (81) to (91);
		\draw [style=arrow] (91) to (82);
		\draw [style=arrow] (82) to (64.center);
		\draw [style=arrow] (65.center) to (5.center);
		\draw [style=arrow] (2.center) to (56.center);
		\draw [style=arrow] (57.center) to (83);
		\draw [style=arrow] (91) to (90);
		\draw [style=arrow] (90) to (60.center);
		\draw [style=arrow] (61.center) to (7.center);
		\draw [style=arrow] (83) to (91);
		\draw [style=arrow] (3.center) to (58.center);
		\draw [style=arrow] (59.center) to (88);
		\draw [style=arrow] (88) to (89);
		\draw [style=arrow] (1.center) to (54.center);
		\draw [style=arrow] (55.center) to (79);
		\draw [style=arrow] (79) to (80);
		\draw [style=arrow] (86) to (85);
		\draw [style=arrow] (85) to (66.center);
		\draw [style=arrow] (67.center) to (4.center);
		\draw [style=arrow] (87) to (84);
		\draw [style=arrow] (84) to (62.center);
		\draw [style=arrow] (63.center) to (6.center);
		\draw (79) to (82);
		\draw (81) to (85);
		\draw (83) to (84);
		\draw (88) to (90);
	\end{pgfonlayer}
\end{tikzpicture}
 \overset{(\hyperlink{eq:OCM}{OCM})}{=} \begin{tikzpicture}
	\begin{pgfonlayer}{nodelayer}
		\node [style=none] (0) at (-2, 2.5) {};
		\node [style=none] (1) at (-2, 1.5) {};
		\node [style=none] (2) at (-2, -0.5) {};
		\node [style=none] (3) at (-2, -2.5) {};
		\node [style=none] (4) at (4, 3.5) {};
		\node [style=none] (5) at (4, 2.5) {};
		\node [style=none] (6) at (4, 0.5) {};
		\node [style=none] (7) at (4, -1.5) {};
		\node [style=none] (8) at (0, 2.5) {};
		\node [style=none] (9) at (2, 2.5) {};
		\node [style=none] (10) at (0, -0.5) {};
		\node [style=none] (11) at (2, -1.5) {};
		\node [style=none] (12) at (-1.25, 2.75) {};
		\node [style=none] (13) at (-0.75, 2.75) {};
		\node [style=none] (14) at (-1.25, 2.25) {};
		\node [style=none] (15) at (-0.75, 2.25) {};
		\node [style=none] (16) at (-1, 2.5) {\tiny lc};
		\node [style=none] (17) at (2.75, 3.75) {};
		\node [style=none] (18) at (3.25, 3.75) {};
		\node [style=none] (19) at (2.75, 3.25) {};
		\node [style=none] (20) at (3.25, 3.25) {};
		\node [style=none] (21) at (3, 3.5) {\tiny lc};
		\node [style=none] (22) at (-1.25, 1.75) {};
		\node [style=none] (23) at (-0.75, 1.75) {};
		\node [style=none] (24) at (-1.25, 1.25) {};
		\node [style=none] (25) at (-0.75, 1.25) {};
		\node [style=none] (26) at (-1, 1.5) {\tiny lc};
		\node [style=none] (27) at (2.75, 2.75) {};
		\node [style=none] (28) at (3.25, 2.75) {};
		\node [style=none] (29) at (2.75, 2.25) {};
		\node [style=none] (30) at (3.25, 2.25) {};
		\node [style=none] (31) at (3, 2.5) {\tiny lc};
		\node [style=none] (32) at (-1.25, -0.25) {};
		\node [style=none] (33) at (-0.75, -0.25) {};
		\node [style=none] (34) at (-1.25, -0.75) {};
		\node [style=none] (35) at (-0.75, -0.75) {};
		\node [style=none] (36) at (-1, -0.5) {\tiny lc};
		\node [style=none] (37) at (2.75, 0.75) {};
		\node [style=none] (38) at (3.25, 0.75) {};
		\node [style=none] (39) at (2.75, 0.25) {};
		\node [style=none] (40) at (3.25, 0.25) {};
		\node [style=none] (41) at (3, 0.5) {\tiny lc};
		\node [style=none] (42) at (-1.25, -2.25) {};
		\node [style=none] (43) at (-0.75, -2.25) {};
		\node [style=none] (44) at (-1.25, -2.75) {};
		\node [style=none] (45) at (-0.75, -2.75) {};
		\node [style=none] (46) at (-1, -2.5) {\tiny lc};
		\node [style=none] (47) at (2.75, -1.25) {};
		\node [style=none] (48) at (3.25, -1.25) {};
		\node [style=none] (49) at (2.75, -1.75) {};
		\node [style=none] (50) at (3.25, -1.75) {};
		\node [style=none] (51) at (3, -1.5) {\tiny lc};
		\node [style=none] (52) at (-1.25, 2.5) {};
		\node [style=none] (53) at (-0.75, 2.5) {};
		\node [style=none] (54) at (-1.25, 1.5) {};
		\node [style=none] (55) at (-0.75, 1.5) {};
		\node [style=none] (56) at (-1.25, -0.5) {};
		\node [style=none] (57) at (-0.75, -0.5) {};
		\node [style=none] (58) at (-1.25, -2.5) {};
		\node [style=none] (59) at (-0.75, -2.5) {};
		\node [style=none] (60) at (2.75, -1.5) {};
		\node [style=none] (61) at (3.25, -1.5) {};
		\node [style=none] (62) at (2.75, 0.5) {};
		\node [style=none] (63) at (3.25, 0.5) {};
		\node [style=none] (64) at (2.75, 2.5) {};
		\node [style=none] (65) at (3.25, 2.5) {};
		\node [style=none] (66) at (2.75, 3.5) {};
		\node [style=none] (67) at (3.25, 3.5) {};
		\node [style=none] (68) at (-1, -1.3) {$\vdots$};
		\node [style=none] (69) at (3, -0.3) {$\vdots$};
		\node [style=none] (70) at (1, 1) {};
		\node [style=none] (71) at (0, 3.5) {};
		\node [style=none] (72) at (-1, 3.5) {};
		\node [style=none] (73) at (2, 1.5) {};
		\node [style=none] (74) at (3, 1.5) {};
		\node [style=none] (75) at (0, 0.5) {};
		\node [style=none] (76) at (-1, 0.5) {};
		\node [style=none] (77) at (2, -2.5) {};
		\node [style=none] (78) at (3, -2.5) {};
		\node [style=Z dot] (79) at (2, 1.5) {};
		\node [style=Z dot] (80) at (3, 1.5) {};
		\node [style=Z dot] (81) at (0, 2.5) {};
		\node [style=Z dot] (82) at (2, 2.5) {};
		\node [style=Z dot] (83) at (0, -0.5) {};
		\node [style=Z dot] (84) at (0, 0.5) {};
		\node [style=Z dot] (85) at (0, 3.5) {};
		\node [style=Z dot] (86) at (-1, 3.5) {};
		\node [style=Z dot] (87) at (-1, 0.5) {};
		\node [style=Z dot] (88) at (2, -2.5) {};
		\node [style=Z dot] (89) at (3, -2.5) {};
		\node [style=Z dot] (90) at (2, -1.5) {};
		\node [style=X dot] (91) at (1, 1) {};
	\end{pgfonlayer}
	\begin{pgfonlayer}{edgelayer}
		\draw (15.center)
			 to (14.center)
			 to (12.center)
			 to (13.center)
			 to cycle;
		\draw (20.center)
			 to (19.center)
			 to (17.center)
			 to (18.center)
			 to cycle;
		\draw (22.center)
			 to (23.center)
			 to (25.center)
			 to (24.center)
			 to cycle;
		\draw (30.center)
			 to (29.center)
			 to (27.center)
			 to (28.center)
			 to cycle;
		\draw (35.center)
			 to (34.center)
			 to (32.center)
			 to (33.center)
			 to cycle;
		\draw (40.center)
			 to (39.center)
			 to (37.center)
			 to (38.center)
			 to cycle;
		\draw (42.center)
			 to (43.center)
			 to (45.center)
			 to (44.center)
			 to cycle;
		\draw (48.center)
			 to (50.center)
			 to (49.center)
			 to (47.center)
			 to cycle;
		\draw [style=arrow] (0.center) to (52.center);
		\draw [style=arrow] (53.center) to (81);
		\draw [style=arrow] (81) to (91);
		\draw [style=arrow] (91) to (82);
		\draw [style=arrow] (82) to (64.center);
		\draw [style=arrow] (65.center) to (5.center);
		\draw [style=arrow] (2.center) to (56.center);
		\draw [style=arrow] (57.center) to (83);
		\draw [style=arrow] (91) to (90);
		\draw [style=arrow] (90) to (60.center);
		\draw [style=arrow] (61.center) to (7.center);
		\draw [style=arrow] (83) to (91);
		\draw [style=arrow] (3.center) to (58.center);
		\draw [style=arrow] (59.center) to (88);
		\draw [style=arrow] (88) to (89);
		\draw [style=arrow] (1.center) to (54.center);
		\draw [style=arrow] (55.center) to (79);
		\draw [style=arrow] (79) to (80);
		\draw [style=arrow] (86) to (85);
		\draw [style=arrow] (85) to (66.center);
		\draw [style=arrow] (67.center) to (4.center);
		\draw [style=arrow] (87) to (84);
		\draw [style=arrow] (84) to (62.center);
		\draw [style=arrow] (63.center) to (6.center);
		\draw (79) to (82);
		\draw (81) to (85);
		\draw (83) to (84);
		\draw (88) to (90);
	\end{pgfonlayer}
\end{tikzpicture}
\]\
 In the first step, we alternate between introducing spiders before and after the measurement using \autoref{thm:elim-rewrite} and \autoref{thm:fuse-rewrite}.
 The next step provides a strict and well-covered MC flow of the expanded diagram, and the final step rearranges the diagram using OCM to match the preorder of the MC flow.

 If the Pauli has an odd weight, we have:
    
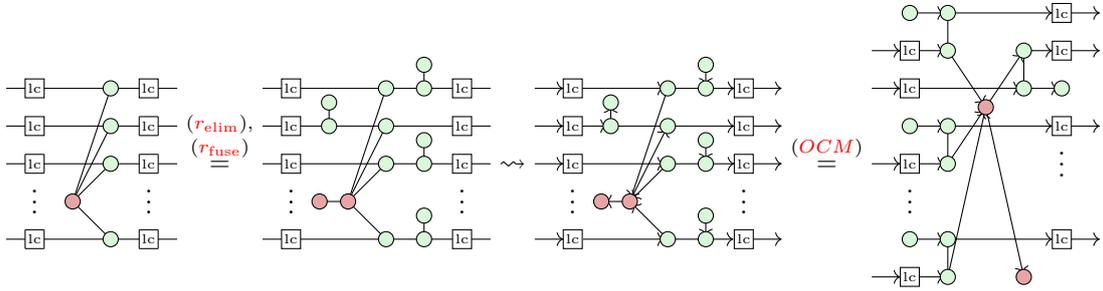
\[\begin{tikzpicture}
	\begin{pgfonlayer}{nodelayer}
		\node [style=none] (0) at (-2, 2) {};
		\node [style=none] (1) at (-2, 1) {};
		\node [style=none] (2) at (-2, 0) {};
		\node [style=none] (3) at (-2, -2) {};
		\node [style=none] (4) at (2.5, 2) {};
		\node [style=none] (5) at (2.5, 1) {};
		\node [style=none] (6) at (2.5, 0) {};
		\node [style=none] (7) at (2.5, -2) {};
		\node [style=Z dot] (8) at (0.75, 2) {};
		\node [style=Z dot] (9) at (0.75, 1) {};
		\node [style=Z dot] (10) at (0.75, 0) {};
		\node [style=Z dot] (11) at (0.75, -2) {};
		\node [style=X dot] (12) at (-0.25, -1) {};
		\node [style=none] (13) at (-1.5, 2.25) {};
		\node [style=none] (14) at (-1, 2.25) {};
		\node [style=none] (15) at (-1.5, 1.75) {};
		\node [style=none] (16) at (-1, 1.75) {};
		\node [style=none] (17) at (-1.25, 2) {\tiny lc};
		\node [style=none] (18) at (1.5, 2.25) {};
		\node [style=none] (19) at (2, 2.25) {};
		\node [style=none] (20) at (1.5, 1.75) {};
		\node [style=none] (21) at (2, 1.75) {};
		\node [style=none] (22) at (1.75, 2) {\tiny lc};
		\node [style=none] (23) at (-1.5, 1.25) {};
		\node [style=none] (24) at (-1, 1.25) {};
		\node [style=none] (25) at (-1.5, 0.75) {};
		\node [style=none] (26) at (-1, 0.75) {};
		\node [style=none] (27) at (-1.25, 1) {\tiny lc};
		\node [style=none] (28) at (1.5, 1.25) {};
		\node [style=none] (29) at (2, 1.25) {};
		\node [style=none] (30) at (1.5, 0.75) {};
		\node [style=none] (31) at (2, 0.75) {};
		\node [style=none] (32) at (1.75, 1) {\tiny lc};
		\node [style=none] (33) at (-1.5, 0.25) {};
		\node [style=none] (34) at (-1, 0.25) {};
		\node [style=none] (35) at (-1.5, -0.25) {};
		\node [style=none] (36) at (-1, -0.25) {};
		\node [style=none] (37) at (-1.25, 0) {\tiny lc};
		\node [style=none] (38) at (1.5, 0.25) {};
		\node [style=none] (39) at (2, 0.25) {};
		\node [style=none] (40) at (1.5, -0.25) {};
		\node [style=none] (41) at (2, -0.25) {};
		\node [style=none] (42) at (1.75, 0) {\tiny lc};
		\node [style=none] (43) at (-1.5, -1.75) {};
		\node [style=none] (44) at (-1, -1.75) {};
		\node [style=none] (45) at (-1.5, -2.25) {};
		\node [style=none] (46) at (-1, -2.25) {};
		\node [style=none] (47) at (-1.25, -2) {\tiny lc};
		\node [style=none] (48) at (1.5, -1.75) {};
		\node [style=none] (49) at (2, -1.75) {};
		\node [style=none] (50) at (1.5, -2.25) {};
		\node [style=none] (51) at (2, -2.25) {};
		\node [style=none] (52) at (1.75, -2) {\tiny lc};
		\node [style=none] (53) at (-1.5, 2) {};
		\node [style=none] (54) at (-1, 2) {};
		\node [style=none] (55) at (-1.5, 1) {};
		\node [style=none] (56) at (-1, 1) {};
		\node [style=none] (57) at (-1.5, 0) {};
		\node [style=none] (58) at (-1, 0) {};
		\node [style=none] (59) at (-1.5, -2) {};
		\node [style=none] (60) at (-1, -2) {};
		\node [style=none] (61) at (1.5, -2) {};
		\node [style=none] (62) at (2, -2) {};
		\node [style=none] (63) at (1.5, 0) {};
		\node [style=none] (64) at (2, 0) {};
		\node [style=none] (65) at (1.5, 1) {};
		\node [style=none] (66) at (2, 1) {};
		\node [style=none] (67) at (1.5, 2) {};
		\node [style=none] (68) at (2, 2) {};
		\node [style=none] (69) at (-1.25, -0.8) {$\vdots$};
		\node [style=none] (70) at (1.75, -0.8) {$\vdots$};
		\node [style=none] (71) at (0.75, 2.625) {};
	\end{pgfonlayer}
	\begin{pgfonlayer}{edgelayer}
		\draw (12) to (11);
		\draw (12) to (10);
		\draw (9) to (12);
		\draw (8) to (12);
		\draw (13.center)
			 to (14.center)
			 to (16.center)
			 to (15.center)
			 to cycle;
		\draw (18.center)
			 to (19.center)
			 to (21.center)
			 to (20.center)
			 to cycle;
		\draw (23.center)
			 to (24.center)
			 to (26.center)
			 to (25.center)
			 to cycle;
		\draw (28.center) to (29.center);
		\draw (29.center) to (31.center);
		\draw (31.center) to (30.center);
		\draw (30.center) to (28.center);
		\draw (33.center) to (34.center);
		\draw (34.center) to (36.center);
		\draw (36.center) to (35.center);
		\draw (35.center) to (33.center);
		\draw (38.center) to (39.center);
		\draw (39.center) to (41.center);
		\draw (41.center) to (40.center);
		\draw (40.center) to (38.center);
		\draw (43.center) to (44.center);
		\draw (44.center) to (46.center);
		\draw (46.center) to (45.center);
		\draw (45.center) to (43.center);
		\draw (48.center) to (49.center);
		\draw (49.center) to (51.center);
		\draw (51.center) to (50.center);
		\draw (50.center) to (48.center);
		\draw (0.center) to (53.center);
		\draw (54.center) to (8);
		\draw (8) to (67.center);
		\draw (68.center) to (4.center);
		\draw (1.center) to (55.center);
		\draw (56.center) to (9);
		\draw (9) to (65.center);
		\draw (66.center) to (5.center);
		\draw (2.center) to (57.center);
		\draw (58.center) to (10);
		\draw (10) to (63.center);
		\draw (64.center) to (6.center);
		\draw (3.center) to (59.center);
		\draw (60.center) to (11);
		\draw (11) to (61.center);
		\draw (62.center) to (7.center);
	\end{pgfonlayer}
\end{tikzpicture}
 \overset{\eqref{elim-rewrite},}{\overset{\eqref{fuse-rewrite}}{=}} \begin{tikzpicture}
	\begin{pgfonlayer}{nodelayer}
		\node [style=none] (0) at (-3.25, 2) {};
		\node [style=none] (1) at (-3.25, 1) {};
		\node [style=none] (2) at (-3.25, 0) {};
		\node [style=none] (3) at (-3.25, -2) {};
		\node [style=none] (4) at (2.75, 2) {};
		\node [style=none] (5) at (2.75, 1) {};
		\node [style=none] (6) at (2.75, 0) {};
		\node [style=none] (7) at (2.75, -2) {};
		\node [style=Z dot] (8) at (0, 2) {};
		\node [style=Z dot] (9) at (0, 1) {};
		\node [style=Z dot] (10) at (0, 0) {};
		\node [style=Z dot] (11) at (0, -2) {};
		\node [style=none] (12) at (-2.75, 2.25) {};
		\node [style=none] (13) at (-2.25, 2.25) {};
		\node [style=none] (14) at (-2.75, 1.75) {};
		\node [style=none] (15) at (-2.25, 1.75) {};
		\node [style=none] (16) at (-2.5, 2) {\tiny lc};
		\node [style=none] (17) at (1.75, 2.25) {};
		\node [style=none] (18) at (2.25, 2.25) {};
		\node [style=none] (19) at (1.75, 1.75) {};
		\node [style=none] (20) at (2.25, 1.75) {};
		\node [style=none] (21) at (2, 2) {\tiny lc};
		\node [style=none] (22) at (-2.75, 1.25) {};
		\node [style=none] (23) at (-2.25, 1.25) {};
		\node [style=none] (24) at (-2.75, 0.75) {};
		\node [style=none] (25) at (-2.25, 0.75) {};
		\node [style=none] (26) at (-2.5, 1) {\tiny lc};
		\node [style=none] (27) at (1.75, 1.25) {};
		\node [style=none] (28) at (2.25, 1.25) {};
		\node [style=none] (29) at (1.75, 0.75) {};
		\node [style=none] (30) at (2.25, 0.75) {};
		\node [style=none] (31) at (2, 1) {\tiny lc};
		\node [style=none] (32) at (-2.75, 0.25) {};
		\node [style=none] (33) at (-2.25, 0.25) {};
		\node [style=none] (34) at (-2.75, -0.25) {};
		\node [style=none] (35) at (-2.25, -0.25) {};
		\node [style=none] (36) at (-2.5, 0) {\tiny lc};
		\node [style=none] (37) at (1.75, 0.25) {};
		\node [style=none] (38) at (2.25, 0.25) {};
		\node [style=none] (39) at (1.75, -0.25) {};
		\node [style=none] (40) at (2.25, -0.25) {};
		\node [style=none] (41) at (2, 0) {\tiny lc};
		\node [style=none] (42) at (-2.75, -1.75) {};
		\node [style=none] (43) at (-2.25, -1.75) {};
		\node [style=none] (44) at (-2.75, -2.25) {};
		\node [style=none] (45) at (-2.25, -2.25) {};
		\node [style=none] (46) at (-2.5, -2) {\tiny lc};
		\node [style=none] (47) at (1.75, -1.75) {};
		\node [style=none] (48) at (2.25, -1.75) {};
		\node [style=none] (49) at (1.75, -2.25) {};
		\node [style=none] (50) at (2.25, -2.25) {};
		\node [style=none] (51) at (2, -2) {\tiny lc};
		\node [style=none] (52) at (-2.75, 2) {};
		\node [style=none] (53) at (-2.25, 2) {};
		\node [style=none] (54) at (-2.75, 1) {};
		\node [style=none] (55) at (-2.25, 1) {};
		\node [style=none] (56) at (-2.75, 0) {};
		\node [style=none] (57) at (-2.25, 0) {};
		\node [style=none] (58) at (-2.75, -2) {};
		\node [style=none] (59) at (-2.25, -2) {};
		\node [style=none] (60) at (1.75, -2) {};
		\node [style=none] (61) at (2.25, -2) {};
		\node [style=none] (62) at (1.75, 0) {};
		\node [style=none] (63) at (2.25, 0) {};
		\node [style=none] (64) at (1.75, 1) {};
		\node [style=none] (65) at (2.25, 1) {};
		\node [style=none] (66) at (1.75, 2) {};
		\node [style=none] (67) at (2.25, 2) {};
		\node [style=none] (68) at (-2.5, -0.8) {$\vdots$};
		\node [style=none] (69) at (2, -0.8) {$\vdots$};
		\node [style=X dot] (70) at (-1, -1) {};
		\node [style=Z dot] (71) at (1, 2) {};
		\node [style=Z dot] (72) at (1, 2.625) {};
		\node [style=Z dot] (73) at (-1.5, 1) {};
		\node [style=Z dot] (74) at (-1.5, 1.625) {};
		\node [style=Z dot] (75) at (1, 0) {};
		\node [style=Z dot] (76) at (1, 0.625) {};
		\node [style=Z dot] (77) at (1, -2) {};
		\node [style=Z dot] (78) at (1, -1.375) {};
		\node [style=X dot] (79) at (-1.75, -1) {};
	\end{pgfonlayer}
	\begin{pgfonlayer}{edgelayer}
		\draw (12.center)
			 to (13.center)
			 to (15.center)
			 to (14.center)
			 to cycle;
		\draw (20.center)
			 to (19.center)
			 to (17.center)
			 to (18.center)
			 to cycle;
		\draw (22.center)
			 to (23.center)
			 to (25.center)
			 to (24.center)
			 to cycle;
		\draw (30.center)
			 to (29.center)
			 to (27.center)
			 to (28.center)
			 to cycle;
		\draw (35.center)
			 to (34.center)
			 to (32.center)
			 to (33.center)
			 to cycle;
		\draw (40.center)
			 to (39.center)
			 to (37.center)
			 to (38.center)
			 to cycle;
		\draw (42.center)
			 to (43.center)
			 to (45.center)
			 to (44.center)
			 to cycle;
		\draw (48.center)
			 to (50.center)
			 to (49.center)
			 to (47.center)
			 to cycle;
		\draw (0.center) to (52.center);
		\draw (53.center) to (8);
		\draw (67.center) to (4.center);
		\draw (1.center) to (54.center);
		\draw (9) to (64.center);
		\draw (65.center) to (5.center);
		\draw (2.center) to (56.center);
		\draw (57.center) to (10);
		\draw (63.center) to (6.center);
		\draw (3.center) to (58.center);
		\draw (61.center) to (7.center);
		\draw (8) to (70);
		\draw (70) to (9);
		\draw (10) to (70);
		\draw (11) to (70);
		\draw (8) to (71);
		\draw (71) to (66.center);
		\draw (72) to (71);
		\draw (55.center) to (73);
		\draw (73) to (9);
		\draw (74) to (73);
		\draw (10) to (75);
		\draw (75) to (62.center);
		\draw (76) to (75);
		\draw (59.center) to (11);
		\draw (78) to (77);
		\draw (77) to (11);
		\draw (77) to (60.center);
		\draw (79) to (70);
	\end{pgfonlayer}
\end{tikzpicture}
 \rightsquigarrow \begin{tikzpicture}
	\begin{pgfonlayer}{nodelayer}
		\node [style=none] (0) at (-3.5, 2) {};
		\node [style=none] (1) at (-3.5, 1) {};
		\node [style=none] (2) at (-3.5, 0) {};
		\node [style=none] (3) at (-3.5, -2) {};
		\node [style=none] (4) at (3, 2) {};
		\node [style=none] (5) at (3, 1) {};
		\node [style=none] (6) at (3, 0) {};
		\node [style=none] (7) at (3, -2) {};
		\node [style=none] (8) at (0, 2) {};
		\node [style=none] (9) at (0, 1) {};
		\node [style=none] (10) at (0, 0) {};
		\node [style=none] (11) at (0, -2) {};
		\node [style=none] (12) at (-2.75, 2.25) {};
		\node [style=none] (13) at (-2.25, 2.25) {};
		\node [style=none] (14) at (-2.75, 1.75) {};
		\node [style=none] (15) at (-2.25, 1.75) {};
		\node [style=none] (16) at (-2.5, 2) {\tiny lc};
		\node [style=none] (17) at (1.75, 2.25) {};
		\node [style=none] (18) at (2.25, 2.25) {};
		\node [style=none] (19) at (1.75, 1.75) {};
		\node [style=none] (20) at (2.25, 1.75) {};
		\node [style=none] (21) at (2, 2) {\tiny lc};
		\node [style=none] (22) at (-2.75, 1.25) {};
		\node [style=none] (23) at (-2.25, 1.25) {};
		\node [style=none] (24) at (-2.75, 0.75) {};
		\node [style=none] (25) at (-2.25, 0.75) {};
		\node [style=none] (26) at (-2.5, 1) {\tiny lc};
		\node [style=none] (27) at (1.75, 1.25) {};
		\node [style=none] (28) at (2.25, 1.25) {};
		\node [style=none] (29) at (1.75, 0.75) {};
		\node [style=none] (30) at (2.25, 0.75) {};
		\node [style=none] (31) at (2, 1) {\tiny lc};
		\node [style=none] (32) at (-2.75, 0.25) {};
		\node [style=none] (33) at (-2.25, 0.25) {};
		\node [style=none] (34) at (-2.75, -0.25) {};
		\node [style=none] (35) at (-2.25, -0.25) {};
		\node [style=none] (36) at (-2.5, 0) {\tiny lc};
		\node [style=none] (37) at (1.75, 0.25) {};
		\node [style=none] (38) at (2.25, 0.25) {};
		\node [style=none] (39) at (1.75, -0.25) {};
		\node [style=none] (40) at (2.25, -0.25) {};
		\node [style=none] (41) at (2, 0) {\tiny lc};
		\node [style=none] (42) at (-2.75, -1.75) {};
		\node [style=none] (43) at (-2.25, -1.75) {};
		\node [style=none] (44) at (-2.75, -2.25) {};
		\node [style=none] (45) at (-2.25, -2.25) {};
		\node [style=none] (46) at (-2.5, -2) {\tiny lc};
		\node [style=none] (47) at (1.75, -1.75) {};
		\node [style=none] (48) at (2.25, -1.75) {};
		\node [style=none] (49) at (1.75, -2.25) {};
		\node [style=none] (50) at (2.25, -2.25) {};
		\node [style=none] (51) at (2, -2) {\tiny lc};
		\node [style=none] (52) at (-2.75, 2) {};
		\node [style=none] (53) at (-2.25, 2) {};
		\node [style=none] (54) at (-2.75, 1) {};
		\node [style=none] (55) at (-2.25, 1) {};
		\node [style=none] (56) at (-2.75, 0) {};
		\node [style=none] (57) at (-2.25, 0) {};
		\node [style=none] (58) at (-2.75, -2) {};
		\node [style=none] (59) at (-2.225, -2) {};
		\node [style=none] (60) at (1.725, -2) {};
		\node [style=none] (61) at (2.25, -2) {};
		\node [style=none] (62) at (1.75, 0) {};
		\node [style=none] (63) at (2.25, 0) {};
		\node [style=none] (64) at (1.75, 1) {};
		\node [style=none] (65) at (2.25, 1) {};
		\node [style=none] (66) at (1.75, 2) {};
		\node [style=none] (67) at (2.25, 2) {};
		\node [style=none] (68) at (-2.5, -0.8) {$\vdots$};
		\node [style=none] (69) at (2, -0.8) {$\vdots$};
		\node [style=none] (70) at (-1, -1) {};
		\node [style=none] (71) at (1, 2) {};
		\node [style=none] (72) at (1, 2.625) {};
		\node [style=none] (73) at (-1.5, 1) {};
		\node [style=none] (74) at (-1.5, 1.625) {};
		\node [style=none] (75) at (1, 0) {};
		\node [style=none] (76) at (1, 0.625) {};
		\node [style=Z dot] (79) at (-1.5, 1) {};
		\node [style=Z dot] (80) at (-1.5, 1.625) {};
		\node [style=Z dot] (81) at (0, 2) {};
		\node [style=Z dot] (82) at (0, 1) {};
		\node [style=Z dot] (83) at (0, 0) {};
		\node [style=Z dot] (84) at (1, 0) {};
		\node [style=Z dot] (85) at (1, 2) {};
		\node [style=Z dot] (86) at (1, 2.625) {};
		\node [style=Z dot] (87) at (1, 0.625) {};
		\node [style=Z dot] (90) at (0, -2) {};
		\node [style=X dot] (91) at (-1, -1) {};
		\node [style=none] (92) at (1, -2) {};
		\node [style=none] (93) at (1, -1.375) {};
		\node [style=Z dot] (94) at (1, -2) {};
		\node [style=Z dot] (95) at (1, -1.375) {};
		\node [style=X dot] (96) at (-1.75, -1) {};
	\end{pgfonlayer}
	\begin{pgfonlayer}{edgelayer}
		\draw (15.center)
			 to (14.center)
			 to (12.center)
			 to (13.center)
			 to cycle;
		\draw (20.center)
			 to (19.center)
			 to (17.center)
			 to (18.center)
			 to cycle;
		\draw (22.center)
			 to (23.center)
			 to (25.center)
			 to (24.center)
			 to cycle;
		\draw (30.center)
			 to (29.center)
			 to (27.center)
			 to (28.center)
			 to cycle;
		\draw (35.center)
			 to (34.center)
			 to (32.center)
			 to (33.center)
			 to cycle;
		\draw (40.center)
			 to (39.center)
			 to (37.center)
			 to (38.center)
			 to cycle;
		\draw (42.center)
			 to (43.center)
			 to (45.center)
			 to (44.center)
			 to cycle;
		\draw (48.center)
			 to (50.center)
			 to (49.center)
			 to (47.center)
			 to cycle;
		\draw [style=arrow] (0.center) to (52.center);
		\draw [style=arrow] (53.center) to (81);
		\draw [style=arrow] (81) to (91);
		\draw [style=arrow] (91) to (82);
		\draw [style=arrow] (82) to (64.center);
		\draw [style=arrow] (65.center) to (5.center);
		\draw [style=arrow] (2.center) to (56.center);
		\draw [style=arrow] (57.center) to (83);
		\draw [style=arrow] (61.center) to (7.center);
		\draw [style=arrow] (83) to (91);
		\draw [style=arrow] (3.center) to (58.center);
		\draw [style=arrow] (1.center) to (54.center);
		\draw [style=arrow] (55.center) to (79);
		\draw [style=arrow] (79) to (80);
		\draw [style=arrow] (86) to (85);
		\draw [style=arrow] (85) to (66.center);
		\draw [style=arrow] (67.center) to (4.center);
		\draw [style=arrow] (87) to (84);
		\draw [style=arrow] (84) to (62.center);
		\draw [style=arrow] (63.center) to (6.center);
		\draw (79) to (82);
		\draw (81) to (85);
		\draw (83) to (84);
		\draw [style=arrow] (59.center) to (90);
		\draw [style=arrow] (90) to (91);
		\draw [style=arrow] (95) to (94);
		\draw [style=arrow] (94) to (60.center);
		\draw (90) to (94);
		\draw [style=arrow] (91) to (96);
	\end{pgfonlayer}
\end{tikzpicture}
 \overset{(\hyperlink{eq:OCM}{OCM})}{=} \begin{tikzpicture}
	\begin{pgfonlayer}{nodelayer}
		\node [style=none] (0) at (-2, 3) {};
		\node [style=none] (1) at (-2, 2) {};
		\node [style=none] (2) at (-2, 0) {};
		\node [style=none] (4) at (4, 4) {};
		\node [style=none] (5) at (4, 3) {};
		\node [style=none] (6) at (4, 1) {};
		\node [style=none] (8) at (0, 3) {};
		\node [style=none] (9) at (2, 3) {};
		\node [style=none] (10) at (0, 0) {};
		\node [style=none] (12) at (-1.25, 3.25) {};
		\node [style=none] (13) at (-0.75, 3.25) {};
		\node [style=none] (14) at (-1.25, 2.75) {};
		\node [style=none] (15) at (-0.75, 2.75) {};
		\node [style=none] (16) at (-1, 3) {\tiny lc};
		\node [style=none] (17) at (2.75, 4.25) {};
		\node [style=none] (18) at (3.25, 4.25) {};
		\node [style=none] (19) at (2.75, 3.75) {};
		\node [style=none] (20) at (3.25, 3.75) {};
		\node [style=none] (21) at (3, 4) {\tiny lc};
		\node [style=none] (22) at (-1.25, 2.25) {};
		\node [style=none] (23) at (-0.75, 2.25) {};
		\node [style=none] (24) at (-1.25, 1.75) {};
		\node [style=none] (25) at (-0.75, 1.75) {};
		\node [style=none] (26) at (-1, 2) {\tiny lc};
		\node [style=none] (27) at (2.75, 3.25) {};
		\node [style=none] (28) at (3.25, 3.25) {};
		\node [style=none] (29) at (2.75, 2.75) {};
		\node [style=none] (30) at (3.25, 2.75) {};
		\node [style=none] (31) at (3, 3) {\tiny lc};
		\node [style=none] (32) at (-1.25, 0.25) {};
		\node [style=none] (33) at (-0.75, 0.25) {};
		\node [style=none] (34) at (-1.25, -0.25) {};
		\node [style=none] (35) at (-0.75, -0.25) {};
		\node [style=none] (36) at (-1, 0) {\tiny lc};
		\node [style=none] (37) at (2.75, 1.25) {};
		\node [style=none] (38) at (3.25, 1.25) {};
		\node [style=none] (39) at (2.75, 0.75) {};
		\node [style=none] (40) at (3.25, 0.75) {};
		\node [style=none] (41) at (3, 1) {\tiny lc};
		\node [style=none] (52) at (-1.25, 3) {};
		\node [style=none] (53) at (-0.75, 3) {};
		\node [style=none] (54) at (-1.25, 2) {};
		\node [style=none] (55) at (-0.75, 2) {};
		\node [style=none] (56) at (-1.25, 0) {};
		\node [style=none] (57) at (-0.75, 0) {};
		\node [style=none] (62) at (2.75, 1) {};
		\node [style=none] (63) at (3.25, 1) {};
		\node [style=none] (64) at (2.75, 3) {};
		\node [style=none] (65) at (3.25, 3) {};
		\node [style=none] (66) at (2.75, 4) {};
		\node [style=none] (67) at (3.25, 4) {};
		\node [style=none] (68) at (-1, -0.8) {$\vdots$};
		\node [style=none] (69) at (3, 0.2) {$\vdots$};
		\node [style=none] (71) at (0, 4) {};
		\node [style=none] (72) at (-1, 4) {};
		\node [style=none] (73) at (2, 2) {};
		\node [style=none] (74) at (3, 2) {};
		\node [style=none] (75) at (0, 1) {};
		\node [style=none] (76) at (-1, 1) {};
		\node [style=Z dot] (79) at (2, 2) {};
		\node [style=Z dot] (80) at (3, 2) {};
		\node [style=Z dot] (81) at (0, 3) {};
		\node [style=Z dot] (82) at (2, 3) {};
		\node [style=Z dot] (83) at (0, 0) {};
		\node [style=Z dot] (84) at (0, 1) {};
		\node [style=Z dot] (85) at (0, 4) {};
		\node [style=Z dot] (86) at (-1, 4) {};
		\node [style=Z dot] (87) at (-1, 1) {};
		\node [style=none] (92) at (-2, -3) {};
		\node [style=none] (93) at (0, -3) {};
		\node [style=none] (94) at (-1.25, -2.75) {};
		\node [style=none] (95) at (-0.75, -2.75) {};
		\node [style=none] (96) at (-1.25, -3.25) {};
		\node [style=none] (97) at (-0.75, -3.25) {};
		\node [style=none] (98) at (-1, -3) {\tiny lc};
		\node [style=none] (99) at (-1.25, -3) {};
		\node [style=none] (100) at (-0.75, -3) {};
		\node [style=Z dot] (101) at (0, -3) {};
		\node [style=X dot] (102) at (1, 1.5) {};
		\node [style=none] (103) at (0, -2) {};
		\node [style=none] (105) at (0, -3) {};
		\node [style=none] (106) at (2, -3) {};
		\node [style=none] (107) at (2, -3) {};
		\node [style=X dot] (108) at (2, -3) {};
		\node [style=none] (109) at (4, -2) {};
		\node [style=none] (110) at (2.75, -1.75) {};
		\node [style=none] (111) at (3.25, -1.75) {};
		\node [style=none] (112) at (2.75, -2.25) {};
		\node [style=none] (113) at (3.25, -2.25) {};
		\node [style=none] (114) at (3, -2) {\tiny lc};
		\node [style=none] (115) at (2.75, -2) {};
		\node [style=none] (116) at (3.25, -2) {};
		\node [style=none] (117) at (0, -2) {};
		\node [style=none] (118) at (-1, -2) {};
		\node [style=Z dot] (119) at (0, -2) {};
		\node [style=Z dot] (120) at (-1, -2) {};
	\end{pgfonlayer}
	\begin{pgfonlayer}{edgelayer}
		\draw (15.center)
			 to (14.center)
			 to (12.center)
			 to (13.center)
			 to cycle;
		\draw (20.center)
			 to (19.center)
			 to (17.center)
			 to (18.center)
			 to cycle;
		\draw (22.center)
			 to (23.center)
			 to (25.center)
			 to (24.center)
			 to cycle;
		\draw (30.center)
			 to (29.center)
			 to (27.center)
			 to (28.center)
			 to cycle;
		\draw (35.center)
			 to (34.center)
			 to (32.center)
			 to (33.center)
			 to cycle;
		\draw (40.center)
			 to (39.center)
			 to (37.center)
			 to (38.center)
			 to cycle;
		\draw [style=arrow] (0.center) to (52.center);
		\draw [style=arrow] (53.center) to (81);
		\draw [style=arrow] (82) to (64.center);
		\draw [style=arrow] (65.center) to (5.center);
		\draw [style=arrow] (2.center) to (56.center);
		\draw [style=arrow] (57.center) to (83);
		\draw [style=arrow] (1.center) to (54.center);
		\draw [style=arrow] (55.center) to (79);
		\draw [style=arrow] (79) to (80);
		\draw [style=arrow] (86) to (85);
		\draw [style=arrow] (85) to (66.center);
		\draw [style=arrow] (67.center) to (4.center);
		\draw [style=arrow] (87) to (84);
		\draw [style=arrow] (84) to (62.center);
		\draw [style=arrow] (63.center) to (6.center);
		\draw (79) to (82);
		\draw (81) to (85);
		\draw (83) to (84);
		\draw (97.center)
			 to (96.center)
			 to (94.center)
			 to (95.center)
			 to cycle;
		\draw [style=arrow] (100.center) to (101);
		\draw [style=arrow] (81) to (102);
		\draw [style=arrow] (102) to (82);
		\draw [style=arrow] (83) to (102);
		\draw [style=arrow] (102) to (108);
		\draw [style=arrow] (92.center) to (99.center);
		\draw [style=arrow] (105.center) to (102);
		\draw (103.center) to (105.center);
		\draw (113.center)
			 to (112.center)
			 to (110.center)
			 to (111.center)
			 to cycle;
		\draw [style=arrow] (120) to (119);
		\draw [style=arrow] (119) to (115.center);
		\draw [style=arrow] (116.center) to (109.center);
	\end{pgfonlayer}
\end{tikzpicture}
\]
 Once again, we alternate between introducing spiders before and after the measurement. 
 The measurement is of odd degree, thus, we introduce one more set of spiders after the measurement than before. 
 As a well-covered MC flow does not allow paths to end in high-degree spiders, we have to make sure that the red spider has an even degree.
 Therefore, we use \autoref{thm:fuse-rewrite} to introduce a red spider of degree one. 
 This allows us to provide a well-covered and strict MC flow.

 Thus, we have shown that any Pauli measurement can be brought into a form with MC flow using only distance-preserving rewrites. 
\end{proof}

\noindent But then, by \autoref{prop:flow-correct} and \autoref{prop:flow-correct-3}, we know that for any Pauli measurement, we can extract a quantum circuit in low-weight Clifford and measurement form using only distance-preserving rewrites. 

To showcase the proposition, we apply it to a Pauli measurement with weight four and one with weight nine
\footnote{
 The weight-four rewrite is similar to the rewrite introduced by~\textcite{townsend-teagueFloquetifyingColourCode2023}.
 The only difference is that~\textcite{townsend-teagueFloquetifyingColourCode2023} do not use a distance-preserving implementation of the four-legged spider; as such, they only have one weight-two Pauli-X measurement.
 If one assumes this measurement to be fault-free, i.e.\@ no errors are possible on its measurement edge, this is distance-preserving. 
 Since, in this work, we do not, a second Pauli-X measurement is necessary.
 For more details, see \autoref{appendix:floquetification-teague}.}.
\begin{example}[Distance-preserving implementation of a weight-four Pauli measurement]
    
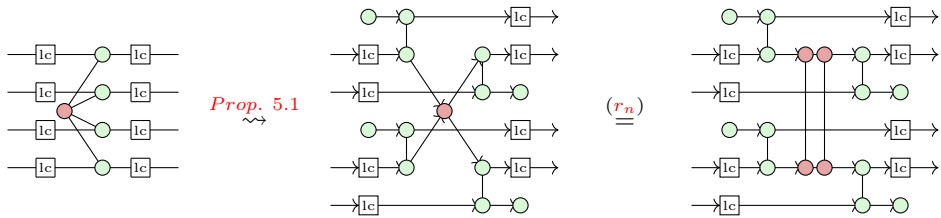
\label{ex:4-qubit-proj}
    \[\input{projectors/weight-four.tikz}\]
 From which, using \autoref{prop:flow-correct-3}, we can extract the following quantum circuit:
    \[\input{projectors/weight-four-circuit.tikz}\]
\end{example}
\noindent We observe that by the translation table in \autoref{subsec:zx-calculus}, we would expect a $XX$ measurement to consist of two green spiders conjugated by Hadamards. 
However, via the \TextColour-rule, we have: 
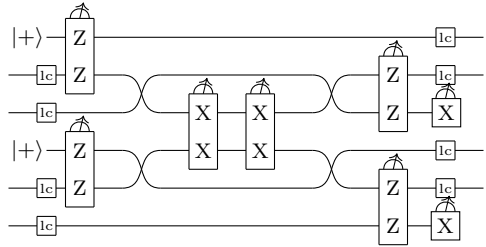
\[\input{XX-alternative.tikz}\]
We use this alternative representation of an $XX$ gate throughout this work to reduce clutter. 

\noindent For $n = 9$, we get:
\begin{example}[Distance-preserving implementation of a weight-nine Pauli measurement]
    \[\input{projectors/weight-nine-1.tikz} \quad \overset{\autorf{prop:decomposing-measurements}}{\rightsquigarrow} \quad \input{projectors/weight-nine-2.tikz}\]
 Applying \autoref{prop:flow-correct}, we can expand the definition of the 9-legged spider to get:
    \[\input{projectors/weight-nine-3.tikz}\]
\end{example}

As such, we have provided a procedure to implement arbitrary weight Pauli measurements in terms of single-qubit local unitaries, preparations and measurements and weight-two Pauli measurements. 
By only using distance-preserving rewrites, we can guarantee that a single, undetectable error in the new circuit creates at most a single data error.
However, this comes at the cost of needing auxiliary qubits.

As observed in \autoref{sec:classifying-overhead}, the number of qubits required to implement the ZX diagram is upper-bounded by the number of paths in the MC flow. 
The number of qubits required by this implementation is equal to the number of paths introduced in \autoref{prop:decomposing-measurements} plus the number of additional paths required to implement the high-weight measurement spider. 
The number of paths introduced in \autoref{prop:decomposing-measurements} for a weight-$w$ measurement is 
$w + \lceil \frac{w}{2} \rceil$; one for each of the $w$ inputs and one for each of the $\lceil \frac{w}{2} \rceil$ spider gadgets introduced after the measurement. 
The number of paths to implement the degree-$w$ measurement spider is given by $f(w) \leq \log_2(w)$; see \autoref{sec:classifying-overhead}.
Thus, overall, to implement a weight-$w$ Pauli measurement, this construction needs at most $w + \lceil \frac{w}{2} \rceil + \log_2(w)$ qubits.

\section{Floquetification of stabiliser codes}
\label{sec:rewriting}
Abstractly, the Floquetification procedure introduced by~\textcite{townsend-teagueFloquetifyingColourCode2023} consists of three steps: (1) writing out the measurement circuit of a code as a ZX diagram, (2) using ZX rewrites to rearrange the diagram, (3) extracting a new code. 

We propose a generalised Floquetification procedure for arbitrary stabiliser codes. 
Instead of simultaneously considering the entire measurement circuit, we perform step (2) --- the rearranging of the ZX diagram --- in two substeps: (2.1) rearranging individual Pauli measurements using the distance-preserving decompositions from \autoref{sec:rewriting-measurements}, (2.2) composing the Pauli measurements. 
The proposed procedure can be described by \autoref{alg:transformation}.

\begin{algorithm}[!ht]
   \caption{Algorithm for Transforming Stabiliser Codes}
   \label{alg:transformation}
   \begin{algorithmic}[1]
       \State \textbf{Input:} Stabiliser code $S$
       \State Write out the measurement circuit in ZX notation according to \autoref{fig:circuit-to-zx} 
       \Comment{Step 1}
       \For{Pauli measurement $S_i$ of $S$}
       \Comment{Step 2.1}
           \State Replace $S_i$ by its distance-preserving implementation
       \EndFor
       \State Remove mid-circuit state preparations and measurements using distance-preserving rewrites
       \Comment{Step 2.2}
       \State Extract the new code
       \Comment{Step 3}
   \end{algorithmic}
\end{algorithm}

\subsection{Infinity notation}
\label{sec:infinitely-repeating}
To reason about the infinite measurement circuits of Floquet codes, we introduce a new notation which captures the idea of infinitely performing an operation on a given state.

\begin{definition}[Infinite post-application]
 Let $f: X \to X$ be a process. Then, we inductively define:
   \[\input{prerequisites/infty-definition.tikz}\]
 meaning that $f$ is performed infinitely many times on the input state.
\end{definition}

\noindent We observe that it is possible to place processes before the infinite post-application gadget, i.e.\@ perform operations on the input state before passing it to the infinitely repeating circuit.
However, it is not possible to place anything afterwards. 

This notation comes with some equalities:
\refstepcounter{equation}
\begin{gather}
 \tag{reorder}\label{infty-rewrite-reorder}\refstepcounter{equation}
 \input{prerequisites/infty-rewrite-reorder.tikz} \\
 \input{prerequisites/infty-rewrite-unroll.tikz}
 \tag{unroll}\label{infty-rewrite-unroll}\refstepcounter{equation}
\end{gather}

\noindent The former corresponds to the fact that infinitely repeating $g \circ f$ corresponds to first doing one $f$ and then infinitely repeating $f \circ g$.
The latter corresponds to the fact that infinitely repeating $f$ is the same as infinitely repeating $k \in \mathbb{N}$ many consecutive instances of $f$.

\subsection{Example: The \texorpdfstring{$\code{4, 2, 2}$}{{[4,2,2]}} code}
In this section, we will explain the details of the proposed Floquetification procedure using the example of the $\code{4, 2, 2}$ code.
However, this procedure applies to arbitrary stabiliser codes.
\subsubsection{Step 1 --- Writing out the measurement circuit}
Using the newly introduced notation, we can write out the measurement circuit of the code as a ZX diagram.
For the $\code{4, 2, 2}$ code we get: 
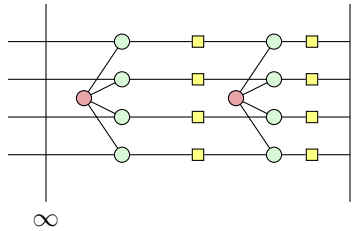
\[\input{floquetification-2/422-operation-circuit.tikz}\]
By \autoref{prop:distance-correspondence}, after establishment, this ZX diagram has ZX distance $2$. 

\subsubsection{Step 2.1 --- Rewriting measurements}
As the first part of the second step, we rewrite the measurements using their distance-preserving implementation provided in \autoref{sec:rewriting-measurements}.
Applying this procedure to our example, we get:
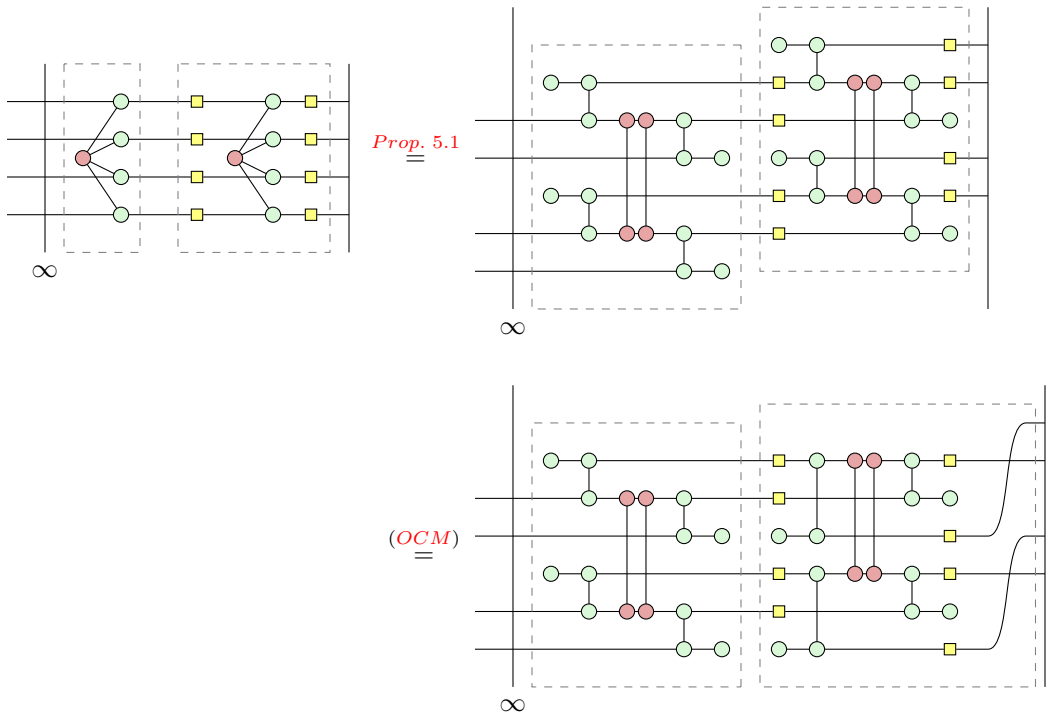
\setlength{\jot}{18pt}
\begin{align*}
 \begin{tikzpicture}
	\begin{pgfonlayer}{nodelayer}
		\node [style=none] (0) at (-4, 1.5) {};
		\node [style=none] (1) at (-4, 0.5) {};
		\node [style=none] (2) at (-4, -0.5) {};
		\node [style=none] (3) at (-4, -1.5) {};
		\node [style=none] (8) at (5, 1.5) {};
		\node [style=none] (9) at (5, 0.5) {};
		\node [style=none] (10) at (5, -0.5) {};
		\node [style=none] (11) at (5, -1.5) {};
		\node [style=Z dot] (12) at (-1, 1.5) {};
		\node [style=Z dot] (13) at (-1, 0.5) {};
		\node [style=Z dot] (14) at (-1, -0.5) {};
		\node [style=Z dot] (15) at (-1, -1.5) {};
		\node [style=Z dot] (16) at (3, 1.5) {};
		\node [style=Z dot] (17) at (3, 0.5) {};
		\node [style=Z dot] (18) at (3, -0.5) {};
		\node [style=Z dot] (19) at (3, -1.5) {};
		\node [style=H box] (20) at (1, 1.5) {};
		\node [style=H box] (21) at (1, 0.5) {};
		\node [style=H box] (22) at (1, -0.5) {};
		\node [style=H box] (23) at (1, -1.5) {};
		\node [style=H box] (24) at (4, 1.5) {};
		\node [style=H box] (25) at (4, 0.5) {};
		\node [style=H box] (26) at (4, -0.5) {};
		\node [style=H box] (27) at (4, -1.5) {};
		\node [style=X dot] (28) at (-2, 0) {};
		\node [style=X dot] (29) at (2, 0) {};
		\node [style=none] (31) at (-0.5, 2) {};
		\node [style=none] (32) at (-2.5, 2) {};
		\node [style=none] (33) at (-2.5, -2) {};
		\node [style=none] (34) at (-0.5, -2) {};
		\node [style=none] (35) at (4.5, 2) {};
		\node [style=none] (36) at (0.5, 2) {};
		\node [style=none] (37) at (0.5, -2) {};
		\node [style=none] (38) at (4.5, -2) {};
		\node [style=none] (39) at (-3.05, -2.75) {$\infty$};
		\node [style=none] (40) at (5, -2.25) {};
		\node [style=none] (41) at (5, 2.25) {};
		\node [style=none] (42) at (5.075, -2.25) {};
		\node [style=none] (43) at (5.075, 2.25) {};
		\node [style=none] (44) at (4.875, 0.125) {.};
		\node [style=none] (45) at (4.875, -0.125) {.};
		\node [style=none] (46) at (-2.975, -2.25) {};
		\node [style=none] (47) at (-2.975, 2.25) {};
		\node [style=none] (48) at (-3.05, -2.25) {};
		\node [style=none] (49) at (-3.05, 2.25) {};
		\node [style=none] (50) at (-2.85, 0.125) {.};
		\node [style=none] (51) at (-2.85, -0.125) {.};
	\end{pgfonlayer}
	\begin{pgfonlayer}{edgelayer}
		\draw (16) to (29);
		\draw (29) to (17);
		\draw (19) to (29);
		\draw (29) to (18);
		\draw (23) to (19);
		\draw (22) to (18);
		\draw (18) to (26);
		\draw (21) to (17);
		\draw (17) to (25);
		\draw (20) to (16);
		\draw (16) to (24);
		\draw (19) to (27);
		\draw (23) to (15);
		\draw (14) to (22);
		\draw (21) to (13);
		\draw (12) to (20);
		\draw (0.center) to (12);
		\draw (13) to (1.center);
		\draw (2.center) to (14);
		\draw (15) to (3.center);
		\draw (12) to (28);
		\draw (28) to (13);
		\draw (14) to (28);
		\draw (15) to (28);
		\draw (27) to (11.center);
		\draw (26) to (10.center);
		\draw (25) to (9.center);
		\draw (24) to (8.center);
		\draw [style=dotted] (32.center)
			 to (33.center)
			 to (34.center)
			 to (31.center)
			 to cycle;
		\draw [style=dotted] (36.center)
			 to (37.center)
			 to (38.center)
			 to (35.center)
			 to cycle;
		\draw (41.center) to (40.center);
		\draw [style=thick] (43.center) to (42.center);
		\draw (47.center) to (46.center);
		\draw [style=thick] (49.center) to (48.center);
	\end{pgfonlayer}
\end{tikzpicture}

  \ \ \overset{\autorf{prop:decomposing-measurements}}{=}\ \ &\begin{tikzpicture}
	\begin{pgfonlayer}{nodelayer}
		\node [style=none] (0) at (-5, 1) {};
		\node [style=none] (1) at (-2.5, 1) {};
		\node [style=none] (2) at (-5, 2) {};
		\node [style=none] (3) at (-6, 2) {};
		\node [style=none] (4) at (-2.5, 0) {};
		\node [style=none] (5) at (-1.5, 0) {};
		\node [style=Z dot] (6) at (-2.5, 0) {};
		\node [style=Z dot] (7) at (-1.5, 0) {};
		\node [style=Z dot] (8) at (-5, 1) {};
		\node [style=Z dot] (9) at (-2.5, 1) {};
		\node [style=Z dot] (10) at (-5, 2) {};
		\node [style=Z dot] (11) at (-6, 2) {};
		\node [style=none] (12) at (-4, 1) {};
		\node [style=none] (13) at (-3.5, 1) {};
		\node [style=X dot] (14) at (-4, 1) {};
		\node [style=X dot] (15) at (-3.5, 1) {};
		\node [style=none] (16) at (-5, -2) {};
		\node [style=none] (17) at (-2.5, -2) {};
		\node [style=none] (18) at (-5, -1) {};
		\node [style=none] (19) at (-6, -1) {};
		\node [style=none] (20) at (-2.5, -3) {};
		\node [style=none] (21) at (-1.5, -3) {};
		\node [style=Z dot] (22) at (-2.5, -3) {};
		\node [style=Z dot] (23) at (-1.5, -3) {};
		\node [style=Z dot] (24) at (-5, -2) {};
		\node [style=Z dot] (25) at (-2.5, -2) {};
		\node [style=Z dot] (26) at (-5, -1) {};
		\node [style=Z dot] (27) at (-6, -1) {};
		\node [style=none] (28) at (-4, -2) {};
		\node [style=none] (29) at (-3.5, -2) {};
		\node [style=X dot] (30) at (-4, -2) {};
		\node [style=X dot] (31) at (-3.5, -2) {};
		\node [style=none] (32) at (-8, 1) {};
		\node [style=none] (33) at (-8, 0) {};
		\node [style=none] (34) at (-8, -2) {};
		\node [style=none] (35) at (-8, -3) {};
		\node [style=H box] (36) at (0, 2) {};
		\node [style=H box] (37) at (0, 1) {};
		\node [style=H box] (38) at (0, -1) {};
		\node [style=H box] (39) at (0, -2) {};
		\node [style=none] (43) at (1, 2) {};
		\node [style=none] (44) at (3.5, 2) {};
		\node [style=none] (45) at (1, 3) {};
		\node [style=none] (46) at (0, 3) {};
		\node [style=none] (47) at (3.5, 1) {};
		\node [style=none] (48) at (4.5, 1) {};
		\node [style=Z dot] (49) at (3.5, 1) {};
		\node [style=Z dot] (50) at (4.5, 1) {};
		\node [style=Z dot] (51) at (1, 2) {};
		\node [style=Z dot] (52) at (3.5, 2) {};
		\node [style=Z dot] (53) at (1, 3) {};
		\node [style=Z dot] (54) at (0, 3) {};
		\node [style=none] (55) at (2, 2) {};
		\node [style=none] (56) at (2.5, 2) {};
		\node [style=X dot] (57) at (2, 2) {};
		\node [style=X dot] (58) at (2.5, 2) {};
		\node [style=none] (59) at (1, -1) {};
		\node [style=none] (60) at (3.5, -1) {};
		\node [style=none] (61) at (1, 0) {};
		\node [style=none] (62) at (0, 0) {};
		\node [style=none] (63) at (3.5, -2) {};
		\node [style=none] (64) at (4.5, -2) {};
		\node [style=Z dot] (65) at (3.5, -2) {};
		\node [style=Z dot] (66) at (4.5, -2) {};
		\node [style=Z dot] (67) at (1, -1) {};
		\node [style=Z dot] (68) at (3.5, -1) {};
		\node [style=Z dot] (69) at (1, 0) {};
		\node [style=Z dot] (70) at (0, 0) {};
		\node [style=none] (71) at (2, -1) {};
		\node [style=none] (72) at (2.5, -1) {};
		\node [style=X dot] (73) at (2, -1) {};
		\node [style=X dot] (74) at (2.5, -1) {};
		\node [style=H box] (75) at (4.5, 2) {};
		\node [style=H box] (76) at (4.5, 3) {};
		\node [style=H box] (77) at (4.5, 0) {};
		\node [style=H box] (78) at (4.5, -1) {};
		\node [style=none] (81) at (5.5, 3) {};
		\node [style=none] (82) at (5.5, 2) {};
		\node [style=none] (83) at (5.5, 0) {};
		\node [style=none] (84) at (5.5, -1) {};
		\node [style=none] (85) at (-6.5, 2.5) {};
		\node [style=none] (86) at (-1, 2.5) {};
		\node [style=none] (87) at (-1, -3.5) {};
		\node [style=none] (88) at (-6.5, -3.5) {};
		\node [style=none] (89) at (-0.5, 3.5) {};
		\node [style=none] (90) at (-0.5, -2.5) {};
		\node [style=none] (91) at (5, -2.5) {};
		\node [style=none] (92) at (5, 3.5) {};
		\node [style=none] (94) at (5.5, -3.75) {};
		\node [style=none] (95) at (5.5, 3.75) {};
		\node [style=none] (96) at (5.575, -3.75) {};
		\node [style=none] (97) at (5.575, 3.75) {};
		\node [style=none] (98) at (5.375, 0.125) {.};
		\node [style=none] (99) at (5.375, -0.125) {.};
		\node [style=none] (100) at (-7.375, -4.25) {$\infty$};
		\node [style=none] (101) at (-7.3, -3.75) {};
		\node [style=none] (102) at (-7.3, 3.75) {};
		\node [style=none] (103) at (-7.375, -3.75) {};
		\node [style=none] (104) at (-7.375, 3.75) {};
		\node [style=none] (105) at (-7.175, 0.125) {.};
		\node [style=none] (106) at (-7.175, -0.125) {.};
	\end{pgfonlayer}
	\begin{pgfonlayer}{edgelayer}
		\draw (11) to (10);
		\draw (10) to (8);
		\draw (15) to (9);
		\draw (8) to (14);
		\draw (14) to (15);
		\draw (9) to (6);
		\draw (6) to (7);
		\draw (27) to (26);
		\draw (26) to (24);
		\draw (31) to (25);
		\draw (24) to (30);
		\draw (30) to (31);
		\draw (25) to (22);
		\draw (22) to (23);
		\draw (14) to (30);
		\draw (31) to (15);
		\draw (32.center) to (8);
		\draw (33.center) to (6);
		\draw (34.center) to (24);
		\draw (35.center) to (22);
		\draw (10) to (36);
		\draw (9) to (37);
		\draw (26) to (38);
		\draw (25) to (39);
		\draw (54) to (53);
		\draw (53) to (51);
		\draw (58) to (52);
		\draw (51) to (57);
		\draw (57) to (58);
		\draw (52) to (49);
		\draw (49) to (50);
		\draw (70) to (69);
		\draw (69) to (67);
		\draw (74) to (68);
		\draw (67) to (73);
		\draw (73) to (74);
		\draw (68) to (65);
		\draw (65) to (66);
		\draw (57) to (73);
		\draw (74) to (58);
		\draw (36) to (51);
		\draw (37) to (49);
		\draw (38) to (67);
		\draw (39) to (65);
		\draw (69) to (77);
		\draw (68) to (78);
		\draw (52) to (75);
		\draw (53) to (76);
		\draw (76) to (81.center);
		\draw (75) to (82.center);
		\draw (77) to (83.center);
		\draw (78) to (84.center);
		\draw [style=dotted] (86.center)
			 to (85.center)
			 to (88.center)
			 to (87.center)
			 to cycle;
		\draw [style=dotted] (91.center)
			 to (92.center)
			 to (89.center)
			 to (90.center)
			 to cycle;
		\draw (95.center) to (94.center);
		\draw [style=thick] (97.center) to (96.center);
		\draw (102.center) to (101.center);
		\draw [style=thick] (104.center) to (103.center);
	\end{pgfonlayer}
\end{tikzpicture}
\\
 \overset{(\hyperlink{eq:OCM}{OCM})}{=}\ \ &\begin{tikzpicture}
	\begin{pgfonlayer}{nodelayer}
		\node [style=none] (0) at (-4, 1.25) {};
		\node [style=none] (1) at (-1.5, 1.25) {};
		\node [style=none] (2) at (-4, 2.25) {};
		\node [style=none] (3) at (-5, 2.25) {};
		\node [style=none] (4) at (-1.5, 0.25) {};
		\node [style=none] (5) at (-0.5, 0.25) {};
		\node [style=Z dot] (6) at (-1.5, 0.25) {};
		\node [style=Z dot] (7) at (-0.5, 0.25) {};
		\node [style=Z dot] (8) at (-4, 1.25) {};
		\node [style=Z dot] (9) at (-1.5, 1.25) {};
		\node [style=Z dot] (10) at (-4, 2.25) {};
		\node [style=Z dot] (11) at (-5, 2.25) {};
		\node [style=none] (12) at (-3, 1.25) {};
		\node [style=none] (13) at (-2.5, 1.25) {};
		\node [style=X dot] (14) at (-3, 1.25) {};
		\node [style=X dot] (15) at (-2.5, 1.25) {};
		\node [style=none] (16) at (-4, -1.75) {};
		\node [style=none] (17) at (-1.5, -1.75) {};
		\node [style=none] (18) at (-4, -0.75) {};
		\node [style=none] (19) at (-5, -0.75) {};
		\node [style=none] (20) at (-1.5, -2.75) {};
		\node [style=none] (21) at (-0.5, -2.75) {};
		\node [style=Z dot] (22) at (-1.5, -2.75) {};
		\node [style=Z dot] (23) at (-0.5, -2.75) {};
		\node [style=Z dot] (24) at (-4, -1.75) {};
		\node [style=Z dot] (25) at (-1.5, -1.75) {};
		\node [style=Z dot] (26) at (-4, -0.75) {};
		\node [style=Z dot] (27) at (-5, -0.75) {};
		\node [style=none] (28) at (-3, -1.75) {};
		\node [style=none] (29) at (-2.5, -1.75) {};
		\node [style=X dot] (30) at (-3, -1.75) {};
		\node [style=X dot] (31) at (-2.5, -1.75) {};
		\node [style=none] (32) at (-7, 1.25) {};
		\node [style=none] (33) at (-7, 0.25) {};
		\node [style=none] (34) at (-7, -1.75) {};
		\node [style=none] (35) at (-7, -2.75) {};
		\node [style=H box] (36) at (1, 2.25) {};
		\node [style=H box] (37) at (1, 1.25) {};
		\node [style=H box] (38) at (1, -0.75) {};
		\node [style=H box] (39) at (1, -1.75) {};
		\node [style=none] (43) at (2, 2.25) {};
		\node [style=none] (44) at (4.5, 2.25) {};
		\node [style=none] (45) at (2, 0.25) {};
		\node [style=none] (46) at (1, 0.25) {};
		\node [style=none] (47) at (4.5, 1.25) {};
		\node [style=none] (48) at (5.5, 1.25) {};
		\node [style=Z dot] (49) at (4.5, 1.25) {};
		\node [style=Z dot] (50) at (5.5, 1.25) {};
		\node [style=Z dot] (51) at (2, 2.25) {};
		\node [style=Z dot] (52) at (4.5, 2.25) {};
		\node [style=Z dot] (53) at (2, 0.25) {};
		\node [style=Z dot] (54) at (1, 0.25) {};
		\node [style=none] (55) at (3, 2.25) {};
		\node [style=none] (56) at (3.5, 2.25) {};
		\node [style=X dot] (57) at (3, 2.25) {};
		\node [style=X dot] (58) at (3.5, 2.25) {};
		\node [style=none] (59) at (2, -0.75) {};
		\node [style=none] (60) at (4.5, -0.75) {};
		\node [style=none] (61) at (2, -2.75) {};
		\node [style=none] (62) at (1, -2.75) {};
		\node [style=none] (63) at (4.5, -1.75) {};
		\node [style=none] (64) at (5.5, -1.75) {};
		\node [style=Z dot] (65) at (4.5, -1.75) {};
		\node [style=Z dot] (66) at (5.5, -1.75) {};
		\node [style=Z dot] (67) at (2, -0.75) {};
		\node [style=Z dot] (68) at (4.5, -0.75) {};
		\node [style=Z dot] (69) at (2, -2.75) {};
		\node [style=Z dot] (70) at (1, -2.75) {};
		\node [style=none] (71) at (3, -0.75) {};
		\node [style=none] (72) at (3.5, -0.75) {};
		\node [style=X dot] (73) at (3, -0.75) {};
		\node [style=X dot] (74) at (3.5, -0.75) {};
		\node [style=H box] (75) at (5.5, 2.25) {};
		\node [style=H box] (76) at (5.5, 0.25) {};
		\node [style=H box] (77) at (5.5, -2.75) {};
		\node [style=H box] (78) at (5.5, -0.75) {};
		\node [style=none] (81) at (7.5, 3.25) {};
		\node [style=none] (82) at (7.5, 2.25) {};
		\node [style=none] (83) at (7.5, 0.25) {};
		\node [style=none] (84) at (7.5, -0.75) {};
		\node [style=none] (85) at (6.5, 2.25) {};
		\node [style=none] (86) at (6.5, 0.25) {};
		\node [style=none] (87) at (6.5, -0.75) {};
		\node [style=none] (88) at (6.5, -2.75) {};
		\node [style=none] (89) at (-5.5, 3.75) {};
		\node [style=none] (90) at (0, 3.75) {};
		\node [style=none] (91) at (0, -3.25) {};
		\node [style=none] (92) at (-5.5, -3.25) {};
		\node [style=none] (93) at (0.5, 3.75) {};
		\node [style=none] (94) at (7.75, 3.75) {};
		\node [style=none] (95) at (7.75, -3.25) {};
		\node [style=none] (96) at (0.5, -3.25) {};
		\node [style=none] (97) at (8, 3.25) {};
		\node [style=none] (98) at (8, 2.25) {};
		\node [style=none] (99) at (8, 0.25) {};
		\node [style=none] (100) at (8, -0.75) {};
		\node [style=none] (101) at (8, -3.5) {};
		\node [style=none] (102) at (8, 4) {};
		\node [style=none] (103) at (8.075, -3.5) {};
		\node [style=none] (104) at (8.075, 4) {};
		\node [style=none] (105) at (7.875, 0.375) {.};
		\node [style=none] (106) at (7.875, 0.125) {.};
		\node [style=none] (107) at (-6.375, -4) {$\infty$};
		\node [style=none] (108) at (-6.3, -3.5) {};
		\node [style=none] (109) at (-6.3, 4) {};
		\node [style=none] (110) at (-6.375, -3.5) {};
		\node [style=none] (111) at (-6.375, 4) {};
		\node [style=none] (112) at (-6.175, 0.375) {.};
		\node [style=none] (113) at (-6.175, 0.125) {.};
	\end{pgfonlayer}
	\begin{pgfonlayer}{edgelayer}
		\draw (11) to (10);
		\draw (10) to (8);
		\draw (15) to (9);
		\draw (8) to (14);
		\draw (14) to (15);
		\draw (9) to (6);
		\draw (6) to (7);
		\draw (27) to (26);
		\draw (26) to (24);
		\draw (31) to (25);
		\draw (24) to (30);
		\draw (30) to (31);
		\draw (25) to (22);
		\draw (22) to (23);
		\draw (14) to (30);
		\draw (31) to (15);
		\draw (32.center) to (8);
		\draw (33.center) to (6);
		\draw (34.center) to (24);
		\draw (35.center) to (22);
		\draw (10) to (36);
		\draw (9) to (37);
		\draw (26) to (38);
		\draw (25) to (39);
		\draw (54) to (53);
		\draw (53) to (51);
		\draw (58) to (52);
		\draw (51) to (57);
		\draw (57) to (58);
		\draw (52) to (49);
		\draw (49) to (50);
		\draw (70) to (69);
		\draw (69) to (67);
		\draw (74) to (68);
		\draw (67) to (73);
		\draw (73) to (74);
		\draw (68) to (65);
		\draw (65) to (66);
		\draw (57) to (73);
		\draw (74) to (58);
		\draw (36) to (51);
		\draw (37) to (49);
		\draw (38) to (67);
		\draw (39) to (65);
		\draw (69) to (77);
		\draw (68) to (78);
		\draw (52) to (75);
		\draw (53) to (76);
		\draw (78)
			 to (87.center)
			 to [in=180, out=0] (84.center)
			 to (100.center);
		\draw (77)
			 to (88.center)
			 to [in=-180, out=0, looseness=0.50] (83.center)
			 to (99.center);
		\draw (75)
			 to (85.center)
			 to [in=180, out=0] (82.center)
			 to (98.center);
		\draw (76)
			 to (86.center)
			 to [in=-180, out=0, looseness=0.50] (81.center)
			 to (97.center);
		\draw [style=dotted] (91.center)
			 to (90.center)
			 to (89.center)
			 to (92.center)
			 to cycle;
		\draw [style=dotted] (95.center)
			 to (94.center)
			 to (93.center)
			 to (96.center)
			 to cycle;
		\draw (102.center) to (101.center);
		\draw [style=thick] (104.center) to (103.center);
		\draw (109.center) to (108.center);
		\draw [style=thick] (111.center) to (110.center);
	\end{pgfonlayer}
\end{tikzpicture}

\end{align*}
where we use the dotted lines to indicate the two Pauli measurements of the $\code{4, 2, 2}$ code.
In the first step, we expand the definition of the Pauli measurement. In the next step, we use \textit{Only Connectivity Matters} to rearrange the diagram by pulling down the two single-qubit preparations.

As the highest number of qubits at any point in time is six, we can read this as an implementation of the $\code{4, 2, 2}$ code using two auxiliary qubits. 
However, the resulting circuit has single-qubit preparations and destructive measurements and is therefore not a Floquet code. 
To complete the Floquetification procedure, we need to remove them.

\subsubsection{Step 2.2 --- Composing measurements}
To remove the mid-circuit measurements and preparations, we will introduce one more distance-preserving rewrite:

\begin{restatable}{theorem}{singleQubitPauli}
  \label{thm:single-qubit-pauli}
 The following rewrite is distance-preserving:
  \begin{gather}
 \tag{$r_{\text{Pauli-1}}$}\label{one-pauli}\refstepcounter{equation}
 \input{floquetification/rewrite.tikz}
  \end{gather}
\end{restatable}
\begin{proof}
 See \autoref{appendix:distance-preserving-rewrites}
\end{proof}

\noindent This rewrite has the following measurement-circuit flow.
It can be read as two consecutive weight-one Pauli measurements of type Z:
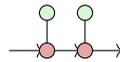
\[\input{floquetification/rewrite-interpretation.tikz}\]

\noindent Using this implementation, we can merge single-qubit measurements with single-qubit preparations:
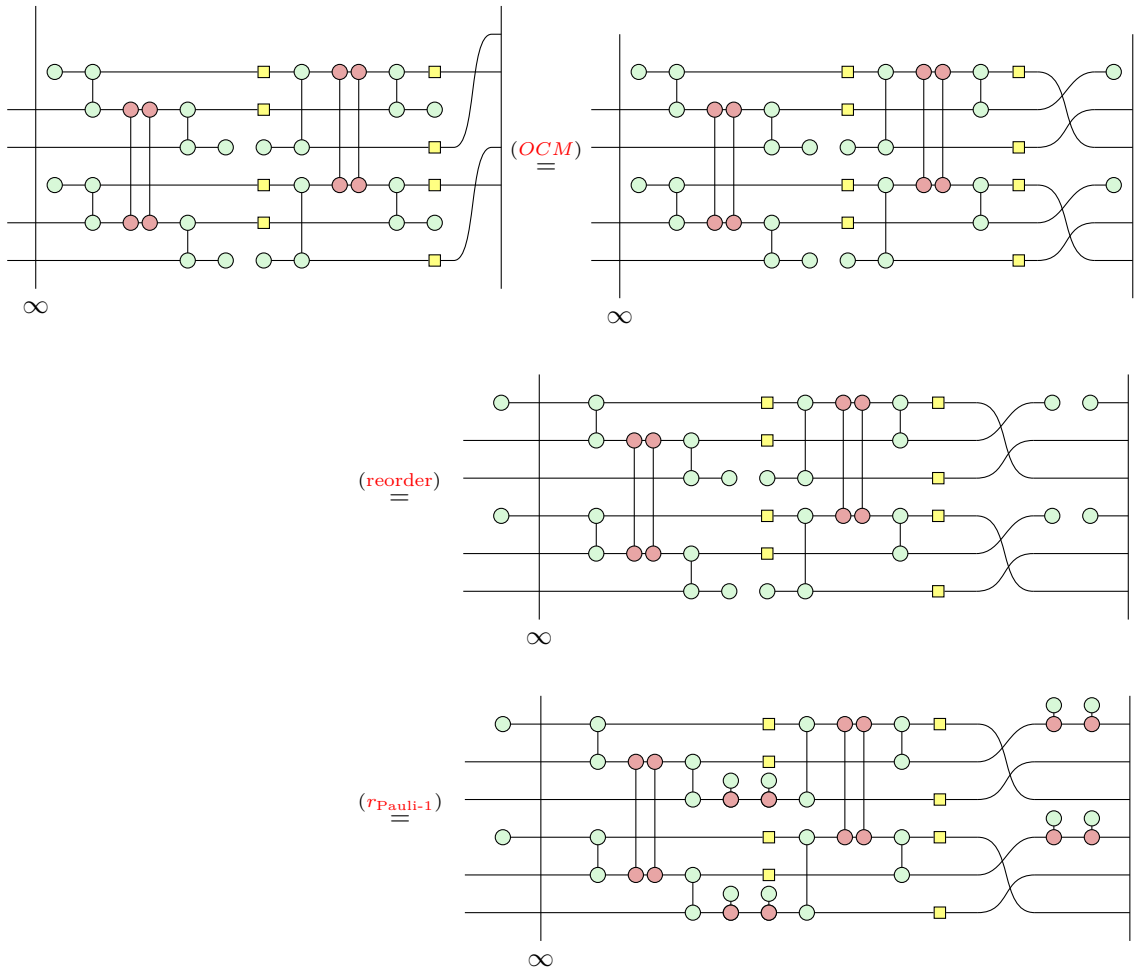
\begin{align*}
 \begin{tikzpicture}
	\begin{pgfonlayer}{nodelayer}
		\node [style=none] (0) at (-4, 1.5) {};
		\node [style=none] (1) at (-1.5, 1.5) {};
		\node [style=none] (2) at (-4, 2.5) {};
		\node [style=none] (3) at (-5, 2.5) {};
		\node [style=none] (4) at (-1.5, 0.5) {};
		\node [style=none] (5) at (-0.5, 0.5) {};
		\node [style=Z dot] (6) at (-1.5, 0.5) {};
		\node [style=Z dot] (7) at (-0.5, 0.5) {};
		\node [style=Z dot] (8) at (-4, 1.5) {};
		\node [style=Z dot] (9) at (-1.5, 1.5) {};
		\node [style=Z dot] (10) at (-4, 2.5) {};
		\node [style=Z dot] (11) at (-5, 2.5) {};
		\node [style=none] (12) at (-3, 1.5) {};
		\node [style=none] (13) at (-2.5, 1.5) {};
		\node [style=X dot] (14) at (-3, 1.5) {};
		\node [style=X dot] (15) at (-2.5, 1.5) {};
		\node [style=none] (16) at (-4, -1.5) {};
		\node [style=none] (17) at (-1.5, -1.5) {};
		\node [style=none] (18) at (-4, -0.5) {};
		\node [style=none] (19) at (-5, -0.5) {};
		\node [style=none] (20) at (-1.5, -2.5) {};
		\node [style=none] (21) at (-0.5, -2.5) {};
		\node [style=Z dot] (22) at (-1.5, -2.5) {};
		\node [style=Z dot] (23) at (-0.5, -2.5) {};
		\node [style=Z dot] (24) at (-4, -1.5) {};
		\node [style=Z dot] (25) at (-1.5, -1.5) {};
		\node [style=Z dot] (26) at (-4, -0.5) {};
		\node [style=Z dot] (27) at (-5, -0.5) {};
		\node [style=none] (28) at (-3, -1.5) {};
		\node [style=none] (29) at (-2.5, -1.5) {};
		\node [style=X dot] (30) at (-3, -1.5) {};
		\node [style=X dot] (31) at (-2.5, -1.5) {};
		\node [style=none] (32) at (-6.25, 1.5) {};
		\node [style=none] (33) at (-6.25, 0.5) {};
		\node [style=none] (34) at (-6.25, -1.5) {};
		\node [style=none] (35) at (-6.25, -2.5) {};
		\node [style=H box] (36) at (0.5, 2.5) {};
		\node [style=H box] (37) at (0.5, 1.5) {};
		\node [style=H box] (38) at (0.5, -0.5) {};
		\node [style=H box] (39) at (0.5, -1.5) {};
		\node [style=none] (43) at (1.5, 2.5) {};
		\node [style=none] (44) at (4, 2.5) {};
		\node [style=none] (45) at (1.5, 0.5) {};
		\node [style=none] (46) at (0.5, 0.5) {};
		\node [style=none] (47) at (4, 1.5) {};
		\node [style=none] (48) at (5, 1.5) {};
		\node [style=Z dot] (49) at (4, 1.5) {};
		\node [style=Z dot] (50) at (5, 1.5) {};
		\node [style=Z dot] (51) at (1.5, 2.5) {};
		\node [style=Z dot] (52) at (4, 2.5) {};
		\node [style=Z dot] (53) at (1.5, 0.5) {};
		\node [style=Z dot] (54) at (0.5, 0.5) {};
		\node [style=none] (55) at (2.5, 2.5) {};
		\node [style=none] (56) at (3, 2.5) {};
		\node [style=X dot] (57) at (2.5, 2.5) {};
		\node [style=X dot] (58) at (3, 2.5) {};
		\node [style=none] (59) at (1.5, -0.5) {};
		\node [style=none] (60) at (4, -0.5) {};
		\node [style=none] (61) at (1.5, -2.5) {};
		\node [style=none] (62) at (0.5, -2.5) {};
		\node [style=none] (63) at (4, -1.5) {};
		\node [style=none] (64) at (5, -1.5) {};
		\node [style=Z dot] (65) at (4, -1.5) {};
		\node [style=Z dot] (66) at (5, -1.5) {};
		\node [style=Z dot] (67) at (1.5, -0.5) {};
		\node [style=Z dot] (68) at (4, -0.5) {};
		\node [style=Z dot] (69) at (1.5, -2.5) {};
		\node [style=Z dot] (70) at (0.5, -2.5) {};
		\node [style=none] (71) at (2.5, -0.5) {};
		\node [style=none] (72) at (3, -0.5) {};
		\node [style=X dot] (73) at (2.5, -0.5) {};
		\node [style=X dot] (74) at (3, -0.5) {};
		\node [style=H box] (75) at (5, 2.5) {};
		\node [style=H box] (76) at (5, 0.5) {};
		\node [style=H box] (77) at (5, -2.5) {};
		\node [style=H box] (78) at (5, -0.5) {};
		\node [style=none] (81) at (6.5, 3.5) {};
		\node [style=none] (82) at (6.5, 2.5) {};
		\node [style=none] (83) at (6.5, 0.5) {};
		\node [style=none] (84) at (6.5, -0.5) {};
		\node [style=none] (85) at (5.5, 2.5) {};
		\node [style=none] (86) at (5.5, 0.5) {};
		\node [style=none] (87) at (5.5, -0.5) {};
		\node [style=none] (88) at (5.5, -2.5) {};
		\node [style=none] (97) at (6.75, 3.5) {};
		\node [style=none] (98) at (6.75, 2.5) {};
		\node [style=none] (99) at (6.75, 0.5) {};
		\node [style=none] (100) at (6.75, -0.5) {};
		\node [style=none] (101) at (6.75, -3.25) {};
		\node [style=none] (102) at (6.75, 4) {};
		\node [style=none] (103) at (6.825, -3.25) {};
		\node [style=none] (104) at (6.825, 4) {};
		\node [style=none] (105) at (6.625, 0.625) {.};
		\node [style=none] (106) at (6.625, 0.375) {.};
		\node [style=none] (107) at (-5.625, -3.75) {$\infty$};
		\node [style=none] (108) at (-5.55, -3.25) {};
		\node [style=none] (109) at (-5.55, 4) {};
		\node [style=none] (110) at (-5.625, -3.25) {};
		\node [style=none] (111) at (-5.625, 4) {};
		\node [style=none] (112) at (-5.425, 0.625) {.};
		\node [style=none] (113) at (-5.425, 0.375) {.};
	\end{pgfonlayer}
	\begin{pgfonlayer}{edgelayer}
		\draw (11) to (10);
		\draw (10) to (8);
		\draw (15) to (9);
		\draw (8) to (14);
		\draw (14) to (15);
		\draw (9) to (6);
		\draw (6) to (7);
		\draw (27) to (26);
		\draw (26) to (24);
		\draw (31) to (25);
		\draw (24) to (30);
		\draw (30) to (31);
		\draw (25) to (22);
		\draw (22) to (23);
		\draw (14) to (30);
		\draw (31) to (15);
		\draw (32.center) to (8);
		\draw (33.center) to (6);
		\draw (34.center) to (24);
		\draw (35.center) to (22);
		\draw (10) to (36);
		\draw (9) to (37);
		\draw (26) to (38);
		\draw (25) to (39);
		\draw (54) to (53);
		\draw (53) to (51);
		\draw (58) to (52);
		\draw (51) to (57);
		\draw (57) to (58);
		\draw (52) to (49);
		\draw (49) to (50);
		\draw (70) to (69);
		\draw (69) to (67);
		\draw (74) to (68);
		\draw (67) to (73);
		\draw (73) to (74);
		\draw (68) to (65);
		\draw (65) to (66);
		\draw (57) to (73);
		\draw (74) to (58);
		\draw (36) to (51);
		\draw (37) to (49);
		\draw (38) to (67);
		\draw (39) to (65);
		\draw (69) to (77);
		\draw (68) to (78);
		\draw (52) to (75);
		\draw (53) to (76);
		\draw (78)
			 to (87.center)
			 to [in=180, out=0] (84.center)
			 to (100.center);
		\draw (77)
			 to (88.center)
			 to [in=-180, out=0, looseness=0.50] (83.center)
			 to (99.center);
		\draw (75)
			 to (85.center)
			 to [in=180, out=0] (82.center)
			 to (98.center);
		\draw (76)
			 to (86.center)
			 to [in=-180, out=0, looseness=0.50] (81.center)
			 to (97.center);
		\draw (102.center) to (101.center);
		\draw [style=thick] (104.center) to (103.center);
		\draw (109.center) to (108.center);
		\draw [style=thick] (111.center) to (110.center);
	\end{pgfonlayer}
\end{tikzpicture}

 \hspace{-2cm}&\hspace{2cm}
 \overset{(\hyperlink{eq:OCM}{OCM})}{=}
 \begin{tikzpicture}
	\begin{pgfonlayer}{nodelayer}
		\node [style=none] (0) at (-5.5, 1.5) {};
		\node [style=none] (1) at (-3, 1.5) {};
		\node [style=none] (2) at (-5.5, 2.5) {};
		\node [style=none] (3) at (-6.5, 2.5) {};
		\node [style=none] (4) at (-3, 0.5) {};
		\node [style=none] (5) at (-2, 0.5) {};
		\node [style=Z dot] (6) at (-3, 0.5) {};
		\node [style=Z dot] (7) at (-2, 0.5) {};
		\node [style=Z dot] (8) at (-5.5, 1.5) {};
		\node [style=Z dot] (9) at (-3, 1.5) {};
		\node [style=Z dot] (10) at (-5.5, 2.5) {};
		\node [style=Z dot] (11) at (-6.5, 2.5) {};
		\node [style=none] (12) at (-4.5, 1.5) {};
		\node [style=none] (13) at (-4, 1.5) {};
		\node [style=X dot] (14) at (-4.5, 1.5) {};
		\node [style=X dot] (15) at (-4, 1.5) {};
		\node [style=none] (16) at (-5.5, -1.5) {};
		\node [style=none] (17) at (-3, -1.5) {};
		\node [style=none] (18) at (-5.5, -0.5) {};
		\node [style=none] (19) at (-6.5, -0.5) {};
		\node [style=none] (20) at (-3, -2.5) {};
		\node [style=none] (21) at (-2, -2.5) {};
		\node [style=Z dot] (22) at (-3, -2.5) {};
		\node [style=Z dot] (23) at (-2, -2.5) {};
		\node [style=Z dot] (24) at (-5.5, -1.5) {};
		\node [style=Z dot] (25) at (-3, -1.5) {};
		\node [style=Z dot] (26) at (-5.5, -0.5) {};
		\node [style=Z dot] (27) at (-6.5, -0.5) {};
		\node [style=none] (28) at (-4.5, -1.5) {};
		\node [style=none] (29) at (-4, -1.5) {};
		\node [style=X dot] (30) at (-4.5, -1.5) {};
		\node [style=X dot] (31) at (-4, -1.5) {};
		\node [style=none] (32) at (-7.75, 1.5) {};
		\node [style=none] (33) at (-7.75, 0.5) {};
		\node [style=none] (34) at (-7.75, -1.5) {};
		\node [style=none] (35) at (-7.75, -2.5) {};
		\node [style=H box] (36) at (-1, 2.5) {};
		\node [style=H box] (37) at (-1, 1.5) {};
		\node [style=H box] (38) at (-1, -0.5) {};
		\node [style=H box] (39) at (-1, -1.5) {};
		\node [style=none] (43) at (0, 2.5) {};
		\node [style=none] (44) at (2.5, 2.5) {};
		\node [style=none] (45) at (0, 0.5) {};
		\node [style=none] (46) at (-1, 0.5) {};
		\node [style=none] (47) at (2.5, 1.5) {};
		\node [style=Z dot] (48) at (2.5, 1.5) {};
		\node [style=none] (49) at (4, 1.5) {};
		\node [style=Z dot] (50) at (0, 2.5) {};
		\node [style=Z dot] (51) at (2.5, 2.5) {};
		\node [style=Z dot] (52) at (0, 0.5) {};
		\node [style=Z dot] (53) at (-1, 0.5) {};
		\node [style=none] (54) at (1, 2.5) {};
		\node [style=none] (55) at (1.5, 2.5) {};
		\node [style=X dot] (56) at (1, 2.5) {};
		\node [style=X dot] (57) at (1.5, 2.5) {};
		\node [style=none] (58) at (0, -0.5) {};
		\node [style=none] (59) at (2.5, -0.5) {};
		\node [style=none] (60) at (0, -2.5) {};
		\node [style=none] (61) at (-1, -2.5) {};
		\node [style=none] (62) at (2.5, -1.5) {};
		\node [style=Z dot] (63) at (2.5, -1.5) {};
		\node [style=none] (64) at (4, -1.5) {};
		\node [style=Z dot] (65) at (0, -0.5) {};
		\node [style=Z dot] (66) at (2.5, -0.5) {};
		\node [style=Z dot] (67) at (0, -2.5) {};
		\node [style=Z dot] (68) at (-1, -2.5) {};
		\node [style=none] (69) at (1, -0.5) {};
		\node [style=none] (70) at (1.5, -0.5) {};
		\node [style=X dot] (71) at (1, -0.5) {};
		\node [style=X dot] (72) at (1.5, -0.5) {};
		\node [style=H box] (73) at (3.5, 2.5) {};
		\node [style=H box] (74) at (3.5, 0.5) {};
		\node [style=H box] (75) at (3.5, -2.5) {};
		\node [style=H box] (76) at (3.5, -0.5) {};
		\node [style=none] (77) at (5.5, 1.5) {};
		\node [style=none] (78) at (5.5, 0.5) {};
		\node [style=none] (79) at (5.5, -1.5) {};
		\node [style=none] (80) at (5.5, -2.5) {};
		\node [style=none] (81) at (4, 2.5) {};
		\node [style=none] (82) at (4, 0.5) {};
		\node [style=none] (83) at (4, -0.5) {};
		\node [style=none] (84) at (4, -2.5) {};
		\node [style=none] (85) at (5.5, 2.5) {};
		\node [style=none] (86) at (5.5, -0.5) {};
		\node [style=none] (89) at (6.5, 1.5) {};
		\node [style=none] (90) at (6.5, 0.5) {};
		\node [style=none] (91) at (6.5, -1.5) {};
		\node [style=none] (92) at (6.5, -2.5) {};
		\node [style=Z dot] (93) at (6, -0.5) {};
		\node [style=Z dot] (94) at (6, 2.5) {};
		\node [style=none] (95) at (6.5, -3.25) {};
		\node [style=none] (96) at (6.5, 3.25) {};
		\node [style=none] (97) at (6.575, -3.25) {};
		\node [style=none] (98) at (6.575, 3.25) {};
		\node [style=none] (99) at (6.375, 0.125) {.};
		\node [style=none] (100) at (6.375, -0.125) {.};
		\node [style=none] (101) at (-7.125, -3.75) {$\infty$};
		\node [style=none] (102) at (-7.05, -3.25) {};
		\node [style=none] (103) at (-7.05, 3.25) {};
		\node [style=none] (104) at (-7.125, -3.25) {};
		\node [style=none] (105) at (-7.125, 3.25) {};
		\node [style=none] (106) at (-6.925, 0.125) {.};
		\node [style=none] (107) at (-6.925, -0.125) {.};
	\end{pgfonlayer}
	\begin{pgfonlayer}{edgelayer}
		\draw (11) to (10);
		\draw (10) to (8);
		\draw (15) to (9);
		\draw (8) to (14);
		\draw (14) to (15);
		\draw (9) to (6);
		\draw (6) to (7);
		\draw (27) to (26);
		\draw (26) to (24);
		\draw (31) to (25);
		\draw (24) to (30);
		\draw (30) to (31);
		\draw (25) to (22);
		\draw (22) to (23);
		\draw (14) to (30);
		\draw (31) to (15);
		\draw (32.center) to (8);
		\draw (33.center) to (6);
		\draw (34.center) to (24);
		\draw (35.center) to (22);
		\draw (10) to (36);
		\draw (9) to (37);
		\draw (26) to (38);
		\draw (25) to (39);
		\draw (53) to (52);
		\draw (52) to (50);
		\draw (57) to (51);
		\draw (50) to (56);
		\draw (56) to (57);
		\draw (51) to (48);
		\draw (48) to (49.center);
		\draw (68) to (67);
		\draw (67) to (65);
		\draw (72) to (66);
		\draw (65) to (71);
		\draw (71) to (72);
		\draw (66) to (63);
		\draw (63) to (64.center);
		\draw (56) to (71);
		\draw (72) to (57);
		\draw (36) to (50);
		\draw (37) to (48);
		\draw (38) to (65);
		\draw (39) to (63);
		\draw (67) to (75);
		\draw (66) to (76);
		\draw (51) to (73);
		\draw (52) to (74);
		\draw [in=180, out=0] (83.center) to (80.center);
		\draw [in=-180, out=0, looseness=1.25] (84.center) to (79.center);
		\draw [in=180, out=0] (81.center) to (78.center);
		\draw [in=-180, out=0, looseness=1.25] (82.center) to (77.center);
		\draw (73) to (81.center);
		\draw (74) to (82.center);
		\draw (76) to (83.center);
		\draw (75) to (84.center);
		\draw [in=-180, out=0] (64.center) to (86.center);
		\draw [in=-180, out=0] (49.center) to (85.center);
		\draw (77.center) to (89.center);
		\draw (78.center) to (90.center);
		\draw (79.center) to (91.center);
		\draw (80.center) to (92.center);
		\draw (85.center) to (94);
		\draw (86.center) to (93);
		\draw (96.center) to (95.center);
		\draw [style=thick] (98.center) to (97.center);
		\draw (103.center) to (102.center);
		\draw [style=thick] (105.center) to (104.center);
	\end{pgfonlayer}
\end{tikzpicture}
\\
  &\overset{\eqref{infty-rewrite-reorder}}{=} \ \ \begin{tikzpicture}
	\begin{pgfonlayer}{nodelayer}
		\node [style=none] (0) at (-3.5, 1.5) {};
		\node [style=none] (1) at (-1, 1.5) {};
		\node [style=none] (2) at (-3.5, 2.5) {};
		\node [style=none] (3) at (-6, 2.5) {};
		\node [style=none] (4) at (-1, 0.5) {};
		\node [style=none] (5) at (0, 0.5) {};
		\node [style=Z dot] (6) at (-1, 0.5) {};
		\node [style=Z dot] (7) at (0, 0.5) {};
		\node [style=Z dot] (8) at (-3.5, 1.5) {};
		\node [style=Z dot] (9) at (-1, 1.5) {};
		\node [style=Z dot] (10) at (-3.5, 2.5) {};
		\node [style=Z dot] (11) at (-6, 2.5) {};
		\node [style=none] (12) at (-2.5, 1.5) {};
		\node [style=none] (13) at (-2, 1.5) {};
		\node [style=X dot] (14) at (-2.5, 1.5) {};
		\node [style=X dot] (15) at (-2, 1.5) {};
		\node [style=none] (16) at (-3.5, -1.5) {};
		\node [style=none] (17) at (-1, -1.5) {};
		\node [style=none] (18) at (-3.5, -0.5) {};
		\node [style=none] (19) at (-7, -0.5) {};
		\node [style=none] (20) at (-1, -2.5) {};
		\node [style=none] (21) at (0, -2.5) {};
		\node [style=Z dot] (22) at (-1, -2.5) {};
		\node [style=Z dot] (23) at (0, -2.5) {};
		\node [style=Z dot] (24) at (-3.5, -1.5) {};
		\node [style=Z dot] (25) at (-1, -1.5) {};
		\node [style=Z dot] (26) at (-3.5, -0.5) {};
		\node [style=Z dot] (27) at (-6, -0.5) {};
		\node [style=none] (28) at (-2.5, -1.5) {};
		\node [style=none] (29) at (-2, -1.5) {};
		\node [style=X dot] (30) at (-2.5, -1.5) {};
		\node [style=X dot] (31) at (-2, -1.5) {};
		\node [style=none] (32) at (-7, 1.5) {};
		\node [style=none] (33) at (-7, 0.5) {};
		\node [style=none] (34) at (-7, -1.5) {};
		\node [style=none] (35) at (-7, -2.5) {};
		\node [style=H box] (36) at (1, 2.5) {};
		\node [style=H box] (37) at (1, 1.5) {};
		\node [style=H box] (38) at (1, -0.5) {};
		\node [style=H box] (39) at (1, -1.5) {};
		\node [style=none] (43) at (2, 2.5) {};
		\node [style=none] (44) at (4.5, 2.5) {};
		\node [style=none] (45) at (2, 0.5) {};
		\node [style=none] (46) at (1, 0.5) {};
		\node [style=none] (47) at (4.5, 1.5) {};
		\node [style=Z dot] (48) at (4.5, 1.5) {};
		\node [style=none] (49) at (6.5, 1.5) {};
		\node [style=Z dot] (50) at (2, 2.5) {};
		\node [style=Z dot] (51) at (4.5, 2.5) {};
		\node [style=Z dot] (52) at (2, 0.5) {};
		\node [style=Z dot] (53) at (1, 0.5) {};
		\node [style=none] (54) at (3, 2.5) {};
		\node [style=none] (55) at (3.5, 2.5) {};
		\node [style=X dot] (56) at (3, 2.5) {};
		\node [style=X dot] (57) at (3.5, 2.5) {};
		\node [style=none] (58) at (2, -0.5) {};
		\node [style=none] (59) at (4.5, -0.5) {};
		\node [style=none] (60) at (2, -2.5) {};
		\node [style=none] (61) at (1, -2.5) {};
		\node [style=none] (62) at (4.5, -1.5) {};
		\node [style=Z dot] (63) at (4.5, -1.5) {};
		\node [style=none] (64) at (6.5, -1.5) {};
		\node [style=Z dot] (65) at (2, -0.5) {};
		\node [style=Z dot] (66) at (4.5, -0.5) {};
		\node [style=Z dot] (67) at (2, -2.5) {};
		\node [style=Z dot] (68) at (1, -2.5) {};
		\node [style=none] (69) at (3, -0.5) {};
		\node [style=none] (70) at (3.5, -0.5) {};
		\node [style=X dot] (71) at (3, -0.5) {};
		\node [style=X dot] (72) at (3.5, -0.5) {};
		\node [style=H box] (73) at (5.5, 2.5) {};
		\node [style=H box] (74) at (5.5, 0.5) {};
		\node [style=H box] (75) at (5.5, -2.5) {};
		\node [style=H box] (76) at (5.5, -0.5) {};
		\node [style=none] (79) at (8, 1.5) {};
		\node [style=none] (80) at (8, 0.5) {};
		\node [style=none] (81) at (8, -1.5) {};
		\node [style=none] (82) at (8, -2.5) {};
		\node [style=none] (83) at (6.5, 2.5) {};
		\node [style=none] (84) at (6.5, 0.5) {};
		\node [style=none] (85) at (6.5, -0.5) {};
		\node [style=none] (86) at (6.5, -2.5) {};
		\node [style=none] (87) at (8, 2.5) {};
		\node [style=none] (88) at (8, -0.5) {};
		\node [style=none] (89) at (10.5, 1.5) {};
		\node [style=none] (90) at (10.5, 0.5) {};
		\node [style=none] (91) at (10.5, -1.5) {};
		\node [style=none] (92) at (10.5, -2.5) {};
		\node [style=Z dot] (93) at (8.5, 2.5) {};
		\node [style=Z dot] (94) at (8.5, -0.5) {};
		\node [style=Z dot] (95) at (9.5, -0.5) {};
		\node [style=Z dot] (96) at (9.5, 2.5) {};
		\node [style=none] (97) at (10.5, 2.5) {};
		\node [style=none] (98) at (10.5, -0.5) {};
		\node [style=none] (99) at (10.5, -3.25) {};
		\node [style=none] (100) at (10.5, 3.25) {};
		\node [style=none] (101) at (10.575, -3.25) {};
		\node [style=none] (102) at (10.575, 3.25) {};
		\node [style=none] (103) at (10.375, 0.125) {.};
		\node [style=none] (104) at (10.375, -0.125) {.};
		\node [style=none] (105) at (-4.875, -3.75) {$\infty$};
		\node [style=none] (106) at (-4.8, -3.25) {};
		\node [style=none] (107) at (-4.8, 3.25) {};
		\node [style=none] (108) at (-4.875, -3.25) {};
		\node [style=none] (109) at (-4.875, 3.25) {};
		\node [style=none] (110) at (-4.675, 0.125) {.};
		\node [style=none] (111) at (-4.675, -0.125) {.};
	\end{pgfonlayer}
	\begin{pgfonlayer}{edgelayer}
		\draw (11) to (10);
		\draw (10) to (8);
		\draw (15) to (9);
		\draw (8) to (14);
		\draw (14) to (15);
		\draw (9) to (6);
		\draw (6) to (7);
		\draw (27) to (26);
		\draw (26) to (24);
		\draw (31) to (25);
		\draw (24) to (30);
		\draw (30) to (31);
		\draw (25) to (22);
		\draw (22) to (23);
		\draw (14) to (30);
		\draw (31) to (15);
		\draw (32.center) to (8);
		\draw (33.center) to (6);
		\draw (34.center) to (24);
		\draw (35.center) to (22);
		\draw (10) to (36);
		\draw (9) to (37);
		\draw (26) to (38);
		\draw (25) to (39);
		\draw (53) to (52);
		\draw (52) to (50);
		\draw (57) to (51);
		\draw (50) to (56);
		\draw (56) to (57);
		\draw (51) to (48);
		\draw (48) to (49.center);
		\draw (68) to (67);
		\draw (67) to (65);
		\draw (72) to (66);
		\draw (65) to (71);
		\draw (71) to (72);
		\draw (66) to (63);
		\draw (63) to (64.center);
		\draw (56) to (71);
		\draw (72) to (57);
		\draw (36) to (50);
		\draw (37) to (48);
		\draw (38) to (65);
		\draw (39) to (63);
		\draw (67) to (75);
		\draw (66) to (76);
		\draw (51) to (73);
		\draw (52) to (74);
		\draw [in=180, out=0] (85.center) to (82.center);
		\draw [in=-180, out=0, looseness=1.25] (86.center) to (81.center);
		\draw [in=180, out=0] (83.center) to (80.center);
		\draw [in=-180, out=0, looseness=1.25] (84.center) to (79.center);
		\draw (73) to (83.center);
		\draw (74) to (84.center);
		\draw (76) to (85.center);
		\draw (75) to (86.center);
		\draw [in=-180, out=0] (64.center) to (88.center);
		\draw [in=-180, out=0] (49.center) to (87.center);
		\draw (79.center) to (89.center);
		\draw (80.center) to (90.center);
		\draw (81.center) to (91.center);
		\draw (82.center) to (92.center);
		\draw (87.center) to (93);
		\draw (88.center) to (94);
		\draw (96) to (97.center);
		\draw (95) to (98.center);
		\draw (100.center) to (99.center);
		\draw [style=thick] (102.center) to (101.center);
		\draw (107.center) to (106.center);
		\draw [style=thick] (109.center) to (108.center);
	\end{pgfonlayer}
\end{tikzpicture}
 \\
  &\overset{\eqref{one-pauli}}{=} \ \ \begin{tikzpicture}
	\begin{pgfonlayer}{nodelayer}
		\node [style=none] (0) at (-5.5, 1.5) {};
		\node [style=none] (1) at (-3, 1.5) {};
		\node [style=none] (2) at (-5.5, 2.5) {};
		\node [style=none] (3) at (-8, 2.5) {};
		\node [style=none] (4) at (-3, 0.5) {};
		\node [style=X dot] (5) at (-2, 0.5) {};
		\node [style=Z dot] (6) at (-3, 0.5) {};
		\node [style=X dot] (7) at (-2, 0.5) {};
		\node [style=Z dot] (8) at (-5.5, 1.5) {};
		\node [style=Z dot] (9) at (-3, 1.5) {};
		\node [style=Z dot] (10) at (-5.5, 2.5) {};
		\node [style=Z dot] (11) at (-8, 2.5) {};
		\node [style=none] (12) at (-4.5, 1.5) {};
		\node [style=none] (13) at (-4, 1.5) {};
		\node [style=X dot] (14) at (-4.5, 1.5) {};
		\node [style=X dot] (15) at (-4, 1.5) {};
		\node [style=none] (16) at (-5.5, -1.5) {};
		\node [style=none] (17) at (-3, -1.5) {};
		\node [style=none] (18) at (-5.5, -0.5) {};
		\node [style=none] (19) at (-9, -0.5) {};
		\node [style=none] (20) at (-3, -2.5) {};
		\node [style=X dot] (21) at (-2, -2.5) {};
		\node [style=Z dot] (22) at (-3, -2.5) {};
		\node [style=X dot] (23) at (-2, -2.5) {};
		\node [style=Z dot] (24) at (-5.5, -1.5) {};
		\node [style=Z dot] (25) at (-3, -1.5) {};
		\node [style=Z dot] (26) at (-5.5, -0.5) {};
		\node [style=Z dot] (27) at (-8, -0.5) {};
		\node [style=none] (28) at (-4.5, -1.5) {};
		\node [style=none] (29) at (-4, -1.5) {};
		\node [style=X dot] (30) at (-4.5, -1.5) {};
		\node [style=X dot] (31) at (-4, -1.5) {};
		\node [style=none] (32) at (-9, 1.5) {};
		\node [style=none] (33) at (-9, 0.5) {};
		\node [style=none] (34) at (-9, -1.5) {};
		\node [style=none] (35) at (-9, -2.5) {};
		\node [style=H box] (36) at (-1, 2.5) {};
		\node [style=H box] (37) at (-1, 1.5) {};
		\node [style=H box] (38) at (-1, -0.5) {};
		\node [style=H box] (39) at (-1, -1.5) {};
		\node [style=none] (43) at (0, 2.5) {};
		\node [style=none] (44) at (2.5, 2.5) {};
		\node [style=none] (45) at (0, 0.5) {};
		\node [style=X dot] (46) at (-1, 0.5) {};
		\node [style=none] (47) at (2.5, 1.5) {};
		\node [style=Z dot] (48) at (2.5, 1.5) {};
		\node [style=none] (49) at (4.5, 1.5) {};
		\node [style=Z dot] (50) at (0, 2.5) {};
		\node [style=Z dot] (51) at (2.5, 2.5) {};
		\node [style=Z dot] (52) at (0, 0.5) {};
		\node [style=X dot] (53) at (-1, 0.5) {};
		\node [style=none] (54) at (1, 2.5) {};
		\node [style=none] (55) at (1.5, 2.5) {};
		\node [style=X dot] (56) at (1, 2.5) {};
		\node [style=X dot] (57) at (1.5, 2.5) {};
		\node [style=none] (58) at (0, -0.5) {};
		\node [style=none] (59) at (2.5, -0.5) {};
		\node [style=none] (60) at (0, -2.5) {};
		\node [style=X dot] (61) at (-1, -2.5) {};
		\node [style=none] (62) at (2.5, -1.5) {};
		\node [style=Z dot] (63) at (2.5, -1.5) {};
		\node [style=none] (64) at (4.5, -1.5) {};
		\node [style=Z dot] (65) at (0, -0.5) {};
		\node [style=Z dot] (66) at (2.5, -0.5) {};
		\node [style=Z dot] (67) at (0, -2.5) {};
		\node [style=X dot] (68) at (-1, -2.5) {};
		\node [style=none] (69) at (1, -0.5) {};
		\node [style=none] (70) at (1.5, -0.5) {};
		\node [style=X dot] (71) at (1, -0.5) {};
		\node [style=X dot] (72) at (1.5, -0.5) {};
		\node [style=H box] (73) at (3.5, 2.5) {};
		\node [style=H box] (74) at (3.5, 0.5) {};
		\node [style=H box] (75) at (3.5, -2.5) {};
		\node [style=H box] (76) at (3.5, -0.5) {};
		\node [style=none] (79) at (6, 1.5) {};
		\node [style=none] (80) at (6, 0.5) {};
		\node [style=none] (81) at (6, -1.5) {};
		\node [style=none] (82) at (6, -2.5) {};
		\node [style=none] (83) at (4.5, 2.5) {};
		\node [style=none] (84) at (4.5, 0.5) {};
		\node [style=none] (85) at (4.5, -0.5) {};
		\node [style=none] (86) at (4.5, -2.5) {};
		\node [style=none] (87) at (6, 2.5) {};
		\node [style=none] (88) at (6, -0.5) {};
		\node [style=none] (89) at (8.5, 1.5) {};
		\node [style=none] (90) at (8.5, 0.5) {};
		\node [style=none] (91) at (8.5, -1.5) {};
		\node [style=none] (92) at (8.5, -2.5) {};
		\node [style=X dot] (93) at (6.5, 2.5) {};
		\node [style=X dot] (94) at (6.5, -0.5) {};
		\node [style=X dot] (95) at (7.5, -0.5) {};
		\node [style=X dot] (96) at (7.5, 2.5) {};
		\node [style=none] (97) at (8.5, 2.5) {};
		\node [style=none] (98) at (8.5, -0.5) {};
		\node [style=Z dot] (99) at (-2, 1) {};
		\node [style=Z dot] (100) at (-1, 1) {};
		\node [style=Z dot] (101) at (-2, -2) {};
		\node [style=Z dot] (102) at (-1, -2) {};
		\node [style=Z dot] (103) at (6.5, 0) {};
		\node [style=Z dot] (104) at (7.5, 0) {};
		\node [style=Z dot] (105) at (6.5, 3) {};
		\node [style=Z dot] (106) at (7.5, 3) {};
		\node [style=none] (107) at (8.5, -3.25) {};
		\node [style=none] (108) at (8.5, 3.25) {};
		\node [style=none] (109) at (8.575, -3.25) {};
		\node [style=none] (110) at (8.575, 3.25) {};
		\node [style=none] (111) at (8.375, 0.125) {.};
		\node [style=none] (112) at (8.375, -0.125) {.};
		\node [style=none] (113) at (-6.875, -3.75) {$\infty$};
		\node [style=none] (114) at (-6.8, -3.25) {};
		\node [style=none] (115) at (-6.8, 3.25) {};
		\node [style=none] (116) at (-6.875, -3.25) {};
		\node [style=none] (117) at (-6.875, 3.25) {};
		\node [style=none] (118) at (-6.675, 0.125) {.};
		\node [style=none] (119) at (-6.675, -0.125) {.};
	\end{pgfonlayer}
	\begin{pgfonlayer}{edgelayer}
		\draw (11) to (10);
		\draw (10) to (8);
		\draw (15) to (9);
		\draw (8) to (14);
		\draw (14) to (15);
		\draw (9) to (6);
		\draw (6) to (7);
		\draw (27) to (26);
		\draw (26) to (24);
		\draw (31) to (25);
		\draw (24) to (30);
		\draw (30) to (31);
		\draw (25) to (22);
		\draw (22) to (23);
		\draw (14) to (30);
		\draw (31) to (15);
		\draw (32.center) to (8);
		\draw (33.center) to (6);
		\draw (34.center) to (24);
		\draw (35.center) to (22);
		\draw (10) to (36);
		\draw (9) to (37);
		\draw (26) to (38);
		\draw (25) to (39);
		\draw (53) to (52);
		\draw (52) to (50);
		\draw (57) to (51);
		\draw (50) to (56);
		\draw (56) to (57);
		\draw (51) to (48);
		\draw (48) to (49.center);
		\draw (68) to (67);
		\draw (67) to (65);
		\draw (72) to (66);
		\draw (65) to (71);
		\draw (71) to (72);
		\draw (66) to (63);
		\draw (63) to (64.center);
		\draw (56) to (71);
		\draw (72) to (57);
		\draw (36) to (50);
		\draw (37) to (48);
		\draw (38) to (65);
		\draw (39) to (63);
		\draw (67) to (75);
		\draw (66) to (76);
		\draw (51) to (73);
		\draw (52) to (74);
		\draw [in=180, out=0] (85.center) to (82.center);
		\draw [in=-180, out=0, looseness=1.25] (86.center) to (81.center);
		\draw [in=180, out=0] (83.center) to (80.center);
		\draw [in=-180, out=0, looseness=1.25] (84.center) to (79.center);
		\draw (73) to (83.center);
		\draw (74) to (84.center);
		\draw (76) to (85.center);
		\draw (75) to (86.center);
		\draw [in=-180, out=0] (64.center) to (88.center);
		\draw [in=-180, out=0] (49.center) to (87.center);
		\draw (79.center) to (89.center);
		\draw (80.center) to (90.center);
		\draw (81.center) to (91.center);
		\draw (82.center) to (92.center);
		\draw (87.center) to (93);
		\draw (88.center) to (94);
		\draw (96) to (97.center);
		\draw (95) to (98.center);
		\draw (93) to (96);
		\draw (94) to (95);
		\draw (23) to (68);
		\draw (7) to (53);
		\draw (106) to (96);
		\draw (105) to (93);
		\draw (103) to (94);
		\draw (104) to (95);
		\draw (101) to (23);
		\draw (102) to (68);
		\draw (99) to (7);
		\draw (100) to (53);
		\draw (108.center) to (107.center);
		\draw [style=thick] (110.center) to (109.center);
		\draw (115.center) to (114.center);
		\draw [style=thick] (117.center) to (116.center);
	\end{pgfonlayer}
\end{tikzpicture}

\end{align*}
After removing the mid-circuit measurements, the only preparations left are at the beginning.

The resulting circuit has swaps. 
To get rid of these, we observe that repeating the swaps three times gives the identity. Therefore, repeating the entire circuit three times lets us cancel out the swaps at the cost of having a slightly longer circuit. 
We get: 
\[\input{floquetification-2/422-final.tikz}\]
\noindent For the exact details, see~\autoref{prop:removing-swaps-422}.

\subsubsection{Step 3 --- Extracting the new code}
The goal of this Floquetification procedure is to derive a new Floquet code from an existing stabiliser code while preserving desirable properties such as distance and number of encoded qubits. We observe that, ignoring the two preparations in the beginning, the remaining measurement circuit of the $\code{4, 2, 2}$ is a circuit consisting of Pauli measurements and single-qubit Clifford gates on six qubits:
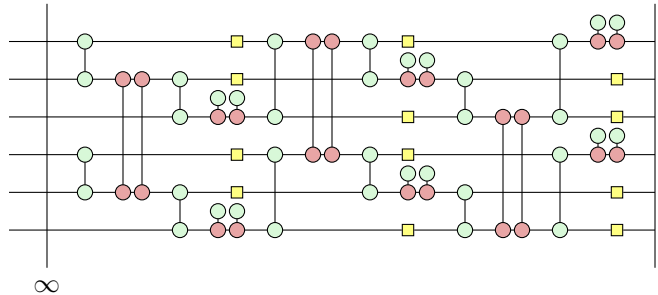
\[\input{floquetification-2/622.tikz}\]

We define: 
\begin{definition}[Dynamical circuit]
 We say a ZX diagram is a \textit{dynamical circuit} on $n$ qubits if it consists of an infinitely repeating operation schedule $\mathcal O = [O_1, O_2, \dots, O_k]$ of Pauli measurements and single-qubit Clifford unitaries. 
\end{definition}

A dynamical circuit is an infinitely repeating circuit that only consists of operations allowed in the MCLC form.
By our decomposition procedure of the measurements, we are guaranteed to always have a dynamical circuit at this step of the algorithm. 

We can now observe: 
\begin{proposition}
  \label{prop:dynamical-circuit-to-floquet-circuit}
 For any dynamical circuit on $n$ qubits, we can create an analogous Floquet code on $n$ qubits of the same ZX distance. 
\end{proposition}
\begin{proof}
 To turn a dynamical circuit into a Floquet code, we will take the single-qubit unitaries and push them towards the end of the circuit until, eventually, they all cancel out. 
 For this, we first observe that the following rewrite for any single-qubit Clifford unitary $C$ is distance-preserving:
  \[\input{Clifford-intro.tikz}\]
 This is clear from the fact that any Pauli error between the two Clifford unitaries will push out to become a different Pauli error. 

 Then, we can distance-preservingly push all the Clifford unitaries towards the end of the operation schedule as follows:
  \[\input{moving-cliffords.tikz}\]
 where the red spiders can be connected to zero or more other green spiders representing a corresponding high-weight measurement.
 This rewrite can be interpreted as updating the measurement: instead of doing the local Clifford before the measurement, we perform a conjugated measurement and do the Clifford afterwards.

Using this process, we can push all the single-qubit unitaries to the end of the operation schedule.
 We are now left with only Pauli measurements, except for some Clifford unitaries at the end.
 If all the Clifford unitaries cancel out, then we are done.
 If there are some left, we can use $\eqref{infty-rewrite-unroll}$ to consider a second iteration of the operation schedule. 
 In this second iteration, we can, once again, push through all the single-qubit Clifford unitaries. 
 By the end, we have obtained the same unitaries as we got the first time. 
 As we also push through the leftovers of the first round, we now have each Clifford unitary twice.
 This will cancel out all Hadamards, but might leave some gates, such as $S^2 = Z$.
 We can repeat this process until all Clifford unitaries are removed (at most 12 times) and we are left with an operation circuit that only consists of Pauli measurements and, therefore, describes a Floquet code.
\end{proof}

Observing that the circuit above is a dynamical circuit, we can apply \autoref{prop:dynamical-circuit-to-floquet-circuit} to get:

\[\input{floquetification-2/622-final.tikz}\]
where we read the red two-qubit measurements as $XX$ measurements, the green ones as $ZZ$ and the green ones where some edges are conjugated by Hadamards as $ZX$ or $XZ$ measurements.
We have, therefore, obtained a circuit solely consisting of infinitely repeating Pauli measurements, which, by definition, is a Floquet code. 

\noindent While we are now guaranteed to get a new Floquet code, derived from the original stabiliser code, it remains to be shown that this new code has the same number of logicals and the same distance as the original stabiliser code.

\subsection{Characterising the Floquetification procedure}
In general, this procedure can take any $\code{n, k, d}$ stabiliser code and transform it into an $\code{n', k', d'}$ Floquet code that uses only one-qubit and two-qubit Pauli measurements. 
In the remainder of this section, we will characterise characterise $n', k'$ and $d'$ and then present a slightly adapted procedure that guarantees that the resulting Floquet code is always a proper Floquet code.

\subsubsection{Number of physical qubits}
The procedure requires additional qubits. 
To calculate the number of physical qubits required by the new code, we observe that the individual Pauli measurements free up the auxiliary qubits they require in their implementation immediately after being completed.
Therefore, these auxiliary qubits can be reused for the next measurement, and thus the number of additional qubits is solely dependent on the largest weight measurement in the measurement schedule of the original code. 
We have previously shown that for a measurement of weight $w$, we require $\lceil \frac{w}{2} \rceil + f(w)$ auxiliary qubits where $f(w) \leq \log_2(w)$.
Therefore, the number of additional qubits required for this procedure is $\lceil \frac{w}{2} \rceil + f(w)$ where $w$ is the weight of the highest-weight Pauli measurement.
For code families where $w$ is constant, the overhead of auxiliary qubits is also constant. 
Overall, we have $n' = n + \lceil \frac{w}{2} \rceil + f(w)$.

\subsubsection{Number of logical qubits}
\label{sec:logical-qubits}
The number of logicals the circuit encodes remains the same. 
We will argue this by considering the three steps of the Floquetification procedure individually: writing out the measurement schedule, rewriting the measurements and extracting the new code. 

The first step expressed the original code as a ZX diagram, which does not affect the number of logicals. The second step only performs ZX rewrites on the measurement circuit.
This means that the circuit at the end of step two is equivalent to the original circuit. But then, after any number of iterations of the measurement schedule, the stabilisers, co-stabilisers and logicals of the ZX diagram must be the same, and thus the two circuits encode the same amount of information.

Finally, we argue that the last step of removing the single-qubit preparation does not change the number of logicals either.
We observe: 
\begin{proposition}
 These two circuits have the same number of logical qubits: 
    \[\input{floquetification/logicals-setup.tikz}\]
    \begin{proof}
 Both diagrams have the same stabilisers: 
      \[\input{floquetification/stabilisers.tikz}\]

      \noindent Additionally, both diagrams have the same number of logicals with representatives that highlight the same output edges.
        \[\input{floquetification/logicals.tikz}\]

    \end{proof}
\end{proposition} 

This Floquetification procedure guarantees that all preparations at the beginning of the circuit that we might want to remove will always be followed by a measurement in this form.
This is because they are either created to achieve measurement-circuit flow on the original measurement, in which case both the preparation and the measurement are green, or during the expansion of the high-weight measurement spider, in which case both are red.
As ZX is colour-symmetric, either way, removing the preparations one by one does not change the number of logicals.
Therefore, we have $k' = k$. 

\subsubsection{Code distance}
Given a $\interp{n, k, d}$ stabiliser code, the outlined algorithm starts with the measurement circuit of that code. 
By \autoref{prop:distance-correspondence}, we know that after establishment, the corresponding ZX diagram has a ZX distance of~$d$. 
Then the procedure applies distance-preserving rewrites to manipulate the measurement circuit into a new shape. 
As we use distance-preserving rewrites, by the end of the second step of the Floquetification procedure, the distance of the rewritten ZX diagram after establishment is still $d$\footnote{
 In practice, the Floquet code might establish slightly earlier than the original stabiliser code. However, in our Floquetification procedure, this difference can be at most one step of local Cliffords, which we can ignore when calculating the code distance. }. 

Next, the algorithm removes the preparations. 
As this changes the underlying linear map, this is not a distance-preserving rewrite. 
Thus, we have to show that this step, nonetheless, preserves the distance of the ZX diagram after establishment. 
This rewrite of removing the preparations happens before the diagram establishes. 
As the distance of the diagram after establihsment only cares about errors after establishment, this means it does not introduce new errors.
The other way that this step could affect the distance after establishment would be by affecting when the diagram establishes or ISGs after establishment. 
However, as we have shown above, this is not the case. 
Therefore, the resulting measurement circuit of a Floquet code has the same distance after establishment as the original stabiliser code. 
But then, by \autoref{prop:distance-correspondence}, we have $d' = d$.

\subsubsection{Proper Floquet code}
\label{sec:proper-floquet}
Our definition of Floquet codes generalises both stabiliser and subsystem codes.
We distinguish these as special cases: for stabiliser codes, the instantaneous stabilising group never changes;
for subsystem codes, the associated subsystem code should have the same number of logicals.
Clearly, the measurement schedule resulting from our Floquetification procedure contains anti-commuting measurements and therefore cannot be understood as a stabiliser code. 
We will now show that for any stabiliser code, there exists a Floquetification such that the associated subsystem code of the Floquetified code does not have any logicals.
For this, we make use of the fact that our Floquetification procedure does not specify the exact measurement schedule of the stabilising code, we start out with simply requiring the schedule to generate all stabilisers of the original code. 
In general, it is inconvenient to argue about all possible choices of schedules and their composition. 
Instead, we argue that for any stabilising code, there exists an associated measurement schedule such that the associated subsystem code of the resulting Floquet code has no logicals.
As such, if one wants to guarantee a proper Floquet code as the result of this procedure, one should start with a measurement schedule of this shape. 

\begin{proposition}
  Let $S_1, S_2, \dots S_m$ be a generating set for the stabilisers of some $\interp{n, k, d}$ stabiliser code with $d > 1$.
  Let $\mathcal M$ be the Floquet code obtained by Floquetifying the measurement schedule $[S_1, S_1, S_2, S_2, \dots S_m, S_m]$, i.e.\@ obtained by repeating each stabiliser measurement twice.
  The associated subsystem code $G_{\mathcal M}$ has no logical qubits.
\end{proposition} 
\begin{proof}
   Consider the circuit fragment we get when Floquetifying two consecutive $ZZZZ$ measurements: 
  \[\input{proper-floquet-proof.tikz}\]
  Clearly, every qubit is, at some point in time, measured by a single-qubit $X$ measurement. 
  More generally, when repeating two consecutive Floquetifications of arbitrary-weight $Z$ measurements, every qubit is measured by a single-qubit Pauli $X$. 
  But then, as we assumed the stabiliser code we started with to have distance more than 1, every qubit must be measured by a $Z$ or $Y$ measurement at least once, otherwise an $X$ flip of that qubit would be undetectable. 
  Therefore, in the Floquetified code, every qubit must, at some point in time, be measured by at least one single-qubit $X$ or a single-qubit $Y$ measurement. 
  But then, the subsystem code associated with this Floquetified code cannot have any non-trivial logical representatives that act as $Z$ on any of the qubits, as those would anticommute with the corresponding single-qubit measurement. 
  By a similar line of reasoning, the associated subsystem code can also not have any non-trivial logical representatives that act as $X$ on any of the qubits.
  Therefore, the associated subsystem code cannot have any logicals.
\end{proof}

As we have previously shown that the Floquetification procedure preserves the number of logical qubits, we know that the Floquet codes resulting from the procedure above are proper Floquet codes; while the Floquet code has $k$ logical qubits, the associated subsystem code has none.

\section{Conclusion}
In this work, we defined a notion of distance on ZX diagrams as the weight of the smallest non-trivial, non-detectable error.
This corresponds to the usual code distance for the measurement circuits of stabiliser codes.
Based on this definition, we introduced distance-preserving rewrites.
These are the subset of all ZX rewrites that provably preserve the ZX distance of any diagram they are applied to.

Next, we introduced the notion of measurement-circuit (MC) flow on ZX diagrams.
We showed that any diagram with MC flow can be synthesised into an equivalent circuit-like diagram.
In particular, we identified a \emph{strict} variant of MC flow that guarantees extraction into a circuit consisting solely of local Cliffords and Pauli measurements (MCLC form).
Moreover, this procedure provably preserves the ZX distance.

Using MC flow, we decomposed arbitrary-weight Pauli measurements into an equivalent quantum circuit consisting of weight-one and weight-two operations.
Since we only use distance-preserving rewrites, we are guaranteed that any error in the decomposed quantum circuit corresponds to a data error of at most the same weight.
These decompositions enable us to generalise the Floquetification procedure of~\textcite{townsend-teagueFloquetifyingColourCode2023} to arbitrary stabiliser codes,
provably preserving the distance and number of logicals of the original code.
Crucially, we showed that the resulting codes are `proper' Floquet codes, distinct from subsystem codes, as they possess a dynamic logical structure and non-commuting measurements.
The qubit overhead scales linearly with the weight of the largest measurement and is, therefore, constant for LDPC code families.

Several aspects of our work can be extended as future work.
For example, one may optimise the Floquetification procedures with respect to various criteria such as qubit and gate count, qubit connectivity and the length of the measurement schedule~\parencite{fuenteDynamicalWeight2024}, or utilise alternative gate sets~\parencite{hilaireEnhancedFaulttolerance2024}.
Additionally, one may consider how our approach could be combined with the detector error model~\parencite{derksDesigningFaulttolerant2024}.
This may lead to methods that leverage efficient decoders of the original stabiliser code to design decoders for the Floquetified code.

More generally, this paper lays a mathematical foundation for reasoning about quantum circuits in a noisy setting.
It invites an alternative perspective on Floquetification, viewing it as a procedure that compiles circuits for memory experiments with stabiliser codes into forms involving only measurements of weight at most two. 
The fact that the resulting circuit can be interpreted as a Floquet code is useful for analytical reasons, yet ultimately incidental. 
Building on this intuition, we believe that distance-preserving rewrites can be extended to other areas in fault-tolerant computation, such as more general circuit synthesis and circuit verification tasks.
Building this bridge requires formalising other noise models in the ZX calculus and extending the distance-preserving rewrite system.
We believe that this can lead to novel techniques in fault-tolerant circuit synthesis, optimisation and verification.

\section*{Acknowledgements}
We would like to thank Alex Townsend-Teague for the guidance he provided during numerous discussions, for pointing out an easier proof of \autoref{thm:detecting-region}, and for his detailed feedback on the paper.
We thank Razin Shaikh for his helpful feedback on our draft and Julio Magdalena de la Fuente, Peter-Jan Derks, and Clemens Schumann for the insightful discussions.
We thank Linnea Grans-Samuelsson and Andrey Boris Khesin for their input on the proof that fault equivalence is the same as distance preservation. 
BR thanks Simon Harrison for his generous support for the Wolfson Harrison UK Research Council Quantum Foundation Scholarship.
BP and AK are supported by the Engineering and Physical Sciences Research Council grant number EP/Z002230/1, \enquote{(De)constructing quantum software (DeQS)}.

\section*{Author contribution statement}
All three authors contributed to proving the main results.
B.P. wrote the first draft of \autoref{sec:impl-interp} and B.R. wrote the first draft of the remaining sections.
All authors contributed towards the final version.

\printbibliography

\newpage
\appendix
\section{Pauli webs --- Deferred proofs}
\label{sec:pauli-webs-proofs}

In this section, we prove some unproven properties of Pauli webs.
Note that we use unsigned Pauli webs, so the statements only hold up to a phase of $\pm 1$.

\subsection{Stabilisers}
\label{subsec:stabilisers}

Pauli webs have a formal correspondence with stabilisers and logicals as tracked by group theory. In this section, we formalise that correspondence. 

We define: 
\begin{definition}[Boundary effect]
 Given a Pauli web $w$, $out(w)$ denotes the Pauli string given by the highlighted output edges of $w$.
 We define $in(w)$ similarly for the highlighted input edges.
\end{definition}

We observe that we can multiply Pauli webs to create new Pauli webs.
\begin{definition}[Product of Pauli webs]
  Let $w_1, w_2$ be two Pauli webs on some diagram $D$.
  Then, we define:
  \[
    w_1 \Delta w_2 \coloneqq w_1 \cup w_2 \setminus (w_1 \cap w_2)
  \]
  i.e.\@ the set of edge and colour pairs that are in exactly one of the webs.
\end{definition}
\noindent The product of two Pauli webs is a valid Pauli web~\parencite{bombinUnifyingFlavorsFault2024}.
Furthermore, we define:
\begin{definition}[Quotient equivalence classes]
  Let $G$ and $N$ be two sets of Pauli webs on some diagram $D$ such that $\forall g \in G, n \in N$ such that $g \Delta n \in G$ then the \textit{quotient equivalence classes} of $G / N$ is defined as:
  \[
    G / N \coloneqq \{g \Delta N : g \in G\}
  \]
  where $g \Delta N \coloneqq \{g \Delta n: n \in N\}$.
\end{definition}

Using this, we can state:

\begin{restatable}{proposition}{webOne}
  \label{prop:webBijectiveStabs}
  Let $C_t$ be the measurement circuit of the first $t$ time steps of a Floquet code with detecting Pauli webs $D$ and stabilising Pauli webs $S$.
  Then the equivalence classes $S / D$ are in bijective correspondence with the instantaneous stabiliser group $S_t$ observed by the map $b: S / D \to S_t \dblcolon E \in S / D \mapsto out(e) \text{ for some } e \in E$.
\end{restatable}
\begin{proof}
  First, we show that the map $b$ is well-defined.
  Let $w_i, w_j$ be stabilising webs in the same equivalence class $E \in S / D$.
  Then, by~\parencite{borghansZXcalculusQuantumStabilizer2019}, we know that $out(w_i), out(w_j) \in S_t$.
  It remains to be shown that $out(w_i) = out(w_j)$.
  As $w_i, w_j \in E$, we know there exists a detecting region $d$ such that $w_i \otimes d = w_j$.
  But then $out(w_j) = out(w_i \otimes d) = out(w_i) \otimes out(d)$.
  As $out(d) = I$, we have $out(w_i) = out(w_j)$.
  Therefore, the map is well-defined.

  Next, we show that the map is injective.
  Let $w_i \in E_1, w_j \in E_2$ be stabilising webs such that $out(w_i) = out(w_j)$.
  Then $w_i \otimes w_j$ is a Pauli web which highlights no inputs (as neither $w_i$ nor $w_j$ does) and no outputs (as $out(w_i) = out(w_j)$).
  But then $d = w_i \otimes w_j$ is a detecting region such that $w_i \otimes d = w_j$.
  Therefore, $E_1 = E_2$ and thus $b$ is injective.

  But by~\parencite{borghansZXcalculusQuantumStabilizer2019}, we know that for each stabiliser in $s \in S_t$ there must exist at least one stabilising web $w$ with $out(w) = s$.
  Thus, $b$ is also surjective and therefore bijective.
\end{proof}

\subsection{Logical operators}
\label{subsec:logical-operators}

Similarly, we can formalise the relationship between logicals and Pauli webs. 
For this, we need two auxiliary propositions.
Firstly, we have:
\begin{proposition}
  \label{prop:circuit-decomposition}
  Let $C$ be a circuit consisting of Pauli measurements and unitaries which stabilises the Pauli strings $S$.
  Then $C$ can be written as $C = \Pi \circ U$ with a projector $\Pi$ that stabilises exactly $S$ and some unitary $U$.
\end{proposition}
\begin{proof}
  Let $C$ be a circuit with $n$ operations.
  Then we can iteratively build up the circuit and bring it back into the desired form after each new operation.
  We refer to $C_i$ for $i \in [0, n]$ as the circuit one gets after applying the first $i$ operations of $C$.

  $C_0$ is the identity, so $C_0 = \Pi_0 \circ U_0$ for $\Pi_0 = U_0 = I$.
  Then for any following step, let us assume that $C_{i - 1}$ can be written as $C_{i - 1} = \Pi_{i - 1} \circ U_{i - 1}$.
  Then, if the $i$-th operation $O_i$ is a projector, we can write:
  \[C_{i} = O_i \circ C_{i - 1} = O_i \circ \Pi_{i - 1} \circ U_{i - 1}\]
  But by~\parencite{gottesmanHeisenbergRepresentation1998,delfosseSpacetimeCodes2023}, we know that $O_i \circ \Pi_{i - 1}$ is a new projector $\Pi_i$, so $C_{i} = \Pi_{i} \circ U_{i}$ with $U_i = U_{i - 1}$.

  If $O_i$ is a unitary, then we have:
  \[
    C_{i} = O_i \circ C_{i - 1} = O_i \circ \Pi_{i - 1} \circ U_{i - 1} = O_i \circ \Pi_{i - 1} \circ O_i^\dagger \circ O_i \circ U_{i - 1}
  \]
  But we know that $\Pi_i = O_i \circ \Pi_{i - 1} \circ O_i^\dagger$ is a projector and that $U_i = O_i \circ U_{i - 1}$, the composition of two unitaries, is a unitary.
  Thus, we can write $C_{i} = \Pi_{i} \circ U_{i}$.

  But then, we can always write $C_i$ as $\Pi_i \circ U_i$.
  As $U_i$ has no stabilisers, all the stabilisers of $C_i$ must also be stabilised by $\Pi_i$.
\end{proof}

Secondly, we have: 
\begin{proposition}
  \label{prop:commuting-webs}
  Let $\Pi$ be a ZX diagram for a Pauli projector, then for all Pauli webs $w$ on $\Pi$, $out(w)$ commutes with $\Pi$.
\end{proposition}
\begin{proof}
  We prove this via a counting argument.
  We provide a set of Pauli webs that spans all Pauli webs of $\Pi$ and argue that for each of them, $out(w)$ commutes with $\Pi$.

  Let $\Pi$ be a ZX diagram for a Pauli projector on $n$ qubits.
  First, we use process-state duality to create the following state on $2n$ qubits:
  \[
    \input{PauliWeb/process-state.tikz}
  \]

  This creates a state $\Pi'$ on $2n$ qubits.
  As it is on $2n$ qubits, we know it has $2n$ independent stabilisers.
  Let $S$ be the stabilisers of $\Pi$ spanned by $m$ generators.

  As $\Pi$ is self-adjoint, it is also co-stabilised by $S$, i.e.\@ for all $s \in S$, $\Pi \circ s = \Pi = s \circ \Pi$.

  But then, we can construct $2m$ independent stabilisers of $\Pi'$, namely:
  \[\input{PauliWeb/state-stabilisers.tikz}\]
  By \textcite{borghansZXcalculusQuantumStabilizer2019}, we know that we have corresponding Pauli webs and trivially the output of these Pauli webs, restricted to the first $n$ qubits commutes with $\Pi$.

  Furthermore, we know that $N(S) / S$ is spanned by $2n - 2m$ generators based on which we can construct $2n - 2m$ stabilisers of $\Pi'$
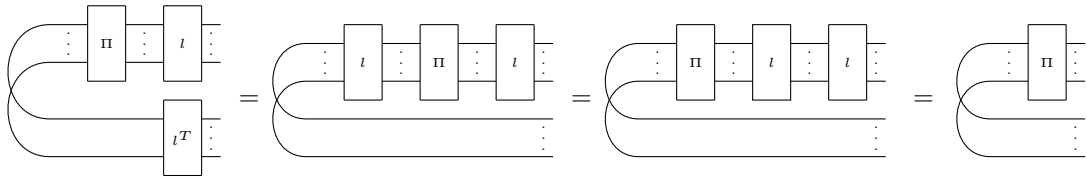
  \[
    \input{PauliWeb/state-logicals.tikz}
  \]
  where in the last step, we use the fact that Pauli strings are self-inverse, i.e.\@ $l \circ l = id$.

  But then we have given $2n$ stabilisers of $\Pi'$, therefore spanning all possible stabilisers and thus all possible Pauli webs of $\Pi$ up to detecting regions.
  As detecting regions do not affect the output of the Pauli webs and all outputs of the provided Pauli webs commute with $\Pi$, we can conclude that the output of all Pauli webs on $\Pi$ commute with $\Pi$.
\end{proof}

Now, we can prove:
\begin{restatable}{proposition}{webTwo}
  \label{prop:logicals}
  Let $C_t$ be the measurement circuit of the first $t$ time steps of a Floquet code with Pauli webs $P$, stabilising Pauli webs $S$ and co-stabilising webs $S_c$.
  Then the equivalence classes $P / (S \cup S_c)$ are in bijective correspondence with $N(S_t) / S_t$ observed by the map $b: P / (S \cup S_c) \to N(S_t) / S_t \dblcolon E \in P / (S \cup S_c) \mapsto out(e) \text{ for some } e \in E$.
\end{restatable}
\begin{proof}
  First, we show that the map $b$ is well-defined.
  Let $w_i$ be a Pauli web.
  By \autoref{prop:circuit-decomposition}, we know that $C_t$ can be written as $\Pi \circ U$ for a projector $\Pi$ with stabilisers $S_t$ and some unitary $U$.
  But then there exists an equivalent Pauli web $w'_i$ on $\Pi \circ U$ with the same inputs and outputs.
  By \autoref{prop:commuting-webs}, we know that $out(w_i)$ must commute with $\Pi$.
  But then $out(w_i) \in N(S_t)$.
  Next, we assume that Pauli webs $w_i, w_j$ are in the same equivalence class $E \in P / (S \cup S_c)$.
  Thus, we know that there exists a $s \in S, s_c \in S_c$ such that $w_i \otimes s \otimes s_c = w_j$.
  But then $out(w_j) = out(w_i \otimes s \otimes s_c)$.
  But as $out(s_c) = I$, we know there exists a stabiliser $out(s)$ such that $out(w_i) \otimes out(s) = out(w_j)$.
  Thus, $w_i$ and $w_j$ map to the same equivalence class and therefore, the map is well defined.

  Next, we show that the map is injective.
  Let the logical operators $l_1 \in L_1, l_2 \in L_2$ be in different quotient groups $L_1, L_2 \in N(S_t) / S_t$.
  Then for all $s \in S_t$, $l_1 \otimes s \not= l_2$.
  But then, there exists no stabilising Pauli web that makes the output of the two corresponding Pauli webs the same.
  As the only type of Pauli webs in $S \cup S_c$ that can change the outputs are stabilising Pauli webs, no Pauli web in $S \cup S_c$ can make the webs of $l_1$ and $l_2$ the same and thus, they must be in different equivalence classes.

  Next, we show that the map is surjective.
  Let $l$ be a logical operator.
  We can write $C_t$ as $\Pi \circ U$ for a projector $\Pi$ with stabilisers $S_t$ and some unitary $U$.
  As $l$ is a logical operator, we know that it commutes with $\Pi$.
  But then:
  \[
    \input{prerequisites/logicals-proof.tikz}
  \]
  Thus, the diagram stabilises when we place $UlU^\dagger$ on the inputs and $l$ on the outputs.
  By~\parencite{borghansZXcalculusQuantumStabilizer2019}, there exists a corresponding Pauli web $w$.
  As we have $out(w) = l$, $b$ is surjective and therefore bijective.
\end{proof}

\section{Distance-preserving rewrites --- Deferred proofs}
\label{appendix:distance-preserving-rewrites}

\inductivePlusDistRewrite*
\begin{proof}
  Given that we have $n$ four-legged spiders, each of which has two internal edges, we have a total of $2n$ internal edges and, therefore, $4n$ potential independent edge flips to consider.
  We observe that any detecting region has to include an even number of the four-legged boundary edges.

  We can construct a basis of size $n - 1$ for all detecting regions by considering the regions that include the first outer spider and the $i$-th outer spider for $i \in [2, n]$.
  We have:
  
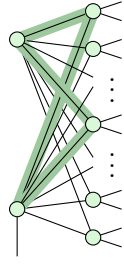
\[
    \begin{tikzpicture}
	\begin{pgfonlayer}{nodelayer}
		\node [style=Z dot] (0) at (-7.25, 3) {};
		\node [style=Z dot] (1) at (-7.25, 2) {};
		\node [style=Z dot] (2) at (-7.25, -2) {};
		\node [style=Z dot] (3) at (-7.25, -3) {};
		\node [style=Z dot] (4) at (-9.25, 2.25) {};
		\node [style=Z dot] (5) at (-9.25, -2.25) {};
		\node [style=none] (6) at (-6.5, 3.25) {};
		\node [style=none] (7) at (-6.5, 2.75) {};
		\node [style=none] (8) at (-6.5, 2.25) {};
		\node [style=none] (9) at (-6.5, 1.75) {};
		\node [style=none] (10) at (-6.5, -1.75) {};
		\node [style=none] (11) at (-6.5, -2.25) {};
		\node [style=none] (12) at (-6.5, -2.75) {};
		\node [style=none] (13) at (-6.5, -3.25) {};
		\node [style=none] (14) at (-7.25, 0) {};
		\node [style=none] (15) at (-7.25, -0.875) {};
		\node [style=none] (16) at (-7.25, 0) {};
		\node [style=none] (17) at (-7.25, -1.125) {};
		\node [style=none] (18) at (-6.75, -0.875) {$\vdots$};
		\node [style=Z dot] (19) at (-7.25, 0) {};
		\node [style=none] (20) at (-6.5, 0.25) {};
		\node [style=none] (21) at (-6.5, -0.25) {};
		\node [style=none] (22) at (-6.75, 1.125) {$\vdots$};
		\node [style=none] (23) at (-7.25, 1.25) {};
		\node [style=none] (24) at (-7.25, 1) {};
		\node [style=none] (25) at (-9.25, -3.5) {};
		\node [style=Z dot] (26) at (0.5, 3) {};
		\node [style=Z dot] (27) at (0.5, 2) {};
		\node [style=Z dot] (28) at (0.5, -2) {};
		\node [style=Z dot] (29) at (0.5, -3) {};
		\node [style=Z dot] (30) at (-1.5, 2.25) {};
		\node [style=Z dot] (31) at (-1.5, -2.25) {};
		\node [style=none] (32) at (1.25, 3.25) {};
		\node [style=none] (33) at (1.25, 2.75) {};
		\node [style=none] (34) at (1.25, 2.25) {};
		\node [style=none] (35) at (1.25, 1.75) {};
		\node [style=none] (36) at (1.25, -1.75) {};
		\node [style=none] (37) at (1.25, -2.25) {};
		\node [style=none] (38) at (1.25, -2.75) {};
		\node [style=none] (39) at (1.25, -3.25) {};
		\node [style=none] (40) at (0.5, 0) {};
		\node [style=none] (41) at (0.5, -0.875) {};
		\node [style=none] (42) at (0.5, 0) {};
		\node [style=none] (43) at (0.5, -1.125) {};
		\node [style=none] (44) at (1, -0.875) {$\vdots$};
		\node [style=Z dot] (45) at (0.5, 0) {};
		\node [style=none] (46) at (1.25, 0.25) {};
		\node [style=none] (47) at (1.25, -0.25) {};
		\node [style=none] (48) at (1, 1.125) {$\vdots$};
		\node [style=none] (49) at (0.5, 1.25) {};
		\node [style=none] (50) at (0.5, 1) {};
		\node [style=none] (51) at (-1.5, -3.5) {};
		\node [style=Z dot] (52) at (8.25, 3) {};
		\node [style=Z dot] (53) at (8.25, 2) {};
		\node [style=Z dot] (54) at (8.25, -2) {};
		\node [style=Z dot] (55) at (8.25, -3) {};
		\node [style=Z dot] (56) at (6.25, 2.25) {};
		\node [style=Z dot] (57) at (6.25, -2.25) {};
		\node [style=none] (58) at (9, 3.25) {};
		\node [style=none] (59) at (9, 2.75) {};
		\node [style=none] (60) at (9, 2.25) {};
		\node [style=none] (61) at (9, 1.75) {};
		\node [style=none] (62) at (9, -1.75) {};
		\node [style=none] (63) at (9, -2.25) {};
		\node [style=none] (64) at (9, -2.75) {};
		\node [style=none] (65) at (9, -3.25) {};
		\node [style=none] (66) at (8.25, 0) {};
		\node [style=none] (67) at (8.25, -0.875) {};
		\node [style=none] (68) at (8.25, 0) {};
		\node [style=none] (69) at (8.25, -1.125) {};
		\node [style=none] (70) at (8.75, -0.875) {$\vdots$};
		\node [style=Z dot] (71) at (8.25, 0) {};
		\node [style=none] (72) at (9, 0.25) {};
		\node [style=none] (73) at (9, -0.25) {};
		\node [style=none] (74) at (8.75, 1.125) {$\vdots$};
		\node [style=none] (75) at (8.25, 1.25) {};
		\node [style=none] (76) at (8.25, 1) {};
		\node [style=none] (77) at (6.25, -3.5) {};
		\node [style=none] (78) at (3.75, 0) {$\ldots$};
		\node [style=none] (79) at (-4, 0) {$\ldots$};
	\end{pgfonlayer}
	\begin{pgfonlayer}{edgelayer}
		\draw (2) to (4);
		\draw (2) to (5);
		\draw (3) to (4);
		\draw (5) to (3);
		\draw (13.center) to (3);
		\draw (12.center) to (3);
		\draw (11.center) to (2);
		\draw (10.center) to (2);
		\draw (9.center) to (1);
		\draw (8.center) to (1);
		\draw (7.center) to (0);
		\draw (6.center) to (0);
		\draw (4) to (15.center);
		\draw (5) to (17.center);
		\draw [style=Z Web] (4) to (0);
		\draw [style=Z Web] (5) to (0);
		\draw (4) to (1);
		\draw (1) to (5);
		\draw (19) to (20.center);
		\draw (19) to (21.center);
		\draw (5) to (24.center);
		\draw (4) to (23.center);
		\draw (25.center) to (5);
		\draw (4) to (15.center);
		\draw (4) to (23.center);
		\draw (7.center) to (0);
		\draw (12.center) to (3);
		\draw (19) to (20.center);
		\draw (13.center) to (3);
		\draw (2) to (4);
		\draw (6.center) to (0);
		\draw (5) to (17.center);
		\draw (3) to (4);
		\draw (4) to (1);
		\draw (9.center) to (1);
		\draw (19) to (21.center);
		\draw (25.center) to (5);
		\draw (8.center) to (1);
		\draw [style=Z Web] (5) to (0);
		\draw (1) to (5);
		\draw (11.center) to (2);
		\draw (2) to (5);
		\draw [style=Z Web] (4) to (0);
		\draw (10.center) to (2);
		\draw (5) to (3);
		\draw (19) to (5);
		\draw (4) to (19);
		\draw (28) to (30);
		\draw (28) to (31);
		\draw (29) to (30);
		\draw (31) to (29);
		\draw (39.center) to (29);
		\draw (38.center) to (29);
		\draw (37.center) to (28);
		\draw (36.center) to (28);
		\draw (35.center) to (27);
		\draw (34.center) to (27);
		\draw (33.center) to (26);
		\draw (32.center) to (26);
		\draw (30) to (41.center);
		\draw (31) to (43.center);
		\draw [style=Z Web] (30) to (26);
		\draw [style=Z Web] (31) to (26);
		\draw (30) to (27);
		\draw (27) to (31);
		\draw (45) to (46.center);
		\draw (45) to (47.center);
		\draw (31) to (50.center);
		\draw (30) to (49.center);
		\draw (51.center) to (31);
		\draw (30) to (41.center);
		\draw (30) to (49.center);
		\draw (33.center) to (26);
		\draw (38.center) to (29);
		\draw (45) to (46.center);
		\draw (39.center) to (29);
		\draw (28) to (30);
		\draw (32.center) to (26);
		\draw (31) to (43.center);
		\draw (29) to (30);
		\draw (30) to (27);
		\draw (35.center) to (27);
		\draw (45) to (47.center);
		\draw (51.center) to (31);
		\draw (34.center) to (27);
		\draw [style=Z Web] (31) to (26);
		\draw (27) to (31);
		\draw (37.center) to (28);
		\draw (28) to (31);
		\draw [style=Z Web] (30) to (26);
		\draw (36.center) to (28);
		\draw (31) to (29);
		\draw (45) to (31);
		\draw (30) to (45);
		\draw (54) to (56);
		\draw (54) to (57);
		\draw (55) to (56);
		\draw (57) to (55);
		\draw (65.center) to (55);
		\draw (64.center) to (55);
		\draw (63.center) to (54);
		\draw (62.center) to (54);
		\draw (61.center) to (53);
		\draw (60.center) to (53);
		\draw (59.center) to (52);
		\draw (58.center) to (52);
		\draw (56) to (67.center);
		\draw (57) to (69.center);
		\draw [style=Z Web] (56) to (52);
		\draw [style=Z Web] (57) to (52);
		\draw (56) to (53);
		\draw (53) to (57);
		\draw (71) to (72.center);
		\draw (71) to (73.center);
		\draw (57) to (76.center);
		\draw (56) to (75.center);
		\draw (77.center) to (57);
		\draw (56) to (67.center);
		\draw (56) to (75.center);
		\draw (59.center) to (52);
		\draw (64.center) to (55);
		\draw (71) to (72.center);
		\draw (65.center) to (55);
		\draw (54) to (56);
		\draw (58.center) to (52);
		\draw (57) to (69.center);
		\draw (55) to (56);
		\draw (56) to (53);
		\draw (61.center) to (53);
		\draw (71) to (73.center);
		\draw (77.center) to (57);
		\draw (60.center) to (53);
		\draw [style=Z Web] (57) to (52);
		\draw (53) to (57);
		\draw (63.center) to (54);
		\draw (54) to (57);
		\draw [style=Z Web] (56) to (52);
		\draw (62.center) to (54);
		\draw (57) to (55);
		\draw (71) to (57);
		\draw (56) to (71);
		\draw [style=Z Web] (4) to (1);
		\draw [style=Z Web] (1) to (5);
		\draw [style=Z Web] (30) to (45);
		\draw [style=Z Web] (45) to (31);
		\draw [style=Z Web] (57) to (55);
		\draw [style=Z Web] (55) to (56);
	\end{pgfonlayer}
\end{tikzpicture}

  \]
  As there are $n - 1$ such regions forming the parity-check matrix $P$, it has a nullspace of size $4n - (n - 1) = 3n + 1$.

  We now give a basis for this null space, prove that any basis vector can be pushed out without increasing its weight and then argue that any combination of basis vectors can be pushed out without increasing their weight.

  The first $2n$ basis vectors $\vec{g_1}, \dots, \vec{g_{2n}}$ are single edge flips of type Z\@.
  Since all detecting regions are green, these cannot be detected and, thus, live in the null space of $P$.
  As we can push spiders of the same type past each other, we can push these flips, no matter where they occur, to the boundary without increasing their weight.

  Secondly, we consider the error $\vec{r_0}$:
         
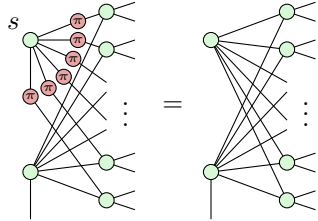
\[\begin{tikzpicture}
	\begin{pgfonlayer}{nodelayer}
		\node [style=Z dot] (0) at (-5.5, 2.5) {};
		\node [style=Z dot] (1) at (-5.5, 1.5) {};
		\node [style=Z dot] (2) at (-5.5, -1.5) {};
		\node [style=Z dot] (3) at (-5.5, -2.5) {};
		\node [style=Z dot] (4) at (-7.5, 1.75) {};
		\node [style=Z dot] (5) at (-7.5, -1.75) {};
		\node [style=none] (6) at (-4.75, 2.75) {};
		\node [style=none] (7) at (-4.75, 2.25) {};
		\node [style=none] (8) at (-4.75, 1.75) {};
		\node [style=none] (9) at (-4.75, 1.25) {};
		\node [style=none] (10) at (-4.75, -1.25) {};
		\node [style=none] (11) at (-4.75, -1.75) {};
		\node [style=none] (12) at (-4.75, -2.25) {};
		\node [style=none] (13) at (-4.75, -2.75) {};
		\node [style=none] (14) at (-5.5, 0.625) {};
		\node [style=none] (15) at (-5.5, -0.375) {};
		\node [style=none] (16) at (-5.5, 0.375) {};
		\node [style=none] (17) at (-5.5, -0.625) {};
		\node [style=none] (18) at (-5, 0) {$\vdots$};
		\node [style=Z dot] (19) at (-0.75, 2.5) {};
		\node [style=Z dot] (20) at (-0.75, 1.5) {};
		\node [style=Z dot] (21) at (-0.75, -1.5) {};
		\node [style=Z dot] (22) at (-0.75, -2.5) {};
		\node [style=Z dot] (23) at (-2.75, 1.75) {};
		\node [style=Z dot] (24) at (-2.75, -1.75) {};
		\node [style=none] (25) at (0, 2.75) {};
		\node [style=none] (26) at (0, 2.25) {};
		\node [style=none] (27) at (0, 1.75) {};
		\node [style=none] (28) at (0, 1.25) {};
		\node [style=none] (29) at (0, -1.25) {};
		\node [style=none] (30) at (0, -1.75) {};
		\node [style=none] (31) at (0, -2.25) {};
		\node [style=none] (32) at (0, -2.75) {};
		\node [style=none] (33) at (-0.75, 0.625) {};
		\node [style=none] (34) at (-0.75, -0.375) {};
		\node [style=none] (35) at (-0.75, 0.375) {};
		\node [style=none] (36) at (-0.75, -0.625) {};
		\node [style=none] (37) at (-0.25, 0) {$\vdots$};
		\node [style=X dot] (38) at (-7.025, 0.475) {\tiny $\pi$};
		\node [style=X dot] (39) at (-6.625, 0.775) {\tiny $\pi$};
		\node [style=X dot] (40) at (-6.375, 1.25) {\tiny $\pi$};
		\node [style=X dot] (41) at (-6.25, 1.75) {\tiny $\pi$};
		\node [style=X dot] (42) at (-6.25, 2.25) {\tiny $\pi$};
		\node [style=X dot] (43) at (-7.5, 0.25) {\tiny $\pi$};
		\node [style=none] (44) at (-3.75, 0) {=};
		\node [style=none] (45) at (-7.95, 2.2) {$s$};
		\node [style=none] (46) at (-7.5, -3) {};
		\node [style=none] (47) at (-2.75, -3) {};
	\end{pgfonlayer}
	\begin{pgfonlayer}{edgelayer}
		\draw (0) to (5);
		\draw (1) to (5);
		\draw (2) to (5);
		\draw (5) to (3);
		\draw (13.center) to (3);
		\draw (12.center) to (3);
		\draw (11.center) to (2);
		\draw (10.center) to (2);
		\draw (9.center) to (1);
		\draw (8.center) to (1);
		\draw (7.center) to (0);
		\draw (6.center) to (0);
		\draw (5) to (17.center);
		\draw (5) to (16.center);
		\draw (23) to (19);
		\draw (19) to (24);
		\draw (20) to (23);
		\draw (20) to (24);
		\draw (21) to (23);
		\draw (21) to (24);
		\draw (22) to (23);
		\draw (24) to (22);
		\draw (32.center) to (22);
		\draw (31.center) to (22);
		\draw (30.center) to (21);
		\draw (29.center) to (21);
		\draw (28.center) to (20);
		\draw (27.center) to (20);
		\draw (26.center) to (19);
		\draw (25.center) to (19);
		\draw (23) to (33.center);
		\draw (23) to (34.center);
		\draw (24) to (36.center);
		\draw (24) to (35.center);
		\draw (42) to (0);
		\draw (41) to (1);
		\draw (40) to (14.center);
		\draw (39) to (15.center);
		\draw (38) to (2);
		\draw (43) to (3);
		\draw (4) to (43);
		\draw (4) to (38);
		\draw (4) to (39);
		\draw (4) to (40);
		\draw (4) to (41);
		\draw (42) to (4);
		\draw (46.center) to (5);
		\draw (47.center) to (24);
	\end{pgfonlayer}
\end{tikzpicture}
\]
  Since each detecting region overlaps with this error on an even number of legs, it is not detectable.
  However, this error is trivial, as firing $s$ removes all edge flips.

  The final set of errors $\vec{r_1}, \dots, \vec{r_{n}}$ consists of errors on both the internal edges of the four-legged spiders.
         
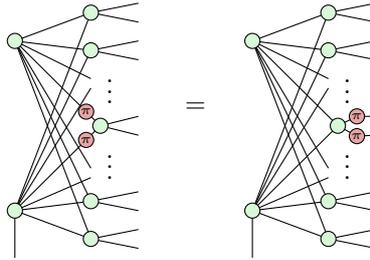
\[\begin{tikzpicture}
	\begin{pgfonlayer}{nodelayer}
		\node [style=Z dot] (44) at (-2.75, 2.5) {};
		\node [style=Z dot] (45) at (-2.75, 1.5) {};
		\node [style=Z dot] (46) at (-2.75, -2.5) {};
		\node [style=Z dot] (47) at (-2.75, -3.5) {};
		\node [style=Z dot] (48) at (-4.75, 1.75) {};
		\node [style=Z dot] (49) at (-4.75, -2.75) {};
		\node [style=none] (50) at (-1.5, 2.75) {};
		\node [style=none] (51) at (-1.5, 2.25) {};
		\node [style=none] (52) at (-1.5, 1.75) {};
		\node [style=none] (53) at (-1.5, 1.25) {};
		\node [style=none] (54) at (-1.5, -2.25) {};
		\node [style=none] (55) at (-1.5, -2.75) {};
		\node [style=none] (56) at (-1.5, -3.25) {};
		\node [style=none] (57) at (-1.5, -3.75) {};
		\node [style=none] (58) at (-2.5, -0.5) {};
		\node [style=none] (59) at (-2.75, -1.625) {};
		\node [style=none] (60) at (-2.5, -0.5) {};
		\node [style=none] (61) at (-2.75, -1.875) {};
		\node [style=none] (62) at (-2.25, -1.375) {$\vdots$};
		\node [style=Z dot] (63) at (-2.5, -0.5) {};
		\node [style=none] (64) at (-1.5, -0.25) {};
		\node [style=none] (65) at (-1.5, -0.75) {};
		\node [style=none] (66) at (-2.25, 0.625) {$\vdots$};
		\node [style=none] (67) at (-2.75, 0.75) {};
		\node [style=none] (68) at (-2.75, 0.5) {};
		\node [style=none] (94) at (0, 0) {=};
		\node [style=X dot] (95) at (-2.875, -0.125) {\tiny $\pi$};
		\node [style=X dot] (96) at (-2.875, -0.875) {\tiny $\pi$};
		\node [style=Z dot] (98) at (3.5, 2.5) {};
		\node [style=Z dot] (99) at (3.5, 1.5) {};
		\node [style=Z dot] (100) at (3.5, -2.5) {};
		\node [style=Z dot] (101) at (3.5, -3.5) {};
		\node [style=Z dot] (102) at (1.5, 1.75) {};
		\node [style=Z dot] (103) at (1.5, -2.75) {};
		\node [style=none] (104) at (4.75, 2.75) {};
		\node [style=none] (105) at (4.75, 2.25) {};
		\node [style=none] (106) at (4.75, 1.75) {};
		\node [style=none] (107) at (4.75, 1.25) {};
		\node [style=none] (108) at (4.75, -2.25) {};
		\node [style=none] (109) at (4.75, -2.75) {};
		\node [style=none] (110) at (4.75, -3.25) {};
		\node [style=none] (111) at (4.75, -3.75) {};
		\node [style=none] (112) at (3.75, -0.5) {};
		\node [style=none] (113) at (3.5, -1.625) {};
		\node [style=none] (114) at (3.75, -0.5) {};
		\node [style=none] (115) at (3.5, -1.875) {};
		\node [style=none] (116) at (4, -1.375) {$\vdots$};
		\node [style=Z dot] (117) at (3.75, -0.5) {};
		\node [style=none] (118) at (4.75, -0.25) {};
		\node [style=none] (119) at (4.75, -0.75) {};
		\node [style=none] (120) at (4, 0.625) {$\vdots$};
		\node [style=none] (121) at (3.5, 0.75) {};
		\node [style=none] (122) at (3.5, 0.5) {};
		\node [style=X dot] (123) at (4.25, -0.25) {\tiny $\pi$};
		\node [style=X dot] (124) at (4.25, -0.775) {\tiny $\pi$};
		\node [style=none] (125) at (-4.75, -4) {};
		\node [style=none] (126) at (1.5, -4) {};
	\end{pgfonlayer}
	\begin{pgfonlayer}{edgelayer}
		\draw (46) to (48);
		\draw (46) to (49);
		\draw (47) to (48);
		\draw (49) to (47);
		\draw (57.center) to (47);
		\draw (56.center) to (47);
		\draw (55.center) to (46);
		\draw (54.center) to (46);
		\draw (53.center) to (45);
		\draw (52.center) to (45);
		\draw (51.center) to (44);
		\draw (50.center) to (44);
		\draw (48) to (59.center);
		\draw (49) to (61.center);
		\draw (48) to (45);
		\draw (45) to (49);
		\draw (63) to (64.center);
		\draw (63) to (65.center);
		\draw (49) to (68.center);
		\draw (48) to (67.center);
		\draw (48) to (95);
		\draw (95) to (63);
		\draw (49) to (96);
		\draw (96) to (63);
		\draw (100) to (102);
		\draw (100) to (103);
		\draw (101) to (102);
		\draw (103) to (101);
		\draw (111.center) to (101);
		\draw (110.center) to (101);
		\draw (109.center) to (100);
		\draw (108.center) to (100);
		\draw (107.center) to (99);
		\draw (106.center) to (99);
		\draw (105.center) to (98);
		\draw (104.center) to (98);
		\draw (102) to (113.center);
		\draw (103) to (115.center);
		\draw (102) to (99);
		\draw (99) to (103);
		\draw (103) to (122.center);
		\draw (102) to (121.center);
		\draw (102) to (117);
		\draw (117) to (103);
		\draw (117) to (123);
		\draw (123) to (118.center);
		\draw (117) to (124);
		\draw (124) to (119.center);
		\draw (48) to (44);
		\draw (44) to (49);
		\draw (102) to (98);
		\draw (98) to (103);
		\draw (125.center) to (49);
		\draw (126.center) to (103);
	\end{pgfonlayer}
\end{tikzpicture}
\]
  Each detecting region either highlights both internal edges of an outer spider or neither, therefore, this error is not detectable.
  However, by firing the corresponding spider, we can push the error to the boundary edges.

  We now have $2n + 1 + n = 3n + 1$ independent errors, thereby spanning the nullspace of $P$.
  We now argue that any undetectable error can be pushed to the boundary edges.
  We know that any error $\vec{v}$ in the nullspace of $P$ can be expressed as a linear combination of the basis vectors, i.e.\@ $\vec{v} = \lambda_{g_1} \vec{g_1} + \dots +  \lambda_{g_{2n}} \vec{g_{2n}} + \lambda_{r_0} \vec{r_0} + \lambda_{r_1} \vec{r_1} + \dots +  \lambda_{r_n} \vec{r_n}$.

  As for the proofs above, we will first consider errors that only consist of X flips.
  First, we consider errors that are linear combinations of $\vec{r_1}, \dots, \vec{r_{n}}$.
  Since these do not have shared edge flips, $|\lambda_{r_1} \vec{r_1} + \dots +  \lambda_{r_{n}} \vec{r_{n}}| = \lambda_{r_1} |\vec{r_1}| + \dots +  \lambda_{r_{n}} |\vec{r_{n}}|$.
  Then, since we can push the individual errors to the boundary, we can also push the combined error to the boundary without increasing the distance.

  For errors, of the form $\vec{r_0} + \lambda_{r_1} \vec{r_1} + \dots +  \lambda_{r_{n}} \vec{r_{n}}$, we observe that they all have weight $n$.
  For any $\vec{r_i}$ ($i > 0$) added to $\vec{r_0}$, we add one edge flip and remove one, therefore preserving the weight of the total error.

  Since such an error is composed of basis elements, we know that we can push it out.
  If $\sum \lambda_{r_i} \leq \frac{n}{2}$, that is, if at most half of $r_1, \dots, r_{n}$ are included, then we can push the errors of each basis out individually.
  This creates an error of weight at most $n$.

  However, if $\sum \lambda_{r_i} > \frac{n}{2}$, that is, if strictly more than half of $r_1, \dots, r_{n}$ are included, then pushing the errors individually creates an error of weight at least $n + 2$.
  Now, if you fire the spiders according to the Pauli-web that witnesses the all $X$ stabiliser, i.e.\@ fire all spiders in red, you get an edge flip for each boundary edge.
  This flips the errors on each external edge, resulting in at most $(2n + 1) - \left(n + 2\right) = n - 1$ edge flips on the boundary.
  Therefore, such an error can also be pushed to the boundary without increasing its weight.

  As we only have green spiders and the X errors can be pushed out, we know that any combination of errors can therefore be pushed to the outside without increasing their weight.
  Therefore, the recursive rewrite for $(2n + 1)$-legged spiders satisfies \autoref{prop:pushing-out}.
  As the LHS has no internal edges, it also satisfies \autoref{prop:pushing-out} in the other direction and so the rewrite is distance-preserving.
\end{proof}

\singleQubitPauli*
\begin{proof}
  The RHS has one red detecting region on all the internal edges.
  Therefore, we only have to consider $Z$ errors of even weight:
  \[\input{floquetification/rewrite-proof.tikz}\]
  As such, we can push undetectable $Z$ errors to the boundary nodes without increasing their weight.
  $X$ errors on the internal edges are stabilised by the $Z$ spiders or can be pushed to the boundary without increasing in weight.
  Similarly, $Y$ errors do not increase in weight.
  As the LHS has no internal edges, this rewrite is distance-preserving.
\end{proof}
\section{Measurement-circuit flow --- Deferred proofs}

\begin{proposition}
    \label{appendix:proof-f-bound}
    Let 
    \[
 f(n) =
    \begin{dcases}
 0 & \text{if } n = 4\\
 f\! \left(\frac{n}{2}\right) & \text{if } n \bmod 4 = 0\\
 1 + f\left(n + 2\right) & \text{if } n \bmod 4 = 2\\
    \end{dcases}
\]
We have $f(n) \leq \log_2(n)$.
\end{proposition}
\begin{proof}
    We observe that only the third case increases the value of $f(n)$. 
    Therefore, $f(n)$ is largest if this case occurs the most often which is at most every other time.
    Working backwards from the base case, the most often we can get case three is by alternating between case two and case three.
    If we do this for $i$ steps, $n$ will be larger than $4 \cdot 2^i$.
    But then we have:
    \begin{align*}
        \log_2(n) &\geq \log_2(4 * 2^i) \\
        &= i + 2 \geq i = f(n)
    \end{align*}
    Thus, even in the worst case scenario, the bound holds and therefore it must hold in all cases. 
\end{proof}

\begin{proposition}
    \label{appendix:proof-g-bound}
    Let 
    \[
    g(n) =
        \begin{dcases}
    2 & \text{if } n = 4\\
    n + 2g\!\left(\frac{n}{2}\right) & \text{if } n \bmod 4 = 0\\
    g\left(n + 2\right) & \text{if } n \bmod 4 = 2\\
        \end{dcases}
    \]
We have $g(n) \leq 2 n \log_2(n)$.
\end{proposition}
\begin{proof}
    We prove this inductively. \\
    Base case: \\
    For $n = 4, 6, 8, 10, 12$, we can manually check that the equation holds. \\\\
    Inductive step: \\
    Let's assume the IH holds for all $n' < n$. Then we will show that it also holds for $n$. \\
    Case 1: $n \bmod 4 = 0$.\\
    We have:
    \begin{align*}
        g(n) &= n + 2g\!\left(\frac{n}{2}\right) \\
        &\leq n + 2 \times \left(2 \times \frac{n}{2} \times \log_2(\frac{n}{2})\right) & \text{(IH)} \\
        &= n + \left( 2 \times n \times (\log_2(\frac{n}{2}))\right) \\
        &= n + \left( 2 \times n \times (\log_2(n) - \log_2(2))\right) \\
        &= n + \left( 2 \times n \times (\log_2(n) - 1)\right) \\
        &= n + \left( 2 \times n \log_2(n) - 2 \times n \right) \\
        &= 2 \times n \log_2(n) - n \\
        &\leq  2 \times n \log_2(n)
    \end{align*}
    Case 2: $n \bmod 4 = 2$.\\
    We have:
    \begin{align*}
        g(n) &= g(n + 2) \\
        &= n + 2 + 2g\!\left(\frac{n + 2}{2}\right) \\
        &\leq n + 2 + 2 \times \left(2 \times \frac{n + 2}{2} \log_2(\frac{n + 2}{2}) - \frac{n + 2}{2}\right) & \text{(same reasoning as above)} \\
        &= n + 2 + \left(2 \times (n + 2) \log_2(\frac{n + 2}{2}) - (n + 2)\right) \\
        &= 2 \times (n + 2) \log_2(\frac{n + 2}{2}) \\
        &= 2 \times (n + 2) (\log_2(n + 2) - 1) \\
        &= 2 \times n (\log_2(n + 2) - 1) + 4 (\log_2(n + 2) - 1)\\
        &= 2 \times n \log_2(n + 2) - 2n + 4 \log_2(n + 2) - 4\\
        &\leq 2 \times n \log_2(n) + n - 2n + 4 \log_2(n + 2) - 4\\
        &= 2 \times n \log_2(n) - n + 4 \log_2(n + 2) - 4\\
        &\leq 2 \times n \log_2(n) & \text{for } n > 12\\
    \end{align*}
    But then we have proven that the bound holds. 
    We could prove this bound for a lower constant, however, as we did not optimise our procedure for gate count, we omit this here. 
\end{proof}

\section{Floquetification --- Deferred proofs}
\begin{proposition}
    \label{prop:removing-swaps-422}
    We have: 
    \begin{center}
        \resizebox{\columnwidth}{!}{
            \begin{tikzpicture}
	\begin{pgfonlayer}{nodelayer}
		\node [style=none] (0) at (-5.25, 1.5) {};
		\node [style=none] (1) at (-2.75, 1.5) {};
		\node [style=none] (2) at (-5.25, 2.5) {};
		\node [style=none] (3) at (-7, 2.5) {};
		\node [style=none] (4) at (-2.75, 0.5) {};
		\node [style=X dot] (5) at (-1.75, 0.5) {};
		\node [style=Z dot] (6) at (-2.75, 0.5) {};
		\node [style=X dot] (7) at (-1.75, 0.5) {};
		\node [style=Z dot] (8) at (-5.25, 1.5) {};
		\node [style=Z dot] (9) at (-2.75, 1.5) {};
		\node [style=Z dot] (10) at (-5.25, 2.5) {};
		\node [style=Z dot] (11) at (-7, 2.5) {};
		\node [style=none] (12) at (-4.25, 1.5) {};
		\node [style=none] (13) at (-3.75, 1.5) {};
		\node [style=X dot] (14) at (-4.25, 1.5) {};
		\node [style=X dot] (15) at (-3.75, 1.5) {};
		\node [style=none] (16) at (-5.25, -1.5) {};
		\node [style=none] (17) at (-2.75, -1.5) {};
		\node [style=none] (18) at (-5.25, -0.5) {};
		\node [style=none] (19) at (-8, -0.5) {};
		\node [style=none] (20) at (-2.75, -2.5) {};
		\node [style=X dot] (21) at (-1.75, -2.5) {};
		\node [style=Z dot] (22) at (-2.75, -2.5) {};
		\node [style=X dot] (23) at (-1.75, -2.5) {};
		\node [style=Z dot] (24) at (-5.25, -1.5) {};
		\node [style=Z dot] (25) at (-2.75, -1.5) {};
		\node [style=Z dot] (26) at (-5.25, -0.5) {};
		\node [style=Z dot] (27) at (-7, -0.5) {};
		\node [style=none] (28) at (-4.25, -1.5) {};
		\node [style=none] (29) at (-3.75, -1.5) {};
		\node [style=X dot] (30) at (-4.25, -1.5) {};
		\node [style=X dot] (31) at (-3.75, -1.5) {};
		\node [style=none] (32) at (-8, 1.5) {};
		\node [style=none] (33) at (-8, 0.5) {};
		\node [style=none] (34) at (-8, -1.5) {};
		\node [style=none] (35) at (-8, -2.5) {};
		\node [style=H box] (36) at (-1.25, 2.5) {};
		\node [style=H box] (37) at (-1.25, 1.5) {};
		\node [style=H box] (38) at (-1.25, -0.5) {};
		\node [style=H box] (39) at (-1.25, -1.5) {};
		\node [style=none] (43) at (-0.25, 2.5) {};
		\node [style=none] (44) at (2.25, 2.5) {};
		\node [style=none] (45) at (-0.25, 0.5) {};
		\node [style=X dot] (46) at (-1.25, 0.5) {};
		\node [style=none] (47) at (2.25, 1.5) {};
		\node [style=Z dot] (48) at (2.25, 1.5) {};
		\node [style=none] (49) at (4.25, 1.5) {};
		\node [style=Z dot] (50) at (-0.25, 2.5) {};
		\node [style=Z dot] (51) at (2.25, 2.5) {};
		\node [style=Z dot] (52) at (-0.25, 0.5) {};
		\node [style=X dot] (53) at (-1.25, 0.5) {};
		\node [style=none] (54) at (0.75, 2.5) {};
		\node [style=none] (55) at (1.25, 2.5) {};
		\node [style=X dot] (56) at (0.75, 2.5) {};
		\node [style=X dot] (57) at (1.25, 2.5) {};
		\node [style=none] (58) at (-0.25, -0.5) {};
		\node [style=none] (59) at (2.25, -0.5) {};
		\node [style=none] (60) at (-0.25, -2.5) {};
		\node [style=X dot] (61) at (-1.25, -2.5) {};
		\node [style=none] (62) at (2.25, -1.5) {};
		\node [style=Z dot] (63) at (2.25, -1.5) {};
		\node [style=none] (64) at (4.25, -1.5) {};
		\node [style=Z dot] (65) at (-0.25, -0.5) {};
		\node [style=Z dot] (66) at (2.25, -0.5) {};
		\node [style=Z dot] (67) at (-0.25, -2.5) {};
		\node [style=X dot] (68) at (-1.25, -2.5) {};
		\node [style=none] (69) at (0.75, -0.5) {};
		\node [style=none] (70) at (1.25, -0.5) {};
		\node [style=X dot] (71) at (0.75, -0.5) {};
		\node [style=X dot] (72) at (1.25, -0.5) {};
		\node [style=H box] (73) at (3.25, 2.5) {};
		\node [style=H box] (74) at (3.25, 0.5) {};
		\node [style=H box] (75) at (3.25, -2.5) {};
		\node [style=H box] (76) at (3.25, -0.5) {};
		\node [style=none] (79) at (5.75, 1.5) {};
		\node [style=none] (80) at (5.75, 0.5) {};
		\node [style=none] (81) at (5.75, -1.5) {};
		\node [style=none] (82) at (5.75, -2.5) {};
		\node [style=none] (83) at (4.25, 2.5) {};
		\node [style=none] (84) at (4.25, 0.5) {};
		\node [style=none] (85) at (4.25, -0.5) {};
		\node [style=none] (86) at (4.25, -2.5) {};
		\node [style=none] (87) at (5.75, 2.5) {};
		\node [style=none] (88) at (5.75, -0.5) {};
		\node [style=none] (89) at (7.75, 1.5) {};
		\node [style=none] (90) at (7.75, 0.5) {};
		\node [style=none] (91) at (7.75, -1.5) {};
		\node [style=none] (92) at (7.75, -2.5) {};
		\node [style=X dot] (93) at (6.25, 2.5) {};
		\node [style=X dot] (94) at (6.25, -0.5) {};
		\node [style=X dot] (95) at (6.75, -0.5) {};
		\node [style=X dot] (96) at (6.75, 2.5) {};
		\node [style=none] (97) at (7.75, 2.5) {};
		\node [style=none] (98) at (7.75, -0.5) {};
		\node [style=Z dot] (99) at (-1.75, 1) {};
		\node [style=Z dot] (100) at (-1.25, 1) {};
		\node [style=Z dot] (101) at (-1.75, -2) {};
		\node [style=Z dot] (102) at (-1.25, -2) {};
		\node [style=Z dot] (103) at (6.25, 0) {};
		\node [style=Z dot] (104) at (6.75, 0) {};
		\node [style=Z dot] (105) at (6.25, 3) {};
		\node [style=Z dot] (106) at (6.75, 3) {};
		\node [style=none] (108) at (7.75, -3.25) {};
		\node [style=none] (109) at (7.75, 3.25) {};
		\node [style=none] (110) at (7.825, -3.25) {};
		\node [style=none] (111) at (7.825, 3.25) {};
		\node [style=none] (112) at (7.625, 0.125) {.};
		\node [style=none] (113) at (7.625, -0.125) {.};
		\node [style=none] (114) at (-6, -3.75) {$\infty$};
		\node [style=none] (115) at (-5.925, -3.25) {};
		\node [style=none] (116) at (-5.925, 3.25) {};
		\node [style=none] (117) at (-6, -3.25) {};
		\node [style=none] (118) at (-6, 3.25) {};
		\node [style=none] (119) at (-5.8, 0.125) {.};
		\node [style=none] (120) at (-5.8, -0.125) {.};
	\end{pgfonlayer}
	\begin{pgfonlayer}{edgelayer}
		\draw (11) to (10);
		\draw (10) to (8);
		\draw (15) to (9);
		\draw (8) to (14);
		\draw (14) to (15);
		\draw (9) to (6);
		\draw (6) to (7);
		\draw (27) to (26);
		\draw (26) to (24);
		\draw (31) to (25);
		\draw (24) to (30);
		\draw (30) to (31);
		\draw (25) to (22);
		\draw (22) to (23);
		\draw (14) to (30);
		\draw (31) to (15);
		\draw (32.center) to (8);
		\draw (33.center) to (6);
		\draw (34.center) to (24);
		\draw (35.center) to (22);
		\draw (10) to (36);
		\draw (9) to (37);
		\draw (26) to (38);
		\draw (25) to (39);
		\draw (53) to (52);
		\draw (52) to (50);
		\draw (57) to (51);
		\draw (50) to (56);
		\draw (56) to (57);
		\draw (51) to (48);
		\draw (48) to (49.center);
		\draw (68) to (67);
		\draw (67) to (65);
		\draw (72) to (66);
		\draw (65) to (71);
		\draw (71) to (72);
		\draw (66) to (63);
		\draw (63) to (64.center);
		\draw (56) to (71);
		\draw (72) to (57);
		\draw (36) to (50);
		\draw (37) to (48);
		\draw (38) to (65);
		\draw (39) to (63);
		\draw (67) to (75);
		\draw (66) to (76);
		\draw (51) to (73);
		\draw (52) to (74);
		\draw [in=180, out=0] (85.center) to (82.center);
		\draw [in=-180, out=0, looseness=1.25] (86.center) to (81.center);
		\draw [in=180, out=0] (83.center) to (80.center);
		\draw [in=-180, out=0, looseness=1.25] (84.center) to (79.center);
		\draw (73) to (83.center);
		\draw (74) to (84.center);
		\draw (76) to (85.center);
		\draw (75) to (86.center);
		\draw [in=-180, out=0] (64.center) to (88.center);
		\draw [in=-180, out=0] (49.center) to (87.center);
		\draw (79.center) to (89.center);
		\draw (80.center) to (90.center);
		\draw (81.center) to (91.center);
		\draw (82.center) to (92.center);
		\draw (87.center) to (93);
		\draw (88.center) to (94);
		\draw (96) to (97.center);
		\draw (95) to (98.center);
		\draw (93) to (96);
		\draw (94) to (95);
		\draw (23) to (68);
		\draw (7) to (53);
		\draw (106) to (96);
		\draw (105) to (93);
		\draw (103) to (94);
		\draw (104) to (95);
		\draw (101) to (23);
		\draw (102) to (68);
		\draw (99) to (7);
		\draw (100) to (53);
		\draw (109.center) to (108.center);
		\draw [style=thick] (111.center) to (110.center);
		\draw (116.center) to (115.center);
		\draw [style=thick] (118.center) to (117.center);
	\end{pgfonlayer}
\end{tikzpicture}

             = \begin{tikzpicture}
	\begin{pgfonlayer}{nodelayer}
		\node [style=Z dot] (6) at (-8.75, 0.5) {};
		\node [style=X dot] (7) at (-8, 0.5) {};
		\node [style=Z dot] (8) at (-10.75, 1.5) {};
		\node [style=Z dot] (9) at (-8.75, 1.5) {};
		\node [style=Z dot] (10) at (-10.75, 2.5) {};
		\node [style=Z dot] (11) at (-12.25, 2.5) {};
		\node [style=X dot] (14) at (-10, 1.5) {};
		\node [style=X dot] (15) at (-9.5, 1.5) {};
		\node [style=Z dot] (22) at (-8.75, -2.5) {};
		\node [style=X dot] (23) at (-8, -2.5) {};
		\node [style=Z dot] (24) at (-10.75, -1.5) {};
		\node [style=Z dot] (25) at (-8.75, -1.5) {};
		\node [style=Z dot] (26) at (-10.75, -0.5) {};
		\node [style=Z dot] (27) at (-12.25, -0.5) {};
		\node [style=X dot] (30) at (-10, -1.5) {};
		\node [style=X dot] (31) at (-9.5, -1.5) {};
		\node [style=none] (32) at (-13, 1.5) {};
		\node [style=none] (33) at (-13, 0.5) {};
		\node [style=none] (34) at (-13, -1.5) {};
		\node [style=none] (35) at (-13, -2.5) {};
		\node [style=H box] (36) at (-7.5, 2.5) {};
		\node [style=H box] (37) at (-7.5, 1.5) {};
		\node [style=H box] (38) at (-7.5, -0.5) {};
		\node [style=H box] (39) at (-7.5, -1.5) {};
		\node [style=Z dot] (45) at (-4.75, 1.5) {};
		\node [style=Z dot] (47) at (-6.75, 2.5) {};
		\node [style=Z dot] (48) at (-4.75, 2.5) {};
		\node [style=Z dot] (49) at (-6.75, 0.5) {};
		\node [style=X dot] (50) at (-7.5, 0.5) {};
		\node [style=X dot] (53) at (-6, 2.5) {};
		\node [style=X dot] (54) at (-5.5, 2.5) {};
		\node [style=Z dot] (60) at (-4.75, -1.5) {};
		\node [style=Z dot] (62) at (-6.75, -0.5) {};
		\node [style=Z dot] (63) at (-4.75, -0.5) {};
		\node [style=Z dot] (64) at (-6.75, -2.5) {};
		\node [style=X dot] (65) at (-7.5, -2.5) {};
		\node [style=X dot] (68) at (-6, -0.5) {};
		\node [style=X dot] (69) at (-5.5, -0.5) {};
		\node [style=H box] (70) at (-3.5, 2.5) {};
		\node [style=H box] (71) at (-3.5, 0.5) {};
		\node [style=H box] (72) at (-3.5, -2.5) {};
		\node [style=H box] (73) at (-3.5, -0.5) {};
		\node [style=X dot] (84) at (-4, 1.5) {};
		\node [style=X dot] (85) at (-4, -1.5) {};
		\node [style=X dot] (86) at (-3.5, -1.5) {};
		\node [style=X dot] (87) at (-3.5, 1.5) {};
		\node [style=Z dot] (88) at (-8, 1) {};
		\node [style=Z dot] (89) at (-7.5, 1) {};
		\node [style=Z dot] (90) at (-8, -2) {};
		\node [style=Z dot] (91) at (-7.5, -2) {};
		\node [style=Z dot] (92) at (-4, -1) {};
		\node [style=Z dot] (93) at (-3.5, -1) {};
		\node [style=Z dot] (94) at (-4, 2) {};
		\node [style=Z dot] (95) at (-3.5, 2) {};
		\node [style=Z dot] (96) at (-0.75, 2.5) {};
		\node [style=X dot] (97) at (0, 2.5) {};
		\node [style=Z dot] (98) at (-2.75, 0.5) {};
		\node [style=Z dot] (99) at (-0.75, 0.5) {};
		\node [style=Z dot] (100) at (-2.75, 1.5) {};
		\node [style=X dot] (102) at (-2, 0.5) {};
		\node [style=X dot] (103) at (-1.5, 0.5) {};
		\node [style=Z dot] (104) at (-0.75, -0.5) {};
		\node [style=X dot] (105) at (0, -0.5) {};
		\node [style=Z dot] (106) at (-2.75, -2.5) {};
		\node [style=Z dot] (107) at (-0.75, -2.5) {};
		\node [style=Z dot] (108) at (-2.75, -1.5) {};
		\node [style=X dot] (110) at (-2, -2.5) {};
		\node [style=X dot] (111) at (-1.5, -2.5) {};
		\node [style=H box] (116) at (0.5, 1.5) {};
		\node [style=H box] (117) at (0.5, 0.5) {};
		\node [style=H box] (118) at (0.5, -1.5) {};
		\node [style=H box] (119) at (0.5, -2.5) {};
		\node [style=X dot] (125) at (0.5, 2.5) {};
		\node [style=X dot] (133) at (0.5, -0.5) {};
		\node [style=Z dot] (154) at (0, 3) {};
		\node [style=Z dot] (155) at (0.5, 3) {};
		\node [style=Z dot] (156) at (0, 0) {};
		\node [style=Z dot] (157) at (0.5, 0) {};
		\node [style=none] (216) at (1.5, 1.5) {};
		\node [style=none] (217) at (1.5, 0.5) {};
		\node [style=none] (218) at (1.5, -1.5) {};
		\node [style=none] (219) at (1.5, -2.5) {};
		\node [style=none] (224) at (1.5, 2.5) {};
		\node [style=none] (225) at (1.5, -0.5) {};
		\node [style=none] (234) at (1.5, -3.25) {};
		\node [style=none] (235) at (1.5, 3.25) {};
		\node [style=none] (236) at (1.575, -3.25) {};
		\node [style=none] (237) at (1.575, 3.25) {};
		\node [style=none] (238) at (1.375, 0.125) {.};
		\node [style=none] (239) at (1.375, -0.125) {.};
		\node [style=none] (240) at (-11.5, -3.75) {$\infty$};
		\node [style=none] (241) at (-11.425, -3.25) {};
		\node [style=none] (242) at (-11.425, 3.25) {};
		\node [style=none] (243) at (-11.5, -3.25) {};
		\node [style=none] (244) at (-11.5, 3.25) {};
		\node [style=none] (245) at (-11.3, 0.125) {.};
		\node [style=none] (246) at (-11.3, -0.125) {.};
	\end{pgfonlayer}
	\begin{pgfonlayer}{edgelayer}
		\draw (11) to (10);
		\draw (10) to (8);
		\draw (15) to (9);
		\draw (8) to (14);
		\draw (14) to (15);
		\draw (9) to (6);
		\draw (6) to (7);
		\draw (27) to (26);
		\draw (26) to (24);
		\draw (31) to (25);
		\draw (24) to (30);
		\draw (30) to (31);
		\draw (25) to (22);
		\draw (22) to (23);
		\draw (14) to (30);
		\draw (31) to (15);
		\draw (32.center) to (8);
		\draw (33.center) to (6);
		\draw (34.center) to (24);
		\draw (35.center) to (22);
		\draw (10) to (36);
		\draw (9) to (37);
		\draw (26) to (38);
		\draw (25) to (39);
		\draw (50) to (49);
		\draw (49) to (47);
		\draw (54) to (48);
		\draw (47) to (53);
		\draw (53) to (54);
		\draw (48) to (45);
		\draw (65) to (64);
		\draw (64) to (62);
		\draw (69) to (63);
		\draw (62) to (68);
		\draw (68) to (69);
		\draw (63) to (60);
		\draw (53) to (68);
		\draw (69) to (54);
		\draw (36) to (47);
		\draw (37) to (45);
		\draw (38) to (62);
		\draw (39) to (60);
		\draw (64) to (72);
		\draw (63) to (73);
		\draw (48) to (70);
		\draw (49) to (71);
		\draw (84) to (87);
		\draw (85) to (86);
		\draw (23) to (65);
		\draw (7) to (50);
		\draw (95) to (87);
		\draw (94) to (84);
		\draw (92) to (85);
		\draw (93) to (86);
		\draw (90) to (23);
		\draw (91) to (65);
		\draw (88) to (7);
		\draw (89) to (50);
		\draw (100) to (98);
		\draw (103) to (99);
		\draw (98) to (102);
		\draw (102) to (103);
		\draw (99) to (96);
		\draw (96) to (97);
		\draw (108) to (106);
		\draw (111) to (107);
		\draw (106) to (110);
		\draw (110) to (111);
		\draw (107) to (104);
		\draw (104) to (105);
		\draw (102) to (110);
		\draw (111) to (103);
		\draw (100) to (116);
		\draw (99) to (117);
		\draw (108) to (118);
		\draw (107) to (119);
		\draw (105) to (133);
		\draw (97) to (125);
		\draw (156) to (105);
		\draw (157) to (133);
		\draw (154) to (97);
		\draw (155) to (125);
		\draw (235.center) to (234.center);
		\draw [style=thick] (237.center) to (236.center);
		\draw (242.center) to (241.center);
		\draw [style=thick] (244.center) to (243.center);
		\draw (125) to (224.center);
		\draw (116) to (216.center);
		\draw (117) to (217.center);
		\draw (133) to (225.center);
		\draw (118) to (218.center);
		\draw (119) to (219.center);
		\draw (70) to (96);
		\draw (100) to (87);
		\draw (84) to (45);
		\draw (71) to (98);
		\draw (104) to (73);
		\draw (86) to (108);
		\draw (106) to (72);
		\draw (85) to (60);
	\end{pgfonlayer}
\end{tikzpicture}

    }
    \end{center}
\end{proposition}
\begin{proof}
    We have:
    \begin{align*}
        &\scalebox{.55}{} \\
        &\overset{\eqref{infty-rewrite-unroll}}{=}
        \scalebox{.55}{\begin{tikzpicture}
	\begin{pgfonlayer}{nodelayer}
		\node [style=Z dot] (6) at (-12.5, 0.5) {};
		\node [style=X dot] (7) at (-11.5, 0.5) {};
		\node [style=Z dot] (8) at (-15, 1.5) {};
		\node [style=Z dot] (9) at (-12.5, 1.5) {};
		\node [style=Z dot] (10) at (-15, 2.5) {};
		\node [style=Z dot] (11) at (-17, 2.5) {};
		\node [style=X dot] (14) at (-14, 1.5) {};
		\node [style=X dot] (15) at (-13.5, 1.5) {};
		\node [style=none] (19) at (-18, -0.5) {};
		\node [style=Z dot] (22) at (-12.5, -2.5) {};
		\node [style=X dot] (23) at (-11.5, -2.5) {};
		\node [style=Z dot] (24) at (-15, -1.5) {};
		\node [style=Z dot] (25) at (-12.5, -1.5) {};
		\node [style=Z dot] (26) at (-15, -0.5) {};
		\node [style=Z dot] (27) at (-17, -0.5) {};
		\node [style=X dot] (30) at (-14, -1.5) {};
		\node [style=X dot] (31) at (-13.5, -1.5) {};
		\node [style=none] (32) at (-18, 1.5) {};
		\node [style=none] (33) at (-18, 0.5) {};
		\node [style=none] (34) at (-18, -1.5) {};
		\node [style=none] (35) at (-18, -2.5) {};
		\node [style=H box] (36) at (-10.5, 2.5) {};
		\node [style=H box] (37) at (-10.5, 1.5) {};
		\node [style=H box] (38) at (-10.5, -0.5) {};
		\node [style=H box] (39) at (-10.5, -1.5) {};
		\node [style=Z dot] (48) at (-7, 1.5) {};
		\node [style=none] (49) at (-5, 1.5) {};
		\node [style=Z dot] (50) at (-9.5, 2.5) {};
		\node [style=Z dot] (51) at (-7, 2.5) {};
		\node [style=Z dot] (52) at (-9.5, 0.5) {};
		\node [style=X dot] (53) at (-10.5, 0.5) {};
		\node [style=X dot] (56) at (-8.5, 2.5) {};
		\node [style=X dot] (57) at (-8, 2.5) {};
		\node [style=Z dot] (63) at (-7, -1.5) {};
		\node [style=none] (64) at (-5, -1.5) {};
		\node [style=Z dot] (65) at (-9.5, -0.5) {};
		\node [style=Z dot] (66) at (-7, -0.5) {};
		\node [style=Z dot] (67) at (-9.5, -2.5) {};
		\node [style=X dot] (68) at (-10.5, -2.5) {};
		\node [style=X dot] (71) at (-8.5, -0.5) {};
		\node [style=X dot] (72) at (-8, -0.5) {};
		\node [style=H box] (73) at (-6, 2.5) {};
		\node [style=H box] (74) at (-6, 0.5) {};
		\node [style=H box] (75) at (-6, -2.5) {};
		\node [style=H box] (76) at (-6, -0.5) {};
		\node [style=none] (77) at (-3.5, 1.5) {};
		\node [style=none] (78) at (-3.5, 0.5) {};
		\node [style=none] (79) at (-3.5, -1.5) {};
		\node [style=none] (80) at (-3.5, -2.5) {};
		\node [style=none] (81) at (-5, 2.5) {};
		\node [style=none] (82) at (-5, 0.5) {};
		\node [style=none] (83) at (-5, -0.5) {};
		\node [style=none] (84) at (-5, -2.5) {};
		\node [style=X dot] (87) at (-3.5, 2.5) {};
		\node [style=X dot] (88) at (-3.5, -0.5) {};
		\node [style=X dot] (89) at (-3, -0.5) {};
		\node [style=X dot] (90) at (-3, 2.5) {};
		\node [style=Z dot] (91) at (-11.5, 1) {};
		\node [style=Z dot] (92) at (-10.5, 1) {};
		\node [style=Z dot] (93) at (-11.5, -2) {};
		\node [style=Z dot] (94) at (-10.5, -2) {};
		\node [style=Z dot] (95) at (-3.5, 0) {};
		\node [style=Z dot] (96) at (-3, 0) {};
		\node [style=Z dot] (97) at (-3.5, 3) {};
		\node [style=Z dot] (98) at (-3, 3) {};
		\node [style=Z dot] (99) at (0.5, 0.5) {};
		\node [style=X dot] (100) at (1.5, 0.5) {};
		\node [style=Z dot] (101) at (-2, 1.5) {};
		\node [style=Z dot] (102) at (0.5, 1.5) {};
		\node [style=Z dot] (103) at (-2, 2.5) {};
		\node [style=X dot] (105) at (-1, 1.5) {};
		\node [style=X dot] (106) at (-0.5, 1.5) {};
		\node [style=Z dot] (107) at (0.5, -2.5) {};
		\node [style=X dot] (108) at (1.5, -2.5) {};
		\node [style=Z dot] (109) at (-2, -1.5) {};
		\node [style=Z dot] (110) at (0.5, -1.5) {};
		\node [style=Z dot] (111) at (-2, -0.5) {};
		\node [style=X dot] (113) at (-1, -1.5) {};
		\node [style=X dot] (114) at (-0.5, -1.5) {};
		\node [style=H box] (119) at (2.5, 2.5) {};
		\node [style=H box] (120) at (2.5, 1.5) {};
		\node [style=H box] (121) at (2.5, -0.5) {};
		\node [style=H box] (122) at (2.5, -1.5) {};
		\node [style=Z dot] (123) at (6, 1.5) {};
		\node [style=none] (124) at (8, 1.5) {};
		\node [style=Z dot] (125) at (3.5, 2.5) {};
		\node [style=Z dot] (126) at (6, 2.5) {};
		\node [style=Z dot] (127) at (3.5, 0.5) {};
		\node [style=X dot] (128) at (2.5, 0.5) {};
		\node [style=X dot] (129) at (4.5, 2.5) {};
		\node [style=X dot] (130) at (5, 2.5) {};
		\node [style=Z dot] (131) at (6, -1.5) {};
		\node [style=none] (132) at (8, -1.5) {};
		\node [style=Z dot] (133) at (3.5, -0.5) {};
		\node [style=Z dot] (134) at (6, -0.5) {};
		\node [style=Z dot] (135) at (3.5, -2.5) {};
		\node [style=X dot] (136) at (2.5, -2.5) {};
		\node [style=X dot] (137) at (4.5, -0.5) {};
		\node [style=X dot] (138) at (5, -0.5) {};
		\node [style=H box] (139) at (7, 2.5) {};
		\node [style=H box] (140) at (7, 0.5) {};
		\node [style=H box] (141) at (7, -2.5) {};
		\node [style=H box] (142) at (7, -0.5) {};
		\node [style=none] (143) at (9.5, 1.5) {};
		\node [style=none] (144) at (9.5, 0.5) {};
		\node [style=none] (145) at (9.5, -1.5) {};
		\node [style=none] (146) at (9.5, -2.5) {};
		\node [style=none] (147) at (8, 2.5) {};
		\node [style=none] (148) at (8, 0.5) {};
		\node [style=none] (149) at (8, -0.5) {};
		\node [style=none] (150) at (8, -2.5) {};
		\node [style=X dot] (153) at (9.5, 2.5) {};
		\node [style=X dot] (154) at (9.5, -0.5) {};
		\node [style=X dot] (155) at (10, -0.5) {};
		\node [style=X dot] (156) at (10, 2.5) {};
		\node [style=Z dot] (157) at (1.5, 1) {};
		\node [style=Z dot] (158) at (2.5, 1) {};
		\node [style=Z dot] (159) at (1.5, -2) {};
		\node [style=Z dot] (160) at (2.5, -2) {};
		\node [style=Z dot] (161) at (9.5, 0) {};
		\node [style=Z dot] (162) at (10, 0) {};
		\node [style=Z dot] (163) at (9.5, 3) {};
		\node [style=Z dot] (164) at (10, 3) {};
		\node [style=Z dot] (165) at (13.5, 0.5) {};
		\node [style=X dot] (166) at (14.5, 0.5) {};
		\node [style=Z dot] (167) at (11, 1.5) {};
		\node [style=Z dot] (168) at (13.5, 1.5) {};
		\node [style=Z dot] (169) at (11, 2.5) {};
		\node [style=X dot] (171) at (12, 1.5) {};
		\node [style=X dot] (172) at (12.5, 1.5) {};
		\node [style=Z dot] (173) at (13.5, -2.5) {};
		\node [style=X dot] (174) at (14.5, -2.5) {};
		\node [style=Z dot] (175) at (11, -1.5) {};
		\node [style=Z dot] (176) at (13.5, -1.5) {};
		\node [style=Z dot] (177) at (11, -0.5) {};
		\node [style=X dot] (179) at (12, -1.5) {};
		\node [style=X dot] (180) at (12.5, -1.5) {};
		\node [style=H box] (185) at (15.5, 2.5) {};
		\node [style=H box] (186) at (15.5, 1.5) {};
		\node [style=H box] (187) at (15.5, -0.5) {};
		\node [style=H box] (188) at (15.5, -1.5) {};
		\node [style=Z dot] (189) at (19, 1.5) {};
		\node [style=none] (190) at (21, 1.5) {};
		\node [style=Z dot] (191) at (16.5, 2.5) {};
		\node [style=Z dot] (192) at (19, 2.5) {};
		\node [style=Z dot] (193) at (16.5, 0.5) {};
		\node [style=X dot] (194) at (15.5, 0.5) {};
		\node [style=X dot] (195) at (17.5, 2.5) {};
		\node [style=X dot] (196) at (18, 2.5) {};
		\node [style=Z dot] (197) at (19, -1.5) {};
		\node [style=none] (198) at (21, -1.5) {};
		\node [style=Z dot] (199) at (16.5, -0.5) {};
		\node [style=Z dot] (200) at (19, -0.5) {};
		\node [style=Z dot] (201) at (16.5, -2.5) {};
		\node [style=X dot] (202) at (15.5, -2.5) {};
		\node [style=X dot] (203) at (17.5, -0.5) {};
		\node [style=X dot] (204) at (18, -0.5) {};
		\node [style=H box] (205) at (20, 2.5) {};
		\node [style=H box] (206) at (20, 0.5) {};
		\node [style=H box] (207) at (20, -2.5) {};
		\node [style=H box] (208) at (20, -0.5) {};
		\node [style=none] (211) at (22.5, 1.5) {};
		\node [style=none] (212) at (22.5, 0.5) {};
		\node [style=none] (213) at (22.5, -1.5) {};
		\node [style=none] (214) at (22.5, -2.5) {};
		\node [style=none] (215) at (21, 2.5) {};
		\node [style=none] (216) at (21, 0.5) {};
		\node [style=none] (217) at (21, -0.5) {};
		\node [style=none] (218) at (21, -2.5) {};
		\node [style=none] (221) at (24, 1.5) {};
		\node [style=none] (222) at (24, 0.5) {};
		\node [style=none] (223) at (24, -1.5) {};
		\node [style=none] (224) at (24, -2.5) {};
		\node [style=X dot] (225) at (22.5, 2.5) {};
		\node [style=X dot] (226) at (22.5, -0.5) {};
		\node [style=X dot] (227) at (23, -0.5) {};
		\node [style=X dot] (228) at (23, 2.5) {};
		\node [style=none] (229) at (24, 2.5) {};
		\node [style=none] (230) at (24, -0.5) {};
		\node [style=Z dot] (231) at (14.5, 1) {};
		\node [style=Z dot] (232) at (15.5, 1) {};
		\node [style=Z dot] (233) at (14.5, -2) {};
		\node [style=Z dot] (234) at (15.5, -2) {};
		\node [style=Z dot] (235) at (22.5, 0) {};
		\node [style=Z dot] (236) at (23, 0) {};
		\node [style=Z dot] (237) at (22.5, 3) {};
		\node [style=Z dot] (238) at (23, 3) {};
		\node [style=none] (239) at (24, -3.25) {};
		\node [style=none] (240) at (24, 3.25) {};
		\node [style=none] (241) at (24.075, -3.25) {};
		\node [style=none] (242) at (24.075, 3.25) {};
		\node [style=none] (243) at (23.875, 0.125) {.};
		\node [style=none] (244) at (23.875, -0.125) {.};
		\node [style=none] (245) at (-16, -3.75) {$\infty$};
		\node [style=none] (246) at (-15.925, -3.25) {};
		\node [style=none] (247) at (-15.925, 3.25) {};
		\node [style=none] (248) at (-16, -3.25) {};
		\node [style=none] (249) at (-16, 3.25) {};
		\node [style=none] (250) at (-15.8, 0.125) {.};
		\node [style=none] (251) at (-15.8, -0.125) {.};
	\end{pgfonlayer}
	\begin{pgfonlayer}{edgelayer}
		\draw (11) to (10);
		\draw (10) to (8);
		\draw (15) to (9);
		\draw (8) to (14);
		\draw (14) to (15);
		\draw (9) to (6);
		\draw (6) to (7);
		\draw (27) to (26);
		\draw (26) to (24);
		\draw (31) to (25);
		\draw (24) to (30);
		\draw (30) to (31);
		\draw (25) to (22);
		\draw (22) to (23);
		\draw (14) to (30);
		\draw (31) to (15);
		\draw (32.center) to (8);
		\draw (33.center) to (6);
		\draw (34.center) to (24);
		\draw (35.center) to (22);
		\draw (10) to (36);
		\draw (9) to (37);
		\draw (26) to (38);
		\draw (25) to (39);
		\draw (53) to (52);
		\draw (52) to (50);
		\draw (57) to (51);
		\draw (50) to (56);
		\draw (56) to (57);
		\draw (51) to (48);
		\draw (48) to (49.center);
		\draw (68) to (67);
		\draw (67) to (65);
		\draw (72) to (66);
		\draw (65) to (71);
		\draw (71) to (72);
		\draw (66) to (63);
		\draw (63) to (64.center);
		\draw (56) to (71);
		\draw (72) to (57);
		\draw (36) to (50);
		\draw (37) to (48);
		\draw (38) to (65);
		\draw (39) to (63);
		\draw (67) to (75);
		\draw (66) to (76);
		\draw (51) to (73);
		\draw (52) to (74);
		\draw [in=180, out=0] (83.center) to (80.center);
		\draw [in=-180, out=0, looseness=1.25] (84.center) to (79.center);
		\draw [in=180, out=0] (81.center) to (78.center);
		\draw [in=-180, out=0, looseness=1.25] (82.center) to (77.center);
		\draw (73) to (81.center);
		\draw (74) to (82.center);
		\draw (76) to (83.center);
		\draw (75) to (84.center);
		\draw (87) to (90);
		\draw (88) to (89);
		\draw (23) to (68);
		\draw (7) to (53);
		\draw (98) to (90);
		\draw (97) to (87);
		\draw (95) to (88);
		\draw (96) to (89);
		\draw (93) to (23);
		\draw (94) to (68);
		\draw (91) to (7);
		\draw (92) to (53);
		\draw (103) to (101);
		\draw (106) to (102);
		\draw (101) to (105);
		\draw (105) to (106);
		\draw (102) to (99);
		\draw (99) to (100);
		\draw (111) to (109);
		\draw (114) to (110);
		\draw (109) to (113);
		\draw (113) to (114);
		\draw (110) to (107);
		\draw (107) to (108);
		\draw (105) to (113);
		\draw (114) to (106);
		\draw (103) to (119);
		\draw (102) to (120);
		\draw (111) to (121);
		\draw (110) to (122);
		\draw (128) to (127);
		\draw (127) to (125);
		\draw (130) to (126);
		\draw (125) to (129);
		\draw (129) to (130);
		\draw (126) to (123);
		\draw (123) to (124.center);
		\draw (136) to (135);
		\draw (135) to (133);
		\draw (138) to (134);
		\draw (133) to (137);
		\draw (137) to (138);
		\draw (134) to (131);
		\draw (131) to (132.center);
		\draw (129) to (137);
		\draw (138) to (130);
		\draw (119) to (125);
		\draw (120) to (123);
		\draw (121) to (133);
		\draw (122) to (131);
		\draw (135) to (141);
		\draw (134) to (142);
		\draw (126) to (139);
		\draw (127) to (140);
		\draw [in=180, out=0] (149.center) to (146.center);
		\draw [in=-180, out=0, looseness=1.25] (150.center) to (145.center);
		\draw [in=180, out=0] (147.center) to (144.center);
		\draw [in=-180, out=0, looseness=1.25] (148.center) to (143.center);
		\draw (139) to (147.center);
		\draw (140) to (148.center);
		\draw (142) to (149.center);
		\draw (141) to (150.center);
		\draw (153) to (156);
		\draw (154) to (155);
		\draw (108) to (136);
		\draw (100) to (128);
		\draw (164) to (156);
		\draw (163) to (153);
		\draw (161) to (154);
		\draw (162) to (155);
		\draw (159) to (108);
		\draw (160) to (136);
		\draw (157) to (100);
		\draw (158) to (128);
		\draw (169) to (167);
		\draw (172) to (168);
		\draw (167) to (171);
		\draw (171) to (172);
		\draw (168) to (165);
		\draw (165) to (166);
		\draw (177) to (175);
		\draw (180) to (176);
		\draw (175) to (179);
		\draw (179) to (180);
		\draw (176) to (173);
		\draw (173) to (174);
		\draw (171) to (179);
		\draw (180) to (172);
		\draw (169) to (185);
		\draw (168) to (186);
		\draw (177) to (187);
		\draw (176) to (188);
		\draw (194) to (193);
		\draw (193) to (191);
		\draw (196) to (192);
		\draw (191) to (195);
		\draw (195) to (196);
		\draw (192) to (189);
		\draw (189) to (190.center);
		\draw (202) to (201);
		\draw (201) to (199);
		\draw (204) to (200);
		\draw (199) to (203);
		\draw (203) to (204);
		\draw (200) to (197);
		\draw (197) to (198.center);
		\draw (195) to (203);
		\draw (204) to (196);
		\draw (185) to (191);
		\draw (186) to (189);
		\draw (187) to (199);
		\draw (188) to (197);
		\draw (201) to (207);
		\draw (200) to (208);
		\draw (192) to (205);
		\draw (193) to (206);
		\draw [in=180, out=0] (217.center) to (214.center);
		\draw [in=-180, out=0, looseness=1.25] (218.center) to (213.center);
		\draw [in=180, out=0] (215.center) to (212.center);
		\draw [in=-180, out=0, looseness=1.25] (216.center) to (211.center);
		\draw (205) to (215.center);
		\draw (206) to (216.center);
		\draw (208) to (217.center);
		\draw (207) to (218.center);
		\draw (211.center) to (221.center);
		\draw (212.center) to (222.center);
		\draw (213.center) to (223.center);
		\draw (214.center) to (224.center);
		\draw (228) to (229.center);
		\draw (227) to (230.center);
		\draw (225) to (228);
		\draw (226) to (227);
		\draw (174) to (202);
		\draw (166) to (194);
		\draw (238) to (228);
		\draw (237) to (225);
		\draw (235) to (226);
		\draw (236) to (227);
		\draw (233) to (174);
		\draw (234) to (202);
		\draw (231) to (166);
		\draw (232) to (194);
		\draw (240.center) to (239.center);
		\draw [style=thick] (242.center) to (241.center);
		\draw (247.center) to (246.center);
		\draw [style=thick] (249.center) to (248.center);
		\draw (90) to (103);
		\draw (89) to (111);
		\draw (169) to (156);
		\draw (155) to (177);
		\draw [in=-180, out=0] (64.center) to (88);
		\draw [in=-180, out=0] (49.center) to (87);
		\draw [in=-180, out=0] (132.center) to (154);
		\draw [in=-180, out=0] (124.center) to (153);
		\draw [in=-180, out=0] (198.center) to (226);
		\draw [in=-180, out=0] (190.center) to (225);
		\draw (146.center) to (173);
		\draw (175) to (145.center);
		\draw (144.center) to (165);
		\draw (167) to (143.center);
		\draw (101) to (77.center);
		\draw (99) to (78.center);
		\draw (79.center) to (109);
		\draw (107) to (80.center);
	\end{pgfonlayer}
\end{tikzpicture}
} \\
        &\overset{(\hyperlink{eq:OCM}{OCM})}{=}
        \scalebox{.55}{\begin{tikzpicture}
	\begin{pgfonlayer}{nodelayer}
		\node [style=Z dot] (6) at (-9.25, 0.5) {};
		\node [style=X dot] (7) at (-8.5, 0.5) {};
		\node [style=Z dot] (8) at (-11.25, 1.5) {};
		\node [style=Z dot] (9) at (-9.25, 1.5) {};
		\node [style=Z dot] (10) at (-11.25, 2.5) {};
		\node [style=Z dot] (11) at (-12.75, 2.5) {};
		\node [style=X dot] (14) at (-10.5, 1.5) {};
		\node [style=X dot] (15) at (-10, 1.5) {};
		\node [style=Z dot] (22) at (-9.25, -2.5) {};
		\node [style=X dot] (23) at (-8.5, -2.5) {};
		\node [style=Z dot] (24) at (-11.25, -1.5) {};
		\node [style=Z dot] (25) at (-9.25, -1.5) {};
		\node [style=Z dot] (26) at (-11.25, -0.5) {};
		\node [style=Z dot] (27) at (-12.75, -0.5) {};
		\node [style=X dot] (30) at (-10.5, -1.5) {};
		\node [style=X dot] (31) at (-10, -1.5) {};
		\node [style=none] (32) at (-13.5, 1.5) {};
		\node [style=none] (33) at (-13.5, 0.5) {};
		\node [style=none] (34) at (-13.5, -1.5) {};
		\node [style=none] (35) at (-13.5, -2.5) {};
		\node [style=H box] (36) at (-8, 2.5) {};
		\node [style=H box] (37) at (-8, 1.5) {};
		\node [style=H box] (38) at (-8, -0.5) {};
		\node [style=H box] (39) at (-8, -1.5) {};
		\node [style=Z dot] (48) at (-5.25, 1.5) {};
		\node [style=Z dot] (50) at (-7.25, 2.5) {};
		\node [style=Z dot] (51) at (-5.25, 2.5) {};
		\node [style=Z dot] (52) at (-7.25, 0.5) {};
		\node [style=X dot] (53) at (-8, 0.5) {};
		\node [style=X dot] (56) at (-6.5, 2.5) {};
		\node [style=X dot] (57) at (-6, 2.5) {};
		\node [style=Z dot] (63) at (-5.25, -1.5) {};
		\node [style=Z dot] (65) at (-7.25, -0.5) {};
		\node [style=Z dot] (66) at (-5.25, -0.5) {};
		\node [style=Z dot] (67) at (-7.25, -2.5) {};
		\node [style=X dot] (68) at (-8, -2.5) {};
		\node [style=X dot] (71) at (-6.5, -0.5) {};
		\node [style=X dot] (72) at (-6, -0.5) {};
		\node [style=H box] (73) at (-4, 2.5) {};
		\node [style=H box] (74) at (-4, 0.5) {};
		\node [style=H box] (75) at (-4, -2.5) {};
		\node [style=H box] (76) at (-4, -0.5) {};
		\node [style=X dot] (87) at (-4.5, 1.5) {};
		\node [style=X dot] (88) at (-4.5, -1.5) {};
		\node [style=X dot] (89) at (-4, -1.5) {};
		\node [style=X dot] (90) at (-4, 1.5) {};
		\node [style=Z dot] (91) at (-8.5, 1) {};
		\node [style=Z dot] (92) at (-8, 1) {};
		\node [style=Z dot] (93) at (-8.5, -2) {};
		\node [style=Z dot] (94) at (-8, -2) {};
		\node [style=Z dot] (95) at (-4.5, -1) {};
		\node [style=Z dot] (96) at (-4, -1) {};
		\node [style=Z dot] (97) at (-4.5, 2) {};
		\node [style=Z dot] (98) at (-4, 2) {};
		\node [style=Z dot] (99) at (-1.25, 2.5) {};
		\node [style=X dot] (100) at (-0.5, 2.5) {};
		\node [style=Z dot] (101) at (-3.25, 0.5) {};
		\node [style=Z dot] (102) at (-1.25, 0.5) {};
		\node [style=Z dot] (103) at (-3.25, 1.5) {};
		\node [style=X dot] (105) at (-2.5, 0.5) {};
		\node [style=X dot] (106) at (-2, 0.5) {};
		\node [style=Z dot] (107) at (-1.25, -0.5) {};
		\node [style=X dot] (108) at (-0.5, -0.5) {};
		\node [style=Z dot] (109) at (-3.25, -2.5) {};
		\node [style=Z dot] (110) at (-1.25, -2.5) {};
		\node [style=Z dot] (111) at (-3.25, -1.5) {};
		\node [style=X dot] (113) at (-2.5, -2.5) {};
		\node [style=X dot] (114) at (-2, -2.5) {};
		\node [style=H box] (119) at (0, 1.5) {};
		\node [style=H box] (120) at (0, 0.5) {};
		\node [style=H box] (121) at (0, -1.5) {};
		\node [style=H box] (122) at (0, -2.5) {};
		\node [style=Z dot] (123) at (2.75, 0.5) {};
		\node [style=Z dot] (125) at (0.75, 1.5) {};
		\node [style=Z dot] (126) at (2.75, 1.5) {};
		\node [style=Z dot] (127) at (0.75, 2.5) {};
		\node [style=X dot] (128) at (0, 2.5) {};
		\node [style=X dot] (129) at (1.5, 1.5) {};
		\node [style=X dot] (130) at (2, 1.5) {};
		\node [style=Z dot] (131) at (2.75, -2.5) {};
		\node [style=Z dot] (133) at (0.75, -1.5) {};
		\node [style=Z dot] (134) at (2.75, -1.5) {};
		\node [style=Z dot] (135) at (0.75, -0.5) {};
		\node [style=X dot] (136) at (0, -0.5) {};
		\node [style=X dot] (137) at (1.5, -1.5) {};
		\node [style=X dot] (138) at (2, -1.5) {};
		\node [style=H box] (141) at (4, -0.5) {};
		\node [style=X dot] (153) at (3.5, 0.5) {};
		\node [style=X dot] (154) at (3.5, -2.5) {};
		\node [style=X dot] (155) at (4, -2.5) {};
		\node [style=X dot] (156) at (4, 0.5) {};
		\node [style=Z dot] (157) at (-0.5, 3) {};
		\node [style=Z dot] (158) at (0, 3) {};
		\node [style=Z dot] (159) at (-0.5, 0) {};
		\node [style=Z dot] (160) at (0, 0) {};
		\node [style=Z dot] (161) at (3.5, -2) {};
		\node [style=Z dot] (162) at (4, -2) {};
		\node [style=Z dot] (163) at (3.5, 1) {};
		\node [style=Z dot] (164) at (4, 1) {};
		\node [style=Z dot] (165) at (6.75, 1.5) {};
		\node [style=X dot] (166) at (7.5, 1.5) {};
		\node [style=Z dot] (167) at (4.75, 2.5) {};
		\node [style=Z dot] (168) at (6.75, 2.5) {};
		\node [style=Z dot] (169) at (4.75, 0.5) {};
		\node [style=X dot] (171) at (5.5, 2.5) {};
		\node [style=X dot] (172) at (6, 2.5) {};
		\node [style=Z dot] (173) at (6.75, -1.5) {};
		\node [style=X dot] (174) at (7.5, -1.5) {};
		\node [style=Z dot] (175) at (4.75, -0.5) {};
		\node [style=Z dot] (176) at (6.75, -0.5) {};
		\node [style=Z dot] (177) at (4.75, -2.5) {};
		\node [style=X dot] (179) at (5.5, -0.5) {};
		\node [style=X dot] (180) at (6, -0.5) {};
		\node [style=H box] (185) at (8, 0.5) {};
		\node [style=H box] (186) at (8, 2.5) {};
		\node [style=H box] (187) at (8, -2.5) {};
		\node [style=H box] (188) at (8, -0.5) {};
		\node [style=Z dot] (189) at (10.75, 2.5) {};
		\node [style=Z dot] (191) at (8.75, 0.5) {};
		\node [style=Z dot] (192) at (10.75, 0.5) {};
		\node [style=Z dot] (193) at (8.75, 1.5) {};
		\node [style=X dot] (194) at (8, 1.5) {};
		\node [style=X dot] (195) at (9.5, 0.5) {};
		\node [style=X dot] (196) at (10, 0.5) {};
		\node [style=Z dot] (197) at (10.75, -0.5) {};
		\node [style=Z dot] (199) at (8.75, -2.5) {};
		\node [style=Z dot] (200) at (10.75, -2.5) {};
		\node [style=Z dot] (201) at (8.75, -1.5) {};
		\node [style=X dot] (202) at (8, -1.5) {};
		\node [style=X dot] (203) at (9.5, -2.5) {};
		\node [style=X dot] (204) at (10, -2.5) {};
		\node [style=H box] (205) at (12, 0.5) {};
		\node [style=H box] (206) at (12, 1.5) {};
		\node [style=H box] (207) at (12, -1.5) {};
		\node [style=H box] (208) at (12, -2.5) {};
		\node [style=none] (221) at (12.75, 1.5) {};
		\node [style=none] (222) at (12.75, 0.5) {};
		\node [style=none] (223) at (12.75, -1.5) {};
		\node [style=none] (224) at (12.75, -2.5) {};
		\node [style=X dot] (225) at (11.5, 2.5) {};
		\node [style=X dot] (226) at (11.5, -0.5) {};
		\node [style=X dot] (227) at (12, -0.5) {};
		\node [style=X dot] (228) at (12, 2.5) {};
		\node [style=none] (229) at (12.75, 2.5) {};
		\node [style=none] (230) at (12.75, -0.5) {};
		\node [style=Z dot] (231) at (7.5, 2) {};
		\node [style=Z dot] (232) at (8, 2) {};
		\node [style=Z dot] (233) at (7.5, -1) {};
		\node [style=Z dot] (234) at (8, -1) {};
		\node [style=Z dot] (235) at (11.5, 0) {};
		\node [style=Z dot] (236) at (12, 0) {};
		\node [style=Z dot] (237) at (11.5, 3) {};
		\node [style=Z dot] (238) at (12, 3) {};
		\node [style=none] (239) at (12.75, -3.25) {};
		\node [style=none] (240) at (12.75, 3.25) {};
		\node [style=none] (241) at (12.825, -3.25) {};
		\node [style=none] (242) at (12.825, 3.25) {};
		\node [style=none] (243) at (12.625, 0.125) {.};
		\node [style=none] (244) at (12.625, -0.125) {.};
		\node [style=none] (245) at (-12, -3.75) {$\infty$};
		\node [style=none] (246) at (-11.925, -3.25) {};
		\node [style=none] (247) at (-11.925, 3.25) {};
		\node [style=none] (248) at (-12, -3.25) {};
		\node [style=none] (249) at (-12, 3.25) {};
		\node [style=none] (250) at (-11.8, 0.125) {.};
		\node [style=none] (251) at (-11.8, -0.125) {.};
	\end{pgfonlayer}
	\begin{pgfonlayer}{edgelayer}
		\draw (11) to (10);
		\draw (10) to (8);
		\draw (15) to (9);
		\draw (8) to (14);
		\draw (14) to (15);
		\draw (9) to (6);
		\draw (6) to (7);
		\draw (27) to (26);
		\draw (26) to (24);
		\draw (31) to (25);
		\draw (24) to (30);
		\draw (30) to (31);
		\draw (25) to (22);
		\draw (22) to (23);
		\draw (14) to (30);
		\draw (31) to (15);
		\draw (32.center) to (8);
		\draw (33.center) to (6);
		\draw (34.center) to (24);
		\draw (35.center) to (22);
		\draw (10) to (36);
		\draw (9) to (37);
		\draw (26) to (38);
		\draw (25) to (39);
		\draw (53) to (52);
		\draw (52) to (50);
		\draw (57) to (51);
		\draw (50) to (56);
		\draw (56) to (57);
		\draw (51) to (48);
		\draw (68) to (67);
		\draw (67) to (65);
		\draw (72) to (66);
		\draw (65) to (71);
		\draw (71) to (72);
		\draw (66) to (63);
		\draw (56) to (71);
		\draw (72) to (57);
		\draw (36) to (50);
		\draw (37) to (48);
		\draw (38) to (65);
		\draw (39) to (63);
		\draw (67) to (75);
		\draw (66) to (76);
		\draw (51) to (73);
		\draw (52) to (74);
		\draw (87) to (90);
		\draw (88) to (89);
		\draw (23) to (68);
		\draw (7) to (53);
		\draw (98) to (90);
		\draw (97) to (87);
		\draw (95) to (88);
		\draw (96) to (89);
		\draw (93) to (23);
		\draw (94) to (68);
		\draw (91) to (7);
		\draw (92) to (53);
		\draw (103) to (101);
		\draw (106) to (102);
		\draw (101) to (105);
		\draw (105) to (106);
		\draw (102) to (99);
		\draw (99) to (100);
		\draw (111) to (109);
		\draw (114) to (110);
		\draw (109) to (113);
		\draw (113) to (114);
		\draw (110) to (107);
		\draw (107) to (108);
		\draw (105) to (113);
		\draw (114) to (106);
		\draw (103) to (119);
		\draw (102) to (120);
		\draw (111) to (121);
		\draw (110) to (122);
		\draw (128) to (127);
		\draw (127) to (125);
		\draw (130) to (126);
		\draw (125) to (129);
		\draw (129) to (130);
		\draw (126) to (123);
		\draw (136) to (135);
		\draw (135) to (133);
		\draw (138) to (134);
		\draw (133) to (137);
		\draw (137) to (138);
		\draw (134) to (131);
		\draw (129) to (137);
		\draw (138) to (130);
		\draw (119) to (125);
		\draw (120) to (123);
		\draw (121) to (133);
		\draw (122) to (131);
		\draw (135) to (141);
		\draw (153) to (156);
		\draw (154) to (155);
		\draw (108) to (136);
		\draw (100) to (128);
		\draw (164) to (156);
		\draw (163) to (153);
		\draw (161) to (154);
		\draw (162) to (155);
		\draw (159) to (108);
		\draw (160) to (136);
		\draw (157) to (100);
		\draw (158) to (128);
		\draw (169) to (167);
		\draw (172) to (168);
		\draw (167) to (171);
		\draw (171) to (172);
		\draw (168) to (165);
		\draw (165) to (166);
		\draw (177) to (175);
		\draw (180) to (176);
		\draw (175) to (179);
		\draw (179) to (180);
		\draw (176) to (173);
		\draw (173) to (174);
		\draw (171) to (179);
		\draw (180) to (172);
		\draw (169) to (185);
		\draw (168) to (186);
		\draw (177) to (187);
		\draw (176) to (188);
		\draw (194) to (193);
		\draw (193) to (191);
		\draw (196) to (192);
		\draw (191) to (195);
		\draw (195) to (196);
		\draw (192) to (189);
		\draw (202) to (201);
		\draw (201) to (199);
		\draw (204) to (200);
		\draw (199) to (203);
		\draw (203) to (204);
		\draw (200) to (197);
		\draw (195) to (203);
		\draw (204) to (196);
		\draw (185) to (191);
		\draw (186) to (189);
		\draw (187) to (199);
		\draw (188) to (197);
		\draw (201) to (207);
		\draw (200) to (208);
		\draw (192) to (205);
		\draw (193) to (206);
		\draw (228) to (229.center);
		\draw (227) to (230.center);
		\draw (225) to (228);
		\draw (226) to (227);
		\draw (174) to (202);
		\draw (166) to (194);
		\draw (238) to (228);
		\draw (237) to (225);
		\draw (235) to (226);
		\draw (236) to (227);
		\draw (233) to (174);
		\draw (234) to (202);
		\draw (231) to (166);
		\draw (232) to (194);
		\draw (240.center) to (239.center);
		\draw [style=thick] (242.center) to (241.center);
		\draw (247.center) to (246.center);
		\draw [style=thick] (249.center) to (248.center);
		\draw (123) to (153);
		\draw (131) to (154);
		\draw (87) to (48);
		\draw (101) to (74);
		\draw (73) to (99);
		\draw (88) to (63);
		\draw (76) to (107);
		\draw (111) to (89);
		\draw (109) to (75);
		\draw (90) to (103);
		\draw (167) to (127);
		\draw (165) to (126);
		\draw (169) to (156);
		\draw (175) to (141);
		\draw (173) to (134);
		\draw (177) to (155);
		\draw (189) to (225);
		\draw (206) to (221.center);
		\draw (222.center) to (205);
		\draw (226) to (197);
		\draw (207) to (223.center);
		\draw (224.center) to (208);
	\end{pgfonlayer}
\end{tikzpicture}
}
        \overset{\eqref{infty-rewrite-unroll}}{=}
        \scalebox{.55}{} \\
    \end{align*}
\end{proof}
\section{Relationship to previous Floquetification procedures}
\label{appendix:floquetification-teague}
The work in this paper builds upon the Floquetification procedure proposed by~\textcite{townsend-teagueFloquetifyingColourCode2023}.
We can gain a new perspective on their proposed Floquetification of the $\interp{4, 2, 2}$ code using our distance-preserving rewrites and our local perspective of individually Floquetified measurements composed together.
Reframing their work in our notation gives a different, less topological perspective and showcases the differences.
Their procedure results in a $\interp{12, 2, 2}$ code, whereas our method results in a $\interp{6, 2, 2}$ code.

\subsection{Step 1 - Writing out the measurement circuit}
Similar to our procedure, \textcite{townsend-teagueFloquetifyingColourCode2023} start by writing out the measurement schedule of the $\interp{4, 2, 2}$ code as a ZX diagram. 
However, instead of writing the $XXXX$ measurement as a conjugated $ZZZZ$ measurement, they make use of the \TextColour-rule to remove the Hadamard boxes: 
\[\input{floquetification-2/422-operation-circuit.tikz} \quad = \quad \input{floquetification-teague/teague-1.tikz}\]

\subsection{Step 2.1 - Rewriting measurements}
They expand the measurements:
\setlength{\jot}{18pt}
\begin{align*}
  \scalebox{.9}{\input{floquetification-teague/teague-1.tikz}} \ \ \ =\ \ &\scalebox{.9}{\input{floquetification-teague/teague-2.tikz}}\\
  \overset{\text\tiny(\hyperlink{eq:OCM}{OCM})}{=}\ \ &\scalebox{.9}{\input{floquetification-teague/teague-3.tikz}}
\end{align*}
This procedure expands the four-legged spiders into a single weight-two measurement instead of two weight-two measurements. 
While this is not distance-preserving, it happens to be distance-preserving in their particular setting. 
This is due to the fact that they assume that measurements are error-free. 
As we had to repeat the weight-two measurement to catch measurement errors, in their case, a single weight-two measurement does not decrease the distance. 
This relationship between distance-preservation and the assumed noise model is further explored and formalised in \textcite{rodatzFaultTolerance2025}.

\subsection{Step 2.2 - Composing measurements}
They then merge the destructive measurements and state preparations. 
For this, they also pull the state preparations through the repetition using $\eqref{infty-rewrite-reorder}$.
However, in contrast to our procedure, they apply $\eqref{infty-rewrite-reorder}$ to all measurements instead of just the last ones. 
Thus, while we only pull through two measurements, they pull through four, and they have to do this in both bases.
This leads to them having a qubit overhead of eight qubits instead of the two qubit overhead incurred by our procedure:
\begin{align*}
  &\scalebox{.9}{\input{floquetification-teague/teague-3.tikz}}
  \overset{\text\tiny\eqref{infty-rewrite-reorder}}{=}
  \scalebox{.9}{\input{floquetification-teague/teague-4.tikz}} \\
  &\overset{\text\tiny\eqref{one-pauli}}{=}
  \scalebox{.9}{\input{floquetification-teague/teague-5.tikz}}
  \overset{\text\tiny\eqref{one-pauli}}{=}
  \scalebox{.9}{\input{floquetification-teague/teague-6.tikz}}
\end{align*}
Then, they introduce the single-qubit measurements:
\begin{align*}
    &\scalebox{.9}{\input{floquetification-teague/teague-7.tikz}} = \scalebox{.9}{\input{floquetification-teague/teague-8.tikz}}
\end{align*}

Similar to above, they do not repeat the measurement twice, which is ok in their setting where they assume measurements to be error-free.

An interesting difference that emerges between the two procedures is that by initially removing the Hadamards using \TextColour, all Pauli operators they obtain are CSS, i.e.\@ only $X$-type or only $Z$-type Paulis.
We could adapt our method in a similar manner but that would also increase our qubit overhead from two to four, as we would now have to pull through all the destructive measurements instead of just the measurements of one type. 

To remove the swaps, they unroll and reorder:
\begin{align*}
    &\scalebox{.9}{\input{floquetification-teague/teague-8.tikz}} \\
    &\overset{\text\tiny\eqref{infty-rewrite-unroll}}{=} \scalebox{.9}{\input{floquetification-teague/teague-9.tikz}} \\
    &\overset{\text\tiny(\hyperlink{eq:OCM}{OCM})}{=} \scalebox{.8}{\input{floquetification-teague/teague-10.tikz}}
\end{align*}

\textcite{townsend-teagueFloquetifyingColourCode2023} perform one more step that does not feature in our procedure: they observe that they can reorder the measurements as many of them are non-overlapping:

\begin{align*}
    &\scalebox{.8}{\input{floquetification-teague/teague-10.tikz}} \\
    &\overset{\text\tiny(\hyperlink{eq:OCM}{OCM})}{=} \scalebox{.8}{\input{floquetification-teague/teague-11.tikz}}
\end{align*}
\noindent
This reordering of measurements corresponds to the shift in spacetime or reordering of world-lines, discussed in their paper.
This is enabled by the fact that they have more overhead, leading to these non-overlapping measurements. 
In our Floquetification of the $\interp{4, 2, 2}$ code, this is not possible.

\subsection{Step 3 - Extracting the new code}
As \textcite{townsend-teagueFloquetifyingColourCode2023} removed the Hadamards before their Floquetification procedure, they do not have any single-qubit Cliffords left in the circuit and can immediately read off the measurement schedule of the corresponding Floquet code.

\subsection{Floquetifying the colour code}
\textcite{townsend-teagueFloquetifyingColourCode2023} also Floquetify the bulk of the hexagonal colour code.
However, their expansion of the weight-six measurement is not distance-preserving, even if you assume measurements to be error-free as it consists of the following step:
\[\input{floquetification-teague/teague-weight-six.tikz}\]
Here, a single $Z$ error on the vertical wires will spread out to become an error of weight two or even three.
Therefore, we hypothesise that the resulting Floquet code will likely have a lower distance. 
\section{Fault equivalence vs distance preservation}
\label{appendix:fe-equals-dist-pres}
In a follow-up work to this paper \parencite{rodatzFaultTolerance2025}, we introduce fault-equivalent rewrites as:
\begin{definition}
  \label{def:fault-equvialent}
 A semantic preserving rewrite $r : D_1 \to D_2$ is distance-preserving if, for any error $E_2$ in $D_2$, either
    \begin{itemize}
        \item $E_2$ is detectable in $D_2$, or
        \item there exists an error $E_1$ in $D_1$ such that $|E_1| \leq |E_2|$ and $\interp{D_1^{E_1}} = \interp{D_2^{E_2}}$.
    \end{itemize}
 and similarly for any error $E_1$ in $D_1$ the same condition holds.
\end{definition}

These rewrites are exactly the ones that satisfy the conditions for \autoref{prop:dist-preservation}. 
In \autoref{prop:dist-preservation}, we show that all fault-equivalent rewrites are distance-preserving. 
Here, we will additionally show that all distance-preserving rewrites are fault-equivalent, i.e.\@ that the two concepts are equivalent.

For this proof, we have to introduce one more idea from \parencite{rodatzFaultTolerance2025}, namely, idealised edges.
Idealised edges are edges of a ZX diagram that are idealised as being error-free. 
Thus, valid errors on that diagram may not include edge-flips that involve idealised edges. 
All other definitions, such as ZX distance, remain the same; however, now with some edges being idealised. 
We draw idealised edges in purple. 
We have: 
\begin{theorem}
 All distance-preserving rewrites are fault-equivalent. 
\end{theorem}
\begin{proof}
 We prove the claim by contrapositive.  
 Assume that a rewrite \( r: D_1 \to D_2 \) is \emph{not} fault-equivalent.  
 We will show that \( r \) cannot be distance-preserving.  
 To do so, we construct a context \( D \) such that the composed diagrams \( D(D_1) \) and \( D(D_2) \) differ in distance.

 Since \( r \) is not fault-equivalent, there exists a non-trivial, undetectable error on one of the diagrams that has no corresponding error of equal or smaller weight on the other.  
 Without loss of generality, let \( E_1 \) be the smallest, undetectable error on \( D_1 \) such that there is no \( E_2 \) on \( D_2 \) satisfying \( |E_2| \le |E_1| \) and \( \interp{D_2^{E_2}} = \interp{D_1^{E_1}} \).

 We now construct a context \( D \) such that the only non-trivial, undetectable errors on \( D(D_1) \) are those equivalent to \( E_1 \).  
 Using the Choi-Jamiołkowski isomorphism, let \( \psi \) denote the state obtained from \( D_1 \) by bending its input wires, and let \( \psi_{E_1} \) be the corresponding state from \( D_1^{E_1} \):
    \[
 \input{psi-definition.tikz}
    \]
 Let \( S \) be any quantum error-correcting code encoding one logical qubit with stabiliser generators \( S_1, \dots, S_m \).  

 Next, we construct a Clifford map \( C \) that maps \( \psi \) to \( \overline{\ket{0}} \) and \( \psi_{E_1} \) to \( \overline{\ket{1}} \).  
 Observe that errors do not change the stabilisers of the state, only whether the state lies in the \( +1 \) or \( -1 \) eigenspace of those stabilisers \parencite{ruschCompletenessFault2025a}.
 Therefore, \( \psi \) and \( \psi_{E_1} \) have the same stabilisers, except that they might live in different eigenspaces of those stabilisers. 
 We now construct a basis $P_1, \dots, P_n$, such that \( \psi \) and \( \psi_{E_1} \) all live in the same eigenspaces of $P_1, \dots, P_{n-1}$ and in a different eigenspace of $P_n$. 
 We can do this by taking any other basis $P'_1, \dots, P'_n$. 
 There must exist at least some stabiliser that the two states do not agree on. 
 Without loss of generality, we will assume that this is $P'_n$. 
 Then we can take $P_n = P'_n$. 
 For all other $i$, we take $P_i = P'_i$ if \( \psi \) and \( \psi_{E_1} \) live in the same eigenspace of $P'_i$ and $P_i = P'_iP_n$ otherwise. 
 Clearly, this forms a basis. 

 Now we can construct the desired Clifford map. 
 We chose \( S_1, \dots, S_m \) to be the stabilisers of the quantum error correction code $S$. 
 As $S$ has one logical qubit, we know that $m = n - 1$. 
 Let $S_X$ be the stabiliser of $\overline{\ket{0}}$.
 Then, we can construct $C$ such that it maps $P_i$ to $S_i$ for all $i \in [1, n-1]$ and $P_n$ to $S_X$.
 Clearly, this map maps \( \psi \) to \( \overline{\ket{0}} \) and \( \psi_{E_1} \) to \( \overline{\ket{1}} \).
 
 Then consider the following diagram:
    \[
 \input{dist-pres-context.tikz}
    \]
 We now consider $D$ to be the diagram in the dotted box, displaying $D(D_1)$ above.
 We depict \( C \) and the measurements of \( S_1, \dots, S_m \) in purple to indicate that they are idealised as error-free.

 In the error-free case, this diagram is non-zero, corresponding to $\overline{\ket{0}}$.
 If $E_1$ occurs, the diagram corresponds to $\overline{\ket{1}}$; thus, an undetectable, non-trivial error has occurred.

 We now show that all other non-trivial errors on \( D_1 \) are detectable.  
 Errors do not alter the stabilisers of the state, but only whether the state lies in the \( +1 \) or \( -1 \) eigenspace of those stabilisers \parencite{ruschCompletenessFault2025a}.
 Let some error $E'$ flip different stabilisers than $E_1$. 
 Then it must flip at least $P_i$ for $i \not= n$.
 But then $C\ket{\psi_{E'}}$ lives in the \(-1\) eigenspace of $S_i$. 
 Therefore, when measuring $S_i$, $E'$ is detected. 
 Thus, all non-trivial errors that are not equivalent to $E_1$ must be detectable. 
 Hence, the only non-trivial, undetectable errors on \( D(D_1) \) are those equivalent to \( E_1 \), and the distance of \( D(D_1) \) is \( wt(E_1) \) as we assumed $E_1$ to be the smallest error in its equivalence class.

 But then, if we now place $D_2$ into the same context, then, by assumption, either no equivalent error exists or all equivalent errors have a higher weight than $E_1$. 
 In the first case, $D(D_2)$ has distance infinity, and in the latter case, it has distance larger than $E_2$. 
 Either way, replacing $D_1$ by $D_2$ changes the distance of the overall diagram and therefore, non-fault-equivalent rewrites cannot be distance-preserving.
\end{proof}

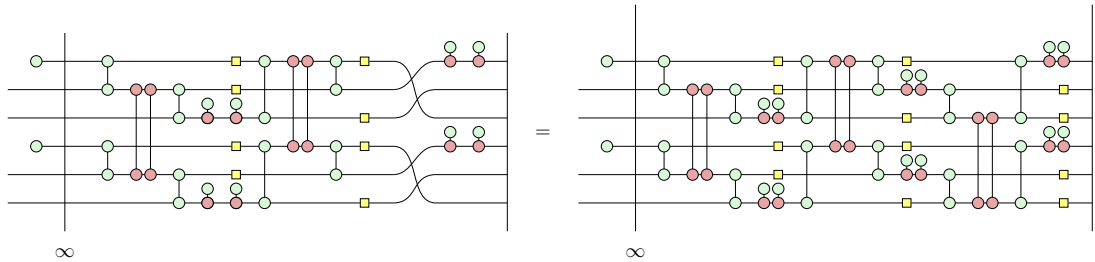
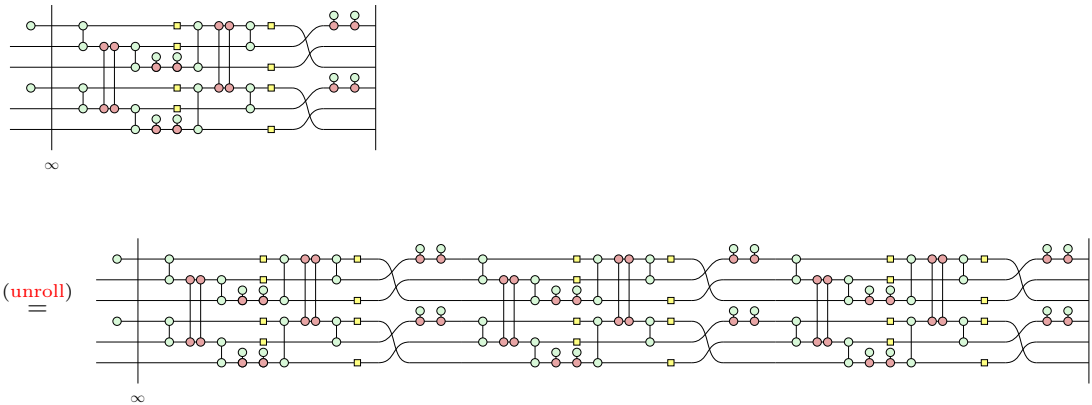
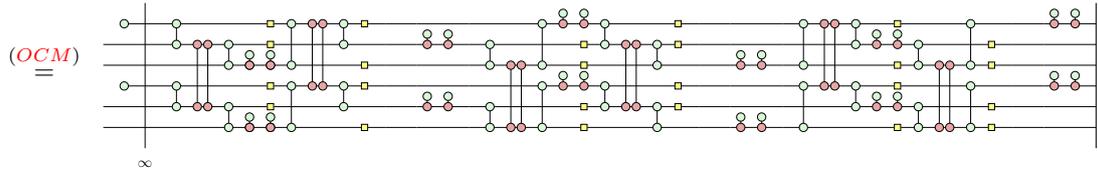
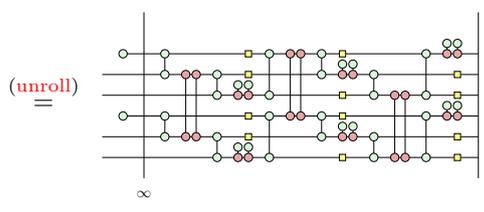
\end{document}